\documentclass[final,12pt]{elsarticle}

\usepackage[margin=1in]{geometry}

\usepackage{graphicx}

\usepackage{amssymb}

\usepackage{url}

\usepackage{amsmath}
\usepackage{upgreek}

\usepackage{amsthm}

\def\doubleunderline#1{\underline{\underline{#1}}}
\usepackage{tensor}
\usepackage{dsfont}

\usepackage{siunitx}

\usepackage{amssymb}

\usepackage{bbm}

\usepackage{soul}

\usepackage{algorithmicx}
\usepackage{algorithm}
\usepackage[noend]{algpseudocode}
\makeatletter
\renewcommand{\ALG@beginalgorithmic}{\footnotesize}
\makeatother

\usepackage[dvipsnames,table,xcdraw]{xcolor}
\newcommand{\highlightgreen}[1]{%
\colorbox{LimeGreen!50}{$\displaystyle#1$}}

\usepackage{booktabs}
\usepackage{amsmath}

\makeatletter
\newcommand{\vast}{\bBigg@{4}}
\newcommand{\Vast}{\bBigg@{7}}
\makeatother

\usepackage{caption}

\usepackage{subcaption}

\setcitestyle{authoryear}
 \biboptions{comma,round}

\nonumnote{\textit{DOI: https://doi.org/10.1016/j.jcp.2020.109606 \enspace}}

\journal{Journal of Computational Physics}
\usepackage{etoolbox}
\makeatletter
\patchcmd{\ps@pprintTitle}% <cmd>
  {Preprint submitted to}% <search>
  {Published in}% <replace>
  {}{}% <succes><failure>
\makeatother

\begin{document}

\begin{frontmatter}

%% Title, authors and addresses
%% Should be lower-case except for first word first letter and proper nouns and abbreviations
%\title{A conservative diffuse-interface method for the simulation of compressible two-phase flows with turbulence and acoustics}
%\title{A conservative diffuse-interface method for the simulation of compressible two-phase flows}
\title{A conservative diffuse-interface method for compressible\\ two-phase flows}
%% SHORT TITLE: 
%% Give a short (running) title that will appear in the heading of the odd pages
%%\shorttitle{Conservative diffuse-interface method for compressible flows}

%% use the tnoteref command within \title for footnotes;
%% use the tnotetext command for the associated footnote;
%% use the fnref command within \author or \address for footnotes;
%% use the fntext command for the associated footnote;
%% use the corref command within \author for corresponding author footnotes;
%% use the cortext command for the associated footnote;
%% use the ead command for the email address,
%% and the form \ead[url] for the home page:
%%
%% \title{Title\tnoteref{label1}}
%% \tnotetext[label1]{}
%% \author{Name\corref{cor1}\fnref{label2}}
%% \ead{email address}
%% \ead[url]{home page}
%% \fntext[label2]{}
%% \cortext[cor1]{}
%% \address{Address\fnref{label3}}
%% \fntext[label3]{}

%% use optional labels to link authors explicitly to addresses:
%% \author[label1,label2]{<author name>}
%% \address[label1]{<address>}
%% \address[label2]{<address>}
%% AUTHORS: Should be separated by comma, the last author should be separated by \and
%% Also, no full names, just the initials
%\author{S. S. Jain, A. Mani \and P. Moin}
\author{Suhas S. Jain\corref{cor1}}
\ead{sjsuresh@stanford.edu}
\cortext[cor1]{Corresponding author}
%\ead[http://web.stanford.edu/~sjsuresh]{home page}
\author{Ali Mani}
\ead{alimani@stanford.edu}
\author{Parviz Moin}
\ead{moin@stanford.edu}
\address{Center for Turbulence Research, Stanford University, California, USA 94305}
%% SHORT AUTHOR LIST: If there are more than three, use 
%% the first author's name and "et.~al"
%\shortauthor{Jain et~al.}
%\shortauthor{Jain, Mani \& Moin.}
%USE
%\shortauthor{Junipero-Serra \& Anza}
%IF THERE WERE ONLY TWO AUTHORS
%%\address{California, United States}

\begin{abstract}
%% Text of abstract
In this article, we propose a novel conservative diffuse-interface method for the simulation of immiscible compressible two-phase flows. The proposed method discretely conserves the mass of each phase, momentum and total energy of the system. We use the baseline five-equation model and propose interface-regularization (diffusion\textendash sharpening) terms in such a way that the resulting model maintains the conservative property of the underlying baseline model; and lets us use a central-difference scheme for the discretization of all the operators in the model, which leads to a non-dissipative implementation that is crucial for the simulation of turbulent flows and acoustics. Furthermore, the provable strengths of the proposed model are: (a) the model maintains the boundedness property of the volume fraction field, which is a physical realizability requirement for the simulation of two-phase flows, (b) the proposed model is such that the transport of volume fraction field inherently satisfies the total-variation-diminishing property without having to add any flux limiters that destroy the non-dissipative nature of the scheme, (c) the proposed interface-regularization terms in the model do not spuriously contribute to the kinetic energy of the system and therefore do not affect the non-linear stability of the numerical simulation, and (d) the model is consistent with the second law of thermodynamics. Finally, we present numerical simulations using the model and assess (a) the accuracy of evolution of the interface shape, (b) implementation of surface tension effects, (c) propagation of acoustics and their interaction with material interfaces, (d) the accuracy and robustness of the numerical scheme for simulation of complex high-Reynolds-number flows, and (e) performance and scalability of the method. 
\end{abstract}

\begin{keyword}
phase-field method \sep compressible flows \sep two-phase flows \sep conservative schemes \sep non-dissipative schemes \sep acoustics
%% keywords here, in the form: keyword \sep keyword

%% MSC codes here, in the form: \MSC code \sep code
%% or \MSC[2008] code \sep code (2000 is the default)

\end{keyword}

\end{frontmatter}

%%
%% Start line numbering here if you want
%%
%\linenumbers

%% main text

\section{Introduction}

%\lfoot{Published in Journal of Computational Physics. DOI: https://doi.org/10.1016/j.jcp.2020.109606}

%\textbf{Add first sentence on bubble acoustics application of compressible two-phase flows.}

Compressible two-phase flows are ubiquitous in nature and are of engineering interest. The applications of compressible two-phase flows span a wide range of areas including bubble acoustics, liquid fuel interaction with gaseous oxidizer in high-pressure environments including supercritical flow regimes, shock-bubble interactions, cavitation, etc. The focus of the article is to present a novel diffuse-interface method for the simulation of compressible two-phase flows, where the interface between two compressible fluid media is resolved on an Eulerian grid. The application of this method for the study of bubble acoustics will be briefly discussed below.   

One of the primary applications of two-phase flows with compressible phases is the study of underwater bubble acoustics. Prediction of bubble dynamics in turbulent seawater is of practical importance for the engineering analysis of naval systems. In ships, the air bubbles entrained by boundary layers and stern waves form an elongated wake that lasts for several kilometers downstream \citep{trevorrow1994acoustical,fu2007,stanic2009attenuation}. Though the bubbles are tiny, with diameters of order $1\ \mathrm{mm}$ or less, they exhibit strong acoustic responses. Hence, the bubbly wake can be detected acoustically, which reveals the presence and position of the ship. The predictive modeling of bubble distributions in wakes, along with their acoustic response, has remained elusive and mostly confined to Reynolds-averaged Navier-Stokes (RANS) analyses because of the multiscale nature of the problem and the computational challenges associated with scalability and performance \citep{carrica1999polydisperse,culver2007measuring}. Hence, the current study is focused on developing a conservative numerical method that enables accurate treatment of the interaction of acoustics with gas\textendash liquid interfaces (single and multiple bubbles) in compressible turbulent flow environments. This aids in investigating the current limitations and in developing subgrid-scale models based on the Rayleigh-Plesset or Keller-Miksis equations used in RANS and large-eddy simulations (LES). %\st{The present paper deals with the development and verification of the numerical method. The application of this to the study of bubble acoustics in a realistic setting will be deferred to a future work}.

%%%%%%%%%%%%%%%%%%%%%%%%%%%%%%%%%%

%\begin{figure}
%\vskip 0.1in
%\centering
%\includegraphics[width=0.9\textwidth]{motivation.pdf}
%\caption{Schematics of the generation of a bubbly wake from ship stern and bubble-hydroacoustic interactions.}
%\label{fig:motivation}
%\end{figure}

In compressible flows, thermodynamics plays an important role and adds more difficulty to an already complex problem of two-phase flows, by imposing an additional requirement that the model should maintain thermodynamic consistency at the interface. Moreover, numerical study of turbulent flows and acoustics requires stable, non-dissipative, and conservative numerical methods. The state-of-the-art techniques to simulate compressible two-phase flows lack many of these features. With this motivation, we have developed a diffuse-interface five-equation model for the simulation of two immiscible compressible fluids that has all the above favorable properties.

Compressible two-phase flows have been extensively studied for the last two decades \citep{saurel2018diffuse}, predominantly using diffuse-interface methods. Consider a close-up view of the molecular picture of the interface between two immiscible phases, schematically shown in Figure \ref{fig:molecule} (a), where the denser fluid is shown in green and the lighter fluid is shown in red. If the phases are volume averaged, we obtain the volumetric representation (a continuum picture) of the interface between two fluids that is also shown in Figure \ref{fig:molecule} (a). Typically, the thickness of these interfaces is on the order of  few nanometers. Therefore, for problems that are of engineering interest, the interface between two fluids can be regarded as mostly sharp because of the inherent scale separation between the interface thickness and the characteristic scales of the flow prevalent in the problem. However, a diffuse-interface method is a computational model where the physical sharp interface is artificially made thick\textemdash on the order of grid-cell size\textemdash so that the gradients at the material interface can be resolved on an Eulerian grid, as illustrated in Figure \ref{fig:molecule} (b), similar to the thickened-flame model for premixed combustion \citep{poinsot2005theoretical}. This has huge implications on the choice of numerical methods used to represent the interface, numerical stability and accuracy of the numerical simulation and hence the diffuse-interface methods have been the focus of study for over two decades. On the other hand, the interface represented using a diffuse-interface method can also be interpreted as a time-averaged interface in a turbulent flow. However, this would require additional modeling to account for the subgrid contribution, similar to the widely-used two-fluid modeling approaches \citep{ishii1984two}. 

%and can be represented using a smooth volume fraction field, provided that the Knudsen number of the interface is small ($Kn_I\ll1$) where $Kn_I=\lambda/\delta$, $\lambda$ is the mean free path and $\delta$ is the interface thickness. 

\begin{figure}
    \centering
    \includegraphics[width=0.65\textwidth]{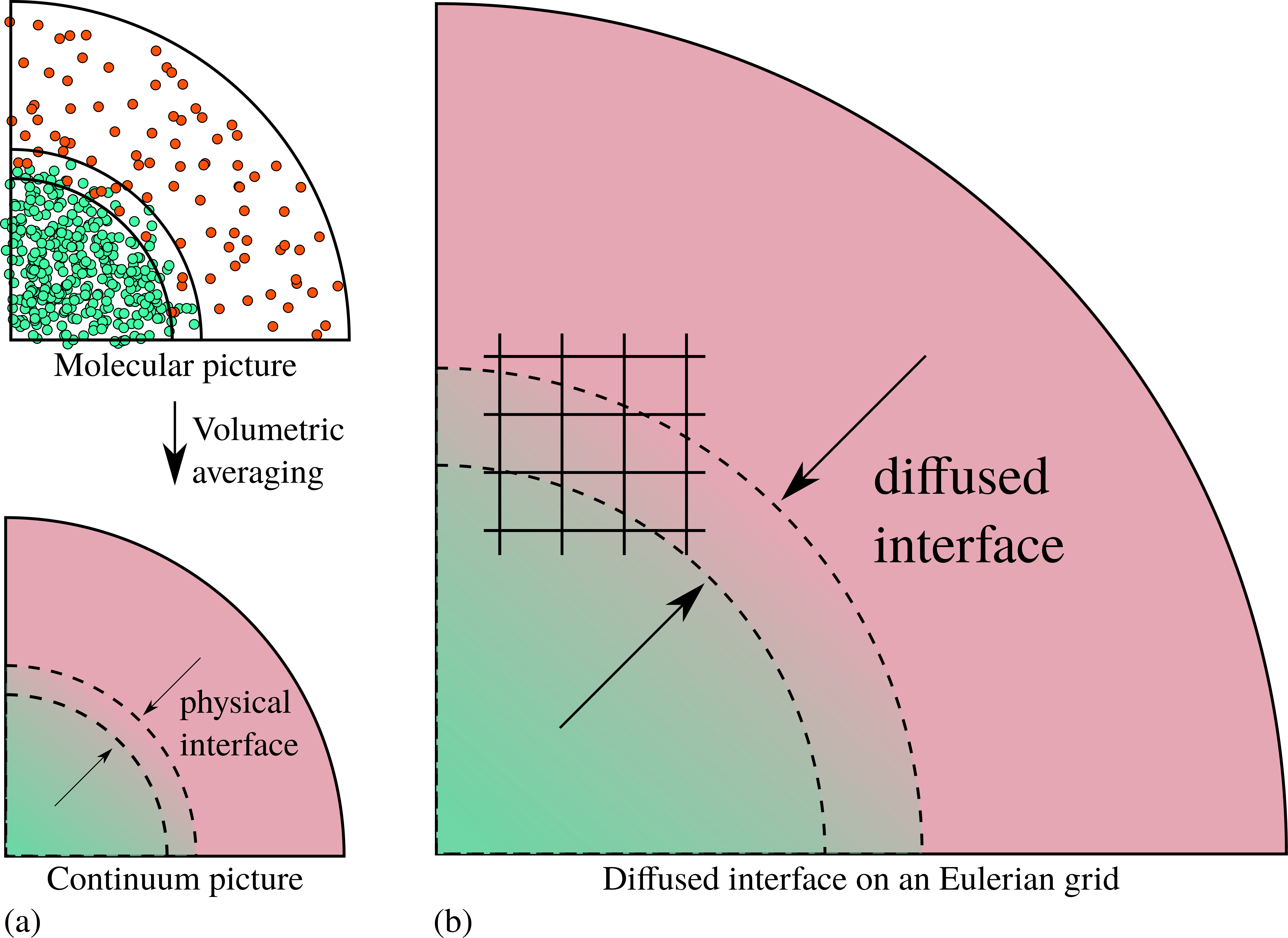}
    \caption{(a) A close-up view of the molecular and continuum representation of the interface between two fluids. The molecules are represented by different colors for two fluids. (b) A diffused interface on an Eulerian grid. The grid is represented by square boxes.}
    \label{fig:molecule}
\end{figure}

Different diffuse-interface models for the simulation of compressible two-phase flows present in the literature can be broadly classified into four major types: (a) The five-equation model \citep{kapila2001two} solves two mass balance equations\textemdash one for each of the phases\textemdash a momentum equation, a total energy equation, and a volume fraction advection equation. This is the model that is most suitable for the simulation of two-phase non-reacting flows with immiscible fluids. More on this model will be described in Section \ref{sec:model}. (b) The six-equation model \citep{yeom2013modified} is similar to the five-equation model but solves two energy equations, one for each of the phases. (c) The seven-equation model \citep{baer1986two} solves two momentum equations and two energy equations and has two separate velocity fields for each of the phases. This is the most general of all the models since it includes non-equilibrium effects such as phase change and mass transfer. (d) The four-equation model \citep{abgrall1996prevent} has no separate mass balance equations for each of the phases; instead, it solves a continuity equation, hence conserving only the total mass and not the individual mass of each phase. The volume fraction advection equation has also been replaced by a transport equation for the polytropic coefficient in this model.

The seven-equation model was first proposed by \citet{baer1986two} for the simulation of detonation-to-deflagration transition in reactive granular materials and was later used by \citet{sainsaulieu1995finite} to simulate two-phase flows using an approximate Roe-type Riemann solver. \citet{abgrall1996prevent} proposed the first four-equation model to simulate two ideal gases using Roe's Riemann solver and derived an interface-equilibrium condition (IEC) to eliminate the spurious pressure oscillations that were observed at the interface. More on the IEC can be found in Sections \ref{sec:model} and \ref{sec:iec}. Further, \citet{saurel1999simple} extended this four-equation model with IEC for the stiffened-gas equation of state (EOS) and also derived the IEC for seven-equation model \citep{saurel1999multiphase}. A more recent use of the four-equation model can be seen in \citet{johnsen2012preventing} and \citet{movahed2013solution}, where a weighted essentially non-oscillatory (WENO) scheme was used at the interfaces and shocks, and a high-order central-difference scheme was used away from these; and was used to simulate Richtmyer$\textit{-}$Meshkov instability.

The five-equation model was first proposed by \citet{kapila2001two} for the simulation of detonation-to-deflagration transition in granular materials and was later used by \citet{allaire2002five} to simulate two-phase flows. They also proposed the isobaric closure law that mimics the IEC for five-equation model and showed that the model can simulate two fluids with arbitrary EOSs. Further, \citet{perigaud2005compressible} extended this model to include capillary and viscous effects. More recently, \citet{shukla2010interface} and \citet{tiwari2013diffuse} proposed interface-regularization (diffusion\textendash sharpening) terms that keep the interface sharp for long-time integrations, thus increasing the accuracy of the simulation. The five-equation model has also been implemented on unstructured grids by \citet{chiapolino2017sharpening}. For various other extensions and modifications to the five-equation model, refer to the articles by \citet{so2012anti,ansari2013numerical,shukla2014nonlinear,coralic2014finite,wong2017high} and \citet{garrick2017finite}; therefore, the five-equation model is by far the most widely used model of all the diffuse-interface models for the simulation of compressible two-phase flows.

%\textbf{Sharp interface methods and why not to use them:}
Attempts to simulate compressible two-phase flows have also been made using sharp-interface methods; see \citet{jemison2014compressible} for the moment-of-fluid approach, \citet{kannan2018cell} for a geometric volume-of-fluid approach, \citet{huber2015time,bai2017sharp} and \citet{fu2017single} for a level-set method, and \citet{he2017characteristic} for an algebraic volume-of-fluid approach. Although, sharp-interface methods are more accurate in representing the shape of the interface than diffuse-interface methods, they are also more expensive. Moreover, the expensive function evaluation of the sharp-interface methods is localized at the interface, which results in load-balancing and parallel scalability issues. When it comes to compressible flows, diffuse-interface methods have an obvious advantage over sharp-interface methods. The volume of each phase is inherently not conserved in compressible flows; hence, the expensive interface reconstruction and the geometric advection step in sharp-interface methods to achieve discrete volume conservation are less useful. Moreover, to achieve mass conservation of each phase in a sharp-interface volume-tracking method, a cell integrated semi-Lagrangian geometric method needs to be used for advection such that the mass flux is consistent with the volume flux calculated from the piecewise-linear reconstructed interface. Whereas, depending on the choice of the model, a mass balance equation in each phase can be solved in a diffuse-interface method to discretely conserve the mass of each phase. For these reasons, in the current study, we choose to use a diffuse-interface method over a sharp-interface method. For a more detailed comparison between sharp-interface and diffuse-interface methods, see \citet{Mirjalili2017}.

In summary, a five-equation model appears to be a suitable choice of diffuse-interface model for the simulation of compressible two-phase flows with immiscible fluids. Some of the limitations in the current state-of-the-art methods are as follows: (a) The study of acoustics and turbulent flows requires non-dissipative methods, but to the best of our knowledge, there is no previous implementation of compressible two-phase flows that is fully non-dissipative. (b) All the interface-regularization (diffusion\textendash sharpening) terms used along with the five-equation model are in non-conservative form \citep{shukla2010interface,tiwari2013diffuse, garrick2017finite}, and the conservative form of the regularization terms is currently considered to be unstable. %(c) None of the previous studies explore compressible two-phase flows in a turbulent environment. 

In this paper, we present a novel diffuse-interface model that (a) can be solved using non-dissipative numerical methods (central-difference schemes) that are crucial for the simulation of turbulence and acoustics, (b) discretely conserves mass of each phase, total momentum, and total energy in the system, (c) maintains mechanical equilibrium and thermodynamic equilibrium across the interface (no spurious behavior in velocity and pressure fields), (d) maintains a steady interface thickness throughout the simulation, (e) maintains boundedness of the volume fraction field, which is a physical realizability requirement for the simulation of two-phase flows, and (f) maintains total-variation-diminishing (TVD) property of the volume fraction field without having to add any flux limiters that destroy the non-dissipative nature of the underlying central-difference scheme. In this paper, we present the model for shock-free compressible regions, but shocks in a high-Mach number regime can potentially be handled with the implementation of the localized artificial bulk viscosity approach \citep{mani2009suitability,kawai2010assessment}.

This paper is divided into 12 sections. Section \ref{sec:model} presents the diffuse-interface formalism and the proposed new model; Section \ref{sec:volume} presents the derivation of the volume-fraction equation and the proof of boundedness and TVD properties; Sections \ref{sec:mass}, \ref{sec:momentum} and \ref{sec:energy} present the derivation of mass, momentum and energy equations, respectively; Section \ref{sec:iec} presents the proof of the IEC condition; Section \ref{sec:finalmodel} presents the final model in its full form including the surface tension and gravity terms; and Sections \ref{sec:implement} and \ref{sec:results} present the numerical implementation and simulation results, respectively, followed by the summary of results and findings in Section \ref{sec:summary}, and concluding remarks in Section \ref{sec:conclude}.

\section{Governing equations and diffuse-interface formalism \label{sec:model}}

We start with the well-known inviscid five-equation model of \citet{allaire2002five}. This form of the model has a volume fraction advection equation [Eq. (\ref{eq:volumefraction})], a mass balance equation for each of the phases $l$ [Eq. (\ref{eq:mass})], a momentum equation [Eq. (\ref{eq:mom})], and a total energy equation [Eq. (\ref{eq:energy})]. 

\begin{equation}
\frac{\partial \phi_1}{\partial t} + \vec{u}\cdot\vec{\nabla}\phi_1 = 0,
\label{eq:volumefraction}
\end{equation}
\begin{equation}
\centering
\frac{\partial \rho_l \phi_l}{\partial t} + \vec{\nabla}\cdot(\rho_l\vec{u}\phi_l) = 0, \hspace{0.5cm} l=1,2,
\label{eq:mass}
\end{equation}
%\begin{equation}
%\frac{\partial \rho_2 \phi_2}{\partial t} + \vec{\nabla}\cdot(\rho_2\vec{u}\phi_2) = 0,
%\label{eq:mass2}
%\end{equation}
\begin{equation}
\frac{\partial \rho\vec{u}}{\partial t} + \vec{\nabla}\cdot(\rho \vec{u} \otimes \vec{u} 
+ p \mathds{1}) = 0,
\label{eq:mom}
\end{equation}
and
\begin{equation}
\frac{\partial \rho (e+k)}{\partial t} + \vec{\nabla}\cdot(\rho H \vec{u}) = 0,
\label{eq:energy}
\end{equation}
where $\phi_l$ is the volume fraction of phase $l$ that satisfies the condition $\sum_{l=1}^2 \phi_l=1$, $\rho_l$ is the density of phase $l$, $\rho$ is the total density defined as $\rho=\sum_{l=1}^2\rho_l\phi_l$, $\vec{u}$ is the velocity, $p$ is the pressure, $e$ is the specific mixture internal energy, which can be related to the specific internal energy of phase $l$, $e_l$, as $e=\sum_{l=1}^2 \rho_le_l$, $k=\frac{1}{2}u_iu_i$ is the specific kinetic energy, and $H=e+k+p/\rho$ is the specific total enthalpy of the mixture. 

\citet{allaire2002five} showed that when this system is solved along with an isobaric closure law at the interface, one can achieve mechanical and thermodynamic equilibrium (Postulate \ref{post:iec}) at the interface that results in stable numerical solutions and eliminates spurious oscillations at the interface. 

\newtheorem{post}{Postulate}[section]

\begin{post}
If $u^k_i=u_0$ and $p^k_i = p_0$ across the interface, where $k$ is the time-step index and $i$ is the grid index, any model or a numerical scheme that satisfies $u^{k+1}_i=u_0$ and $p^{k+1}_i=p_0$, $\forall i$, is said to satisfy the interface-equilibrium condition (IEC) \citep{abgrall1996prevent}.
\label{post:iec}
\end{post}

It is generally known that, in a classical diffuse-interface method, the interface thickness increases with simulation time due to the use of dissipative numerical schemes that are adopted to stabilize the method, reducing the overall accuracy of the solution for long-time integrations. Hence, \citet{shukla2010interface} and \citet{tiwari2013diffuse} proposed interface-regularization (diffusion\textendash sharpening) terms to counter this thickening of the interface. However, their regularization terms are in non-conservative form, and they argued that their conservative form of the interface-regularization terms results in tangential fluxes, which leads to unphysical interface deformations. 

In the current work, we propose a new set of interface-regularization (diffusion\textendash sharpening) terms that are in conservative form and show that the numerical solution is stable for long-time integrations. We propose a model of the form given in Eqs. (\ref{eq:mod_volumefraction})\textendash(\ref{eq:mod_energy}) along with the viscous terms, where the highlighted terms are the newly introduced interface-regularization terms. Equation (\ref{eq:mod_volumefraction}) represents the modified volume fraction advection equation, Eq. (\ref{eq:mod_mass}) represents the modified mass balance equation for phase $l$, Eq. (\ref{eq:mod_mom}) represents the modified momentum equation, and Eq. (\ref{eq:mod_energy}) represents the modified total energy equation. If a general equation of state (EOS) for phase $l$ is written as $p_l=\alpha_l\rho_le_l + \beta_l$, where $\alpha_l$ and $\beta_l$ are constants, then by invoking the isobaric closure law for pressure in the mixture region ($p = p_1 = p_2$), the generalized mixture EOS can be written as in Eq. (\ref{eq:pressure}).  

\begin{equation}
\frac{\partial \phi_1}{\partial t} + \vec{\nabla}\cdot(\vec{u}\phi_1) = \phi_1(\vec{\nabla}\cdot\vec{u})+\highlightgreen{\vec{\nabla}\cdot\vec{a_1}},
\label{eq:mod_volumefraction}
\end{equation}
\begin{equation}
\frac{\partial m_l}{\partial t} + \vec{\nabla}\cdot(\vec{u} m_l) = \highlightgreen{\vec{\nabla}\cdot \vec{R}_l}, \hspace{0.5cm} l=1,2,
\label{eq:mod_mass}
\end{equation}
\begin{equation}
\frac{\partial \rho\vec{u}}{\partial t} + \vec{\nabla}\cdot(\rho \vec{u} \otimes \vec{u} + p \mathds{1}) = \vec{\nabla}\cdot\doubleunderline\tau + \highlightgreen{\vec{\nabla}\cdot(\vec{f}\otimes\vec{u})},
\label{eq:mod_mom}
\end{equation}
\begin{equation}
\frac{\partial E}{\partial t} + \vec{\nabla}\cdot(\vec{u} E) + \vec{\nabla}\cdot(p\vec{u}) = \vec{\nabla}\cdot(\doubleunderline\tau\cdot\vec{u}) + \vec{\nabla}\cdot(\lambda \nabla T) + \highlightgreen{\vec{\nabla}\cdot(\vec{f}k)} + \highlightgreen{\sum_{l=1}^2 \vec{\nabla}\cdot{(\rho_l h_l \vec{a}_l )}},
\label{eq:mod_energy}
\end{equation}
and
\begin{equation}
p = \frac{\rho e + \left\{\frac{\phi_1 \beta_1}{\alpha_1} + \frac{(1 -\phi_1) \beta_2}{\alpha_2}\right\}} {\Big( \frac{\phi_1}{\alpha_1} + \frac{1- \phi_1}{\alpha_2}\Big)}.
\label{eq:pressure}
\end{equation}
In Eqs. (\ref{eq:mod_volumefraction})\textendash(\ref{eq:mod_energy}), $\vec{a}_1=\vec{a}(\phi_1)=\Gamma\{\epsilon\vec{\nabla}\phi_1 - \phi_1(1 - \phi_1)\vec{n}_1\}$ is the flux of the interface-regularization term for phase $1$, and it satisfies the condition $\vec{a}(\phi_1)=-\vec{a}(\phi_2)$; $\vec{n}_1=\vec{\nabla}\phi_1/|\vec{\nabla}\phi_1|$ is the normal of the interface for phase $1$; and $\Gamma$ and $\epsilon$ are the interface parameters, where $\Gamma$ represents an artificial regularization velocity scale and $\epsilon$ represents an interface thickness scale (see Section \ref{sec:volume} for a discussion on the choice of these parameters). $\vec{R}_l=\rho_{0l}\vec{a}_l$ is the flux of the regularization term in the mass balance equation for phase $l$, where $\rho_{0l}$ is the characteristic density representing phase $l$ (see Section \ref{sec:mass}), $\vec{f}=\sum_{l=1}^2 \vec{R}_l$ is the net mass regularization flux, $m_l=\rho_l\phi_l$ is the mass per unit total volume for phase $l$, and $\rho=\sum_{l=1}^2 m_l$ is the total density of the mixture. In Eq. (\ref{eq:mod_mass}), $m_l$ is written instead of $\rho_l\phi_l$ only to show that $m_l$ is the variable being solved and not $\rho_l$ (see Section \ref{sec:mass}). Invoking Stokes' hypothesis, the Cauchy stress tensor is written as $\doubleunderline\tau = 2\mu\mathbb{D} - 2\mu(\vec{\nabla}\cdot\vec{u})\mathds{1}/3$, where $\mu$ is the dynamic viscosity of the mixture evaluated using the one-fluid mixture rule as $\mu=\sum_{l=1}^2 \phi_l \mu_l$, $\mathbb{D}=\{(\vec{\nabla}\vec{u}) + (\vec{\nabla}\vec{u})^T\}/2$ is the strain-rate tensor, $E=\rho(e+k)$ is the total energy per unit volume, $\lambda$ is the thermal conductivity of the mixture, and $T$ is the temperature. In this study, it is assumed that $\lambda=0$, and therefore, the thermal conduction term is dropped in the rest of the paper. However, one needs to include the thermal conduction effects for simulating boiling flows \citep{saurel2018diffuse}. If each of the phases is assumed to follow a stiffened-gas EOS, then the constants in the EOS can be written as $\alpha=\gamma - 1$ and $\beta = -\gamma \pi$, where $\gamma$ is the polytropic coefficient and $\pi$ is the reference pressure. Values of $\gamma$ and $\pi$ are experimentally determined, and the values used in this work are listed in Table \ref{tab:prop}. Then, the speed of sound $c_l$ for phase $l$ can be written as
\begin{equation}
    c_l=\sqrt{\gamma_l\Big(\frac{p + \pi_l}{\rho_l}\Big)}.
\end{equation}

In Eq. (\ref{eq:mod_energy}), $h_l=e_l+p/\rho_l$ represents the specific enthalpy of the phase $l$ and can be expressed in terms of $\rho_l$ and $p$ using the stiffened-gas EOS as
\begin{equation}
    h_l=\frac{(p + \pi_l)\gamma_l}{\rho_l(\gamma_l - 1)}.
    \label{eq:enthalpy}
\end{equation}
All the newly added terms are in conservative form. Hence, the mass of each phase, momentum, and total energy are discretely conserved in the simulation irrespective of the choice of the numerical scheme. Moreover, we choose to use a second-order central-difference scheme for all the discretizations in this study since low-order central-difference schemes are known to have some advantages for the simulation of turbulent flows \citep{moin2016suitability} due to their (a) non-dissipative nature, (b) low aliasing error, (c) ease of boundary treatment, (d) low cost, and (e) improved stability. The non-dissipative nature of these schemes is also crucial for the resolved simulation of acoustics.  %However, one could also use a high-order central-difference scheme, a dispersion-relation preserving (DRP) \textbf{CITE} scheme or a compact finite-difference scheme \textbf{CITE} along with our formulation, that are known to be superior for acoustic simulations, due to their better wavenumber representation. We also present the discretized versions of the equations in our model and discuss the stability properties in Sections \ref{sec:dc}.

Further, a systematic derivation of the newly introduced regularization terms, along with the associated mathematical proofs, is described in the subsequent sections.

\section{Volume fraction advection equation \label{sec:volume}}

Denoting the volume fraction of phase $1$ as $\phi$, then the volume fraction advection equation in Eq. (\ref{eq:mod_volumefraction}) can be written as

\begin{equation}
\frac{\partial \phi}{\partial t} + \vec{\nabla}\cdot(\vec{u}\phi) = \phi(\vec{\nabla}\cdot\vec{u})+\vec{\nabla}\cdot\left[\Gamma\left\{\epsilon\vec{\nabla}\phi - \phi(1 - \phi)\vec{n}\right\}\right].
\label{eq:phi}
\end{equation}
This equation is obtained by combining Eq. (\ref{eq:volumefraction}) and the reinitialization step of the conservative level-set method by \citet{olsson2005conservative} and \citet{olsson2007conservative}, and is also an extension of the incompressible version of the conservative diffuse-interface method introduced by \citet{chiu2011conservative} and \citet{Mirjalili2020}. One can show that Eq. (\ref{eq:phi}) also governs the advection of the volume fraction for phase $2$; i.e., $\phi_2=1 - \phi$ also satisfies Eq. (\ref{eq:phi}). Hence, both phases $1$ and $2$ are consistently transported.  

\subsection{Proof of boundedness of $\phi$ \label{sec:bound}}

Since we choose to use a central-difference scheme to discretize all the equations in our model because of its well-known desirable properties, described in Section \ref{sec:model}, we could potentially encounter overshoots and undershoots in the $\phi$ field due to dispersion errors. Hence, one needs to pick the values of the free parameters $\Gamma$ and $\epsilon$ such that $\phi$ is maintained between $0$ and $1$. 

\citet{Mirjalili2020} showed that there exists a crossover line in the $\Gamma$-$\epsilon$ parameter space above which the boundedness of $\phi$ is guaranteed for an incompressible flow. We extend this analysis to show that the same criterion (Figure \ref{fig:bounded}) is sufficient to maintain the boundedness of the $\phi$ field in a compressible flow, provided that the time-step restriction given in Eq. (\ref{eq:boundtime}) below for a one-dimensional setting, and Eq. (\ref{eq:3Dboundtime}) for a three-dimensional setting, is satisfied (Theorem \ref{theorem:bound}). 

\begin{figure}
    \centering
    \includegraphics[width=0.5\textwidth]{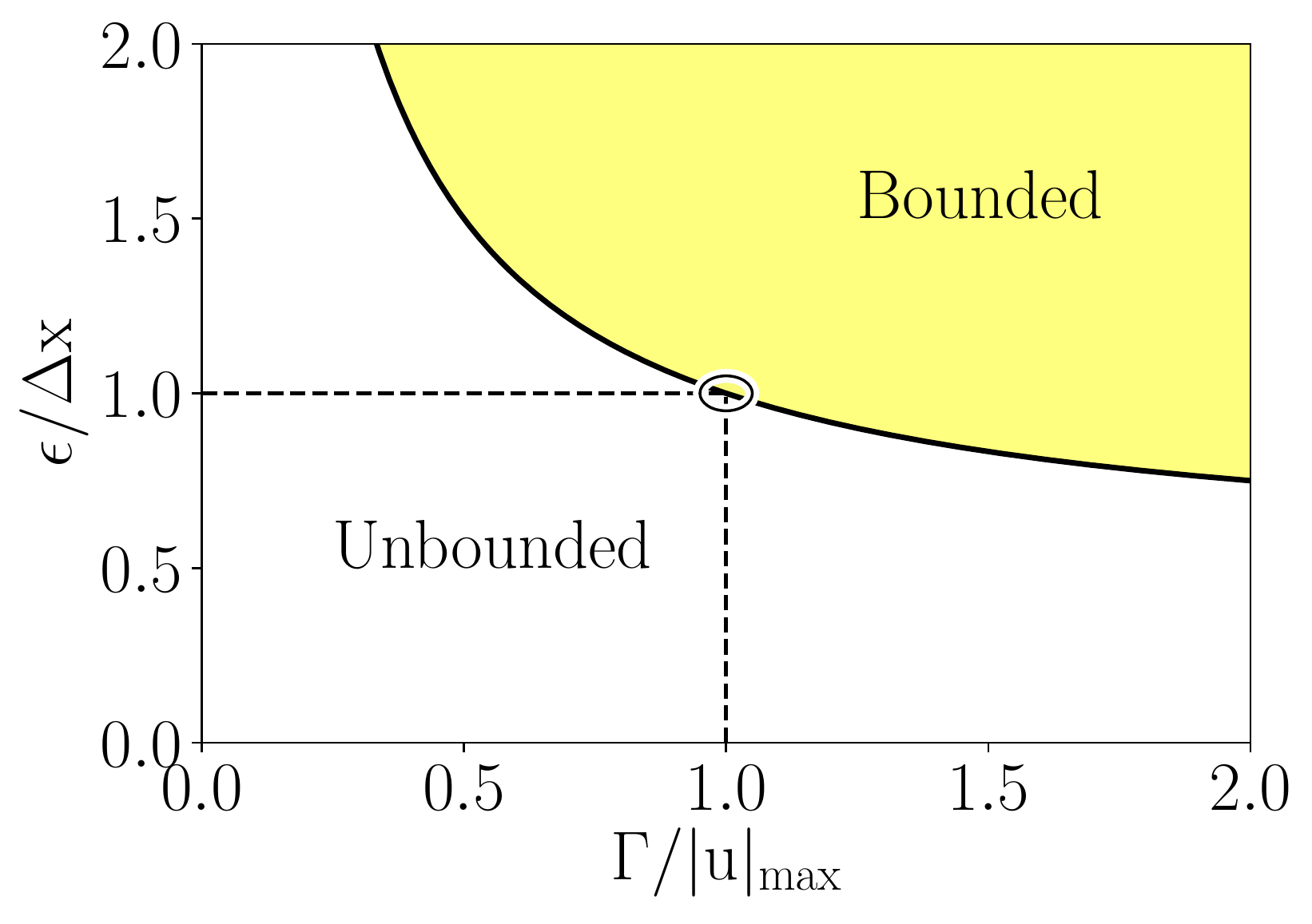}
    \caption{Region of boundedness as given by Eq. (\ref{eq:crossover}). The dashed lines represent the most optimum choice of $\Gamma$ and $\epsilon$ (see Section \ref{sec:parameters}).}
    \label{fig:bounded}
\end{figure}

\newtheorem{theorem}{Theorem}[section]

\begin{theorem}
On a uniform one-dimensional grid, if $0\le\phi^k_i\le1$ is satisfied for $k=0$, then $0\le\phi^k_i\le1$ holds $\forall k$, where $k$ is the time-step index and $i$ is the grid index, provided 
\begin{equation}
    \frac{\epsilon}{\Delta x} \ge \frac{\left(\frac{\vert u\vert_{max}}{\Gamma} + 1\right)}{2}
    \label{eq:crossover}
\end{equation}
and
\begin{equation}
      \Delta t \le \min_i\left[\frac{1}{\max\left\{{\Big(\frac{2\Gamma \epsilon}{\Delta x^2}\Big) - \Big(\frac{u_{i+1}^k - u_{i-1}^k}{\Delta x}\Big), 0}\right\}}\right],
    \label{eq:boundtime}
\end{equation}
are satisfied, where $\Delta x$ is the grid-cell size, $\Delta t$ is the time-step size, and $|u|_{max}$ is the maximum velocity in the domain.

\label{theorem:bound}
\end{theorem}

\begin{proof}
Consider the discretization of Eq. (\ref{eq:phi}) on a one-dimensional uniform grid
\begin{equation}
\begin{aligned}
\phi^{k+1}_i = \phi^k_i + \Delta t \left\{-\left(\frac{u^k_{i+1}\phi^k_{i+1} - u^k_{i-1}\phi^k_{i-1}}{2\Delta x} \right) + \phi^k_i\Big(\frac{u^k_{i+1} - u^k_{i-1}}{2\Delta x} \Big) \right\} \\
+ \Delta t \left[\Gamma\epsilon\left(\frac{\phi_{i+1}^k -2\phi_{i}^k + \phi_{i-1}^k}{\Delta x^2}\right) - \Gamma\left\{\frac{(1 - \phi_{i+1}^k) n_{i+1}^k \phi_{i+1}^k - (1 - \phi_{i-1}^k) n_{i-1}^k \phi_{i-1}^k}{2\Delta x}\right\} \right],
\end{aligned}
\end{equation}
where $k$ represents the time-step index and $i$ the grid index. This can be rearranged as
\begin{equation}
\phi_i^{k+1} = \tilde{C}_{i-1}^k \phi_{i-1}^k + \tilde{C}_{i}^k \phi_{i}^k + \tilde{C}_{i+1}^k \phi_{i+1}^k,
\end{equation}
where $\tilde{C}$'s are coefficients given by
\begin{equation}
    \tilde{C}_{i-1}^k = \frac{\Delta t u^k_{i-1}}{2\Delta x} + \frac{\Delta t \Gamma \epsilon}{\Delta x^2} + \frac{\Delta t \Gamma}{2\Delta x}(1 - \phi_{i-1}^k) n_{i-1}^k,
    \label{eq:ci-1}
\end{equation}
\begin{equation}
    \tilde{C}_{i+1}^k = -\frac{\Delta t u^k_{i+1}}{2\Delta x} + \frac{\Delta t \Gamma \epsilon}{\Delta x^2} - \frac{\Delta t \Gamma}{2\Delta x}(1 - \phi_{i+1}^k) n_{i+1}^k,
    \label{eq:ci+1}
\end{equation}
and
\begin{equation}
    \tilde{C}^k_i=1+\frac{\Delta t}{2\Delta x} (u^k_{i+1} - u^k_{i-1}) - \frac{2 \Delta t \Gamma \epsilon}{\Delta x^2}.
    \label{eq:ci}
\end{equation}
\newtheorem{lemma}{Lemma}[theorem]
\begin{lemma}
A scheme is said to maintain positivity [also called the ``boundedness" criterion in \citet{patankar1980numerical, versteeg2007introduction}] if $\tilde{C}$'s are all positive \citep{laney1998computational}.
\label{lemma:bound}
\end{lemma}

For $k=0$, it is given that $0\le\phi^k_i\le1$ holds, which implies that $(1 - \phi_{i-1}^0) n_{i-1}^0 \ge -1$.\\ Then $\tilde{C}_{i-1}^0 \ge \Delta t u^0_{i-1}/(2\Delta x) + \Delta t \Gamma \epsilon/\Delta x^2 - \Delta t \Gamma/(2\Delta x) \ge - \Delta t/2\Delta x(|u|^0_{max} + \Gamma) + \Delta t \Gamma \epsilon/\Delta x^2$. Now, invoking the condition in Eq. (\ref{eq:crossover}), we can show that $\tilde{C}_{i-1}^0 \ge 0$ holds. Using similar arguments, we can show that $\tilde{C}_{i+1}^0 \ge 0$ holds. Invoking the condition in Eq. (\ref{eq:boundtime}), we can also show that $\tilde{C}_{i}^0 \ge 0$ holds. Thus, Lemma \ref{lemma:bound} proves that $0\le\phi^1$ is satisfied. Since, $1-\phi$ also satisfies Eq. \eqref{eq:phi}, this implies that $\phi^1\le1$ is also true. Hence, $0\le\phi^1_i\le1$ is satisfied. Now, by repeating the same procedure above, we can show that $0\le\phi^{k+1}_i\le1$ is satisfied, provided that $0\le\phi^{k}_i\le1$ is satisfied. Hence, using mathematical induction, $0\le\phi^k_i\le1$ is satisfied $\forall k\in\mathds{Z}^+$, which concludes the proof.

\end{proof}

If $\phi$ is bounded, then $1-\phi$ is also bounded. Hence, the volume fractions of both phases $1$ and $2$ are bounded. Now, generalizing Theorem \ref{theorem:bound} for three dimensions, the time-step restriction required for the boundedness of $\phi$\textemdash assuming an isotropic mesh\textemdash can be written as
\begin{equation}
    \Delta t \le \min_i\left[\frac{1}{\max \left\{ \left(\frac{6\Gamma \epsilon}{\Delta x^2} \right) - \big(\frac{\delta u_i}{\delta x_i}\big), 0\right\}}\right],
    \label{eq:3Dboundtime}
\end{equation}
where $\delta/\delta x$ is the discrete derivative operator. The first term $(6\Gamma \epsilon/{\Delta x^2})$ represents the diffusive Courant--Friedrichs--Lewy (CFL) condition of the interface, with $\Gamma\epsilon$ representing the diffusivity of the interface regularization, and the second term $({\delta u_i}/{\delta x_i})$ represents the time-step constraint associated with the local dilatation of the flow. If the flow is incompressible, the time-step constraint reduces to 
\begin{equation}
    \Delta t \le \frac{\Delta x^2}{6 \Gamma \epsilon}.
    \label{eq:incompressibletime}
\end{equation}
But, if the flow is expanding, the time-step constraint is less restrictive compared to incompressible (dilatation-free) flow; and if the flow is compressing, the time-step constraint is more restrictive compared to incompressible flow. However, the time-step restriction due to the acoustic CFL condition is usually more restrictive than the condition in Eq. \eqref{eq:3Dboundtime} and hence it does not add any additional time-step restriction. In the proof of Theorem \ref{theorem:bound}, a first-order Euler time-stepping scheme was used to arrive at the restrictions on the time-step size in Eq. \eqref{eq:boundtime} and Eq. \eqref{eq:3Dboundtime}; however, these criteria are sufficient to maintain the boundedness of $\phi$ with most higher-order explicit time-stepping schemes since the diffusive CFL condition of the interface in Eq. \eqref{eq:boundtime} and Eq. \eqref{eq:3Dboundtime} are less restrictive for higher-order time-stepping schemes.

%%%%%%%%%%%%%%%%%%%
%TO DO

%For final paper, plot divergence field and also plot the dt restriction in the domain to estimate the severness of this criterion

%%%%%%%%%%%%%%%%%%%

%%%%%%%%%%%%%%%%%%%
%TO DO

%For final paper, compare Kapila's model vs Allaire's model

%%%%%%%%%%%%%%%%%%%

\subsection{Proof of total-variation-diminishing property of $\phi$\label{sec:tvd}}

The boundedness of $\phi$ is very important since it maintains $\phi$ between the physical values of $0$ and $1$ throughout the simulation. However, $\phi$ can still develop oscillations without going unbounded. But we need $\phi$ to be a smooth field that takes a value of $0$ and $1$ in the pure single-phase regions away from the interface and a smooth variation in between in the mixture regions. Hence we seek a stronger non-linear stability condition, the total-variation-diminishing (TVD) property, for $\phi$.     

The total variation of an arbitrary function $f$ is defined as the sum of: two times all the local maxima of $f$; negative two times all the local minima of $f$; and one times the boundary value of $f$, if that is a local maximum, and negative one times of it, if that is a local minimum. Similarly, a numerical approximation of the total variation of $f$ is given by
\begin{equation}
TV=\sum_{i=1}^{N} |f_{i+1} - f_{i}|,
\end{equation}
where $i$ is the grid index, and $N$ is the number of grid points. Below, we show that the criterion in Eq. (\ref{eq:tvd}), in addition to being bounded, is sufficient to maintain the TVD property of the $\phi$ field for compressible flows in a one-dimensional setting (Theorem \ref{theorem:tvd}); and the criterion is given in Eq. (\ref{eq:3Dtvd}) for a three-dimensional setting. We thus want to emphasize that the $\phi$ field satisfies the TVD property without having to add any additional flux limiters that is typically done in the literature to achieve this property \citep{laney1998computational}, which would destroy the non-dissipative property of the numerical method and is detrimental to the simulation of turbulent flows and acoustics.

\begin{theorem}
On a uniform one-dimensional grid, if $0\le\phi^k_i\le1$ is satisfied, then $\phi^k_i$ is said to satisfy the total-variation-diminishing property (TVD), where $k$ is the time-step index and $i$ is the grid index, provided
\begin{equation}
      \left(\frac{2\Gamma \epsilon}{\Delta x^2}\right) \ge \max_i\left\{\left(\frac{u_{i+1}^k - u_{i-1}^k}{2\Delta x}\right)\right\},
    \label{eq:tvd}
\end{equation}
is satisfied.
\label{theorem:tvd}
\end{theorem}

\begin{proof}
Following the proof of Theorem \ref{theorem:bound}, discretizing Eq. (\ref{eq:phi}) on a one-dimensional uniform grid, we can arrive at the form 
\begin{equation}
\phi_i^{k+1} = \tilde{C}_{i-1}^k \phi_{i-1}^k + \tilde{C}_{i}^k \phi_{i}^k + \tilde{C}_{i+1}^k \phi_{i+1}^k,
\end{equation}
where $k$ is the time-step index, $i$ is the grid index, and $\tilde{C}$'s are the coefficients given in Eqs. \eqref{eq:ci-1}-\eqref{eq:ci}. 
\begin{lemma}
A scheme is said to be TVD if $\tilde{C}$'s are all positive and additionally $\tilde{C}^k_i \le 1$ \citep{harten1983high}.
\label{lemma:tvd}
\end{lemma}
Following the proof of Theorem \ref{theorem:bound}, if $\tilde{C}$'s are all positive, then $\phi^k$ is bounded. Additionally, if $\tilde{C}^k_i\le 1$, invoking Eq. (\ref{eq:ci}), we have $1+ (\Delta t/(2\Delta x)) (u^k_{i+1} - u^k_{i-1}) - 2 \Delta t \Gamma \epsilon /(\Delta x^2) \le 1$. Rearranging this, we arrive at the condition
\begin{equation}
      \left(\frac{2\Gamma \epsilon}{\Delta x^2}\right) \ge \max_i\left\{\left(\frac{u_{i+1}^k - u_{i-1}^k}{2\Delta x}\right)\right\},
\end{equation}
which concludes the proof.
\end{proof}

Now, generalizing Theorem \ref{theorem:tvd} for three dimensions, the condition required for the TVD property of $\phi$\textemdash assuming an isotropic mesh\textemdash can be written as
\begin{equation}
      \left(\frac{6\Gamma \epsilon}{\Delta x^2}\right) \ge \max_i\left\{\left(\frac{\delta u^k_i}{\delta x_i}\right)\right\}.
    \label{eq:3Dtvd}
\end{equation}
If the flow is incompressible, this condition is always trivially satisfied. Therefore, boundedness implies TVD and vice-versa, for an incompressible flow. However, for compressible flows, it depends on the local dilatation of the flow. If the flow is compressing, then the dilatation term is negative, and therefore the condition in Eq. (\ref{eq:3Dtvd}) is again trivially satisfied. But high regions of compression limit the time-step size to maintain the boundedness property (as described in Section \ref{sec:bound}) which is a requirement for TVD. On the other hand, if the flow is expanding, then this brings in an additional constraint on the value of $\Gamma$ for a given $\epsilon$ and $\Delta x$ given by Eq. (\ref{eq:3Dtvd}). However, for all the simulations in this work, the time-step size given by acoustic CFL condition and the constraint on $\Gamma$ given by Eq. (\ref{eq:crossover}) were sufficient to maintain the boundedness and TVD properties for the $\phi$ field (see Section \ref{sec:results} for the values of $\Gamma$, $\epsilon$ and $\Delta t$ used for various simulations in this work). In the proof of Theorem \ref{theorem:tvd}, similar to the Theorem \ref{theorem:bound}, a first-order Euler time-stepping scheme was used to arrive at the conditions in Eq. \eqref{eq:tvd} and Eq. \eqref{eq:3Dtvd}; however, these criteria are sufficient to maintain the TVD property of $\phi$ with most higher-order explicit time-stepping schemes.

Summarizing Theorems \ref{theorem:bound} and \ref{theorem:tvd}, assuming that the constraints for an incompressible flow [Eq. (\ref{eq:crossover}) and Eq. (\ref{eq:incompressibletime})] are already satisfied in a compressible flow, high regions of compression might violate both TVD and boundedness properties if the constraint on time-step size $\Delta t$ [Eq. \eqref{eq:3Dboundtime}] is not satisfied and high regions of expansion might violate the TVD property if the constraint on $\Gamma$ [Eq. \eqref{eq:3Dtvd}] is not satisfied. Typical flow conditions and the consequence when the $\phi$ field violates the TVD and boundedness criteria in a compressible flow are schematically shown in Figure \ref{fig:violate}.

\begin{figure}
    \centering
    \includegraphics[width=0.75\textwidth]{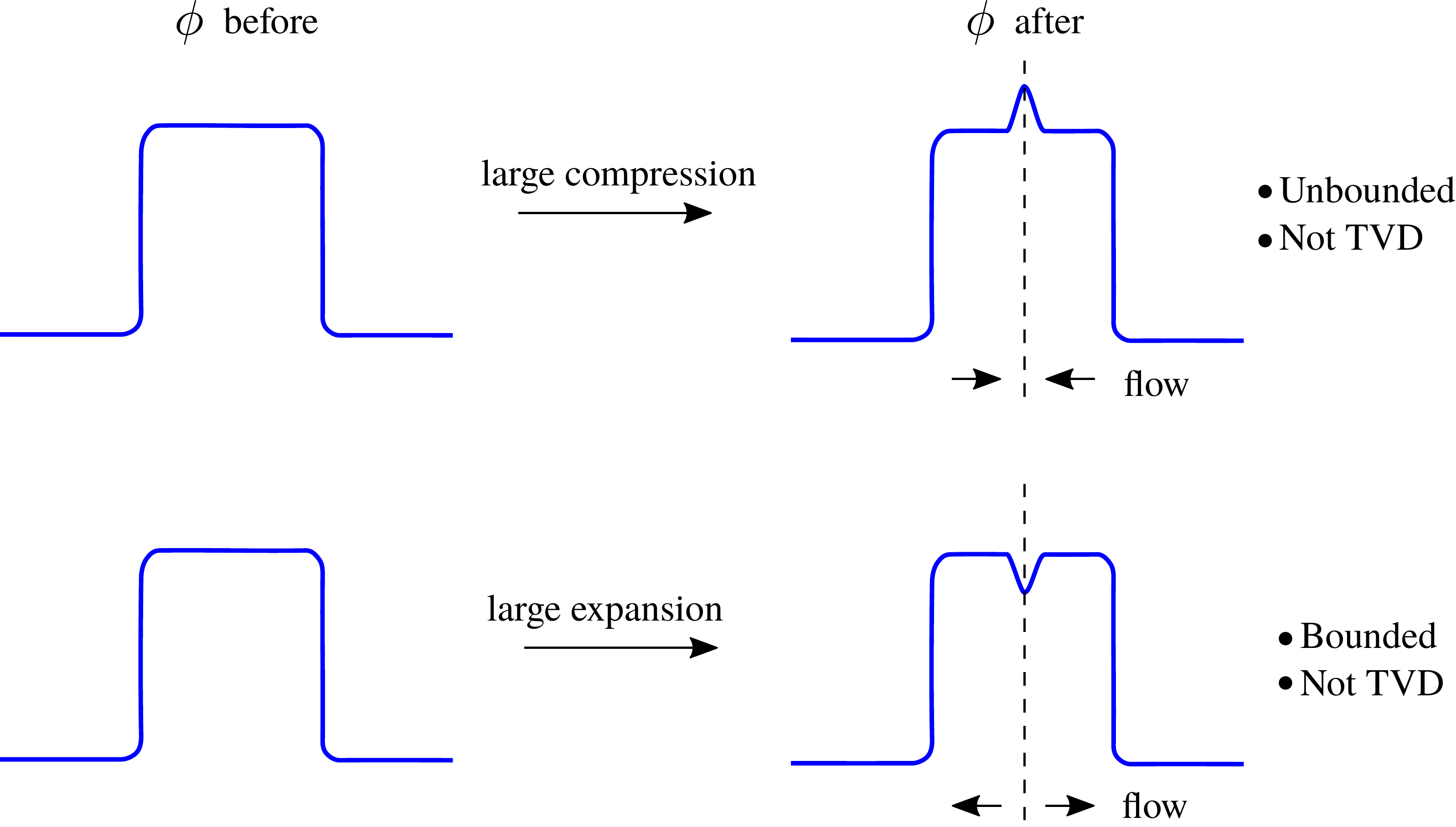}
    \caption{A schematic of a one-dimensional drop in a compressible flow, showing the typical flow conditions that could violate the boundedness and TVD criteria and the consequences of the violation. The solid lines represent the $\phi$ field before and after the violation of the criterion. The small arrows around the dashed line show the flow behavior and the dashed line is the location where the criterion is violated.}
    \label{fig:violate}
\end{figure}

\section{Mass balance equation \label{sec:mass}}

We employ a phenomenological approach to derive the mass balance equation for phase $l$ [Eq. (\ref{eq:mod_mass})]. Similar to Section \ref{sec:bound}, let $\phi=\phi_1$. Then, the mass of phase $1$ per unit total volume is given by $m_1=\rho_1\phi$. Now, starting with the mass balance equation of the form
\begin{equation}
\frac{\partial \rho_1\phi}{\partial t} + \vec{\nabla}\cdot(\rho_1\vec{u}\phi) = \vec{\nabla}\cdot\Big[\rho_1\Gamma\Big\{\epsilon\vec{\nabla}\phi - \phi(1 - \phi)\vec{n}\Big\}\Big],
\label{eq:mass_form1}
\end{equation}
one can see that the equation satisfies two consistency conditions: (a) in the incompressible limit ($\rho_1 \rightarrow \rho_{01}$, $\vec{\nabla}\cdot\vec{u}=0$), it is consistent with the volume fraction advection equation [Eq. \eqref{eq:phi}], where the characteristic density of phase $1$ $\rho_{01}$ is indeed the density of phase $1$ in the incompressible limit; (b) away from the interface ($\phi\rightarrow1$), it consistently reduces to the continuity equation for phase $1$. But, one main disadvantage of this formulation is that it requires explicit computation of $\rho_1$. Typically, $m_1=\rho_1\phi$ is solved in the system of equations, and to obtain $\rho_1$, one uses $\rho_1=m_1/\phi$, which results in inaccurate values of $\rho_1$ at the interface due to round-off errors that stem from division by a small number. To overcome this, we use a form of the equation 
\begin{equation}
\frac{\partial \rho_1\phi}{\partial t} + \vec{\nabla}\cdot(\rho_1\vec{u}\phi) = \vec{\nabla}\cdot\Big[\rho_{01}\Gamma\Big\{\epsilon\vec{\nabla}\phi - \phi(1 - \phi)\vec{n}\Big\}\Big].
\label{eq:mass_form2}
\end{equation}
This form of the equation also satisfies the same consistency conditions in the limit of incompressibility and away from the interface, and is similar to the one proposed in Eq. (\ref{eq:mass_form1}). Hence, we use this form of the mass balance equation since it does not require explicit computation of $\rho_1$. Now, writing Eq. (\ref{eq:mass_form2}) in terms of $m$, we get the mass balance equation for phase $l$ in Eq. \eqref{eq:mod_mass}
\begin{equation}
\frac{\partial m_l}{\partial t} + \vec{\nabla}\cdot(\vec{u}m_l) = \vec{\nabla}\cdot\Big[\rho_{0l}\Gamma\Big\{\epsilon\vec{\nabla}\phi_l - \phi_l(1 - \phi_l)\vec{n}_l\Big\}\Big],
\end{equation}
that is independent of $\rho_l$.

Further, summing up Eq. (\ref{eq:mod_mass}) for phases $1$ and $2$, we can derive the modified version of the continuity equation given by
\begin{equation}
\frac{\partial \rho}{\partial t} + \vec{\nabla}\cdot(\rho\vec{u}) = \vec{\nabla}\cdot\vec{f},
\label{eq:mod_continuity}
\end{equation}
where $\vec{f}=\sum_{l=1}^2 \vec{R}_l=\sum_{l=1}^2 \rho_{0l}\vec{a}_l$ is the net mass-regularization flux. The mass-regularization flux for phase $l$, $\vec{R}_l=\rho_{0l}\vec{a}_l$ in Eq. (\ref{eq:mod_mass}), can be intuitively thought to be a weighted version of the interface-regularization flux $\vec{a}_l$ for phase $l$, where the weight is the characteristic density of the phase, $\rho_{0l}$. This scaling of the flux is employed such that the timescales of regularization of the $\phi$ and $m_l$ fields are similar. It is easy to see that when the densities of two fluids are equal (a single-phase limit), the net mass-regularization flux goes to zero, and therefore, Eq. \eqref{eq:mod_continuity} consistently reduces to the continuity equation for single-phase flows.

The regularization terms in the mass balance equation [Eq. \eqref{eq:mod_mass}] are crucial in maintaining consistency between the mass and volume fraction fields. Figure \ref{fig:regular} shows the effect of regularization terms on all the quantities being solved. Hence, if the volume fraction field is modified due to the regularization terms, reorganization of the mass is required to maintain consistency between the $\rho$ and $\phi$ fields, which is essentially achieved with the use of the regularization terms.

%\textbf{Add a schematic to explain the effect of f}

%\textbf{Caption: Schematics showing the effect of regularization on the $\phi$, $\rho_l$ and $\rho$ fields.}

\begin{figure}
    \centering
    \includegraphics[width=0.4\textwidth]{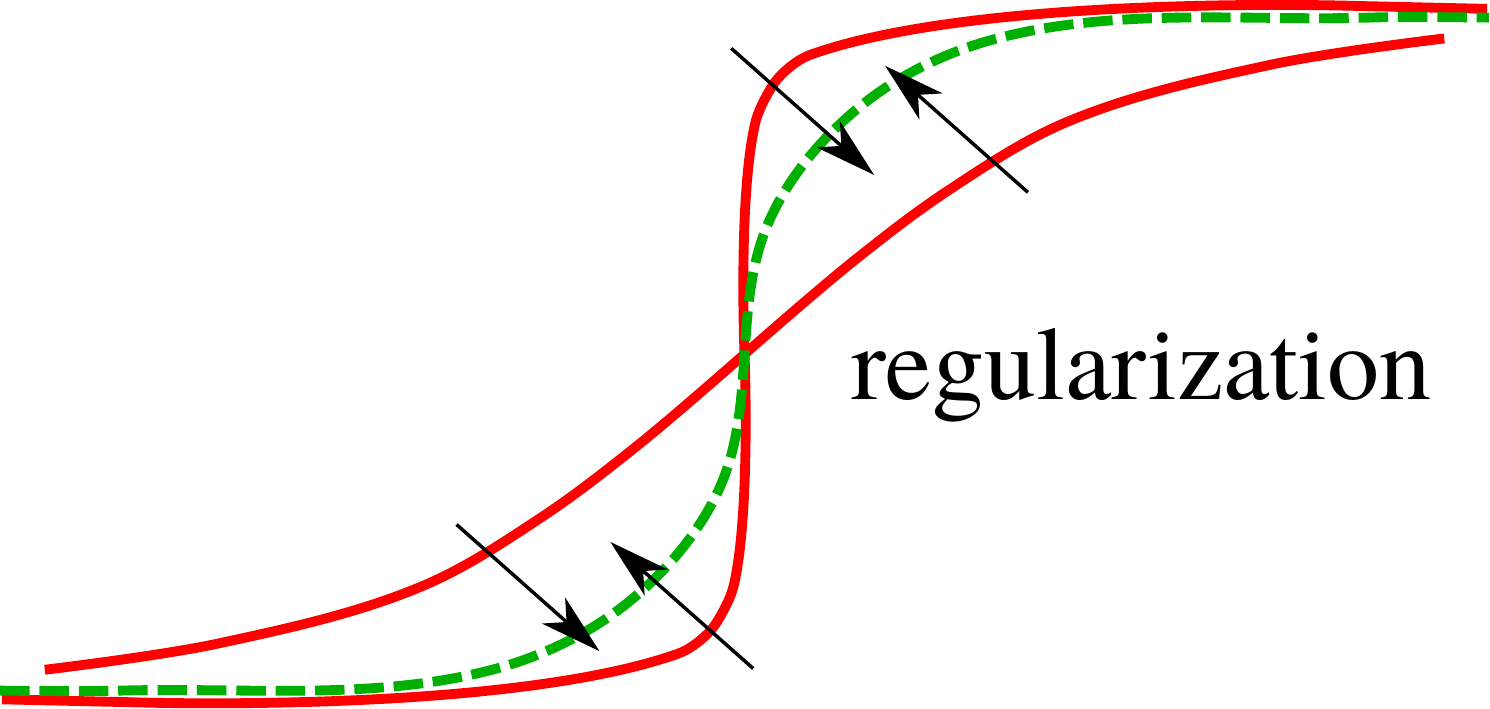}
    \caption{Schematic of effect of regularization terms on all the quantities being solved. The solid lines represent the state of a quantity before regularization and the dashed line represents the equilibrium state of the quantity after regularization.}
    \label{fig:regular}
\end{figure}

\section{Momentum equation \label{sec:momentum}}

Since the momentum of each of the phases is not individually conserved due to exchange of momentum at the interface, it is most efficient to write a single momentum equation for both phases. One can start with the momentum equation of the form 
\begin{equation}
\frac{\partial \rho\vec{u}}{\partial t} + \vec{\nabla}\cdot\{(\rho \vec{u}) \otimes \vec{u} + p \mathds{1}\} = 0,
\label{eq:mom_basic}
\end{equation}
in the inviscid limit. Taking the dot product of this equation with $\vec{u}$, and utilizing the modified continuity equation [Eq. \eqref{eq:mod_continuity}], results in the kinetic energy transport equation of the form
\begin{equation}
    \frac{\partial \rho k}{\partial t} + \vec{\nabla}\cdot(\rho \vec{u}k) + \vec{\nabla}\cdot(\vec{u}p) - p(\vec{\nabla}\cdot\vec{u}) = -k(\vec{\nabla}\cdot\vec{f}),
\end{equation}
where the non-conservative term, $k(\vec{\nabla}\cdot\vec{f})$, represents the spurious contribution to the kinetic energy, that stems from the reorganization of mass across the interface as a result of regularization of the mass and volume fraction fields. Having a spurious non-conservative term in the kinetic energy equation\textemdash even in the continuous form\textemdash is a sign that the solutions of this model could potentially be spurious. This allusion is correct, since the form of the momentum equation in Eq. (\ref{eq:mom_basic}) does not satisfy the interface-equilibrium condition (IEC). 

Now, let's consider the modified version of momentum equation [Eq. \eqref{eq:mod_mom}] in the inviscid limit
\begin{equation}
\frac{\partial \rho\vec{u}}{\partial t} + \vec{\nabla}\cdot\{(\rho \vec{u}) \otimes \vec{u} + p \mathds{1}\} = \vec{\nabla}\cdot(\vec{f}\otimes\vec{u}).
\end{equation}
Taking the dot product of this equation with $\vec{u}$ and utilizing the modified continuity equation [Eq. \eqref{eq:mod_continuity}], results in the kinetic energy transport equation of the form
\begin{equation}
\frac{\partial \rho k}{\partial t} + \vec{\nabla}\cdot(\rho \vec{u}k) + \vec{\nabla}\cdot(\vec{u}p) - p(\vec{\nabla}\cdot\vec{u}) = \vec{\nabla}\cdot(\vec{f}k),
\label{eq:kinetic}
\end{equation}
where there are no non-conservative terms that spuriously contribute to the kinetic energy. Additionally, the form of momentum equation in Eq. (\ref{eq:mod_mom}) also satisfies the IEC (see Section \ref{sec:iec}), thus reinforcing the fact that the solution is not spuriously affected by regularization of the mass and volume fraction fields. In fact, the newly introduced term in the modified version of the momentum equation [Eq. \eqref{eq:mod_mom}] is the regularization term for the momentum, that results in reorganization of the momentum across the interface as shown in Figure \ref{fig:regular} to achieve consistency with the regularized mass and volume fraction fields. For other forms of consistency correction for the momentum equation in the context of a diffuse-interface method, see \citet{tiwari2013diffuse} for compressible flows and \citet{mirjalili2019consistent} for incompressible flows. 

This consistency correction to the momentum is crucial for compressible flows, without which the spurious momentum (or velocity) contribution to kinetic energy might eventually lead to unbounded solutions, especially in a non-dissipative numerical method. However, stable solutions have been obtained in the past without this consistency correction, mostly in the incompressible regime, for low Reynolds numbers and low density ratios. This could be due to: the use of dissipative numerical schemes that stabilize the method by attenuating the kinetic energy, and thus preventing its unbounded growth; and the enforcement of the divergence-free condition for the velocity, which stabilizes the method in the process of projecting the velocity field onto a divergence-free field, because the spurious velocities at the interface due to the inconsistent momentum formulation are predominantly of the dilatational type.

%However, for high-density ratio flows and at high $Re$, \citep{mirjalilithesis} showed that the consistency correction is also required for incompressible flows to obtain stable solutions. %The discrete analogue of all the mathematical operations performed in this section will be presented in the subsequent sections....{\color{red}(UPDATE)}

\section{Energy equation: entropy conservation form\label{sec:energy}}

Entropy is not conserved in a diffuse-interface method, even in the inviscid limit, due to the regularization of the interface (irreversible process). Entropy should only be conserved if the interface is already perfectly regular (equilibrium state) and the effects of all the regularization terms are identically zero. Thus, we seek to achieve approximate entropy conservation instead of exact conservation; and derive the conservative form of the regularization terms in the energy equation, with a constraint that it should satisfy the IEC. We first look at the case of exact entropy conservation and show that it does not satisfy the IEC; and then look at the case where the IEC is satisfied and then state that the entropy is not conserved, as is expected for an irreversible process (second law of thermodynamics).  

%Use the environment lemma1 for lemmas outside the theorem
\newtheorem{lemma1}{Lemma}[section]

\begin{lemma1}
Let $s_l$ be the specific entropy of phase $l$, and $T_l$ be the corresponding temperature. Then, the form of the internal energy equation that satisfies (entropy conservation)
\begin{equation}
    \sum_{l=1}^2 \left(\rho_l\phi_lT_l \frac{Ds_l}{Dt}\right) = 0,
    \label{eq:entropy}
\end{equation}
in the inviscid limit is
\begin{equation}
    \frac{\partial \rho e}{\partial t} + \vec{\nabla}\cdot(\rho \vec{u} e) + \vec{\nabla}\cdot(p\vec{u}) - \vec{u}\cdot\vec{\nabla} p = \sum_{l=1}^2 \left\{h_l \vec{\nabla}\cdot\vec{R}_l\right\}.
    \label{eq:nonconservative_ie}
\end{equation}
\end{lemma1}
\begin{proof}
Let us start with an internal energy equation of the form
\begin{equation}
    \frac{D \rho e}{D t} + \rho h (\vec{\nabla}\cdot\vec{u}) + X = 0,
    \label{eq:noncons_ie}
\end{equation}
where $X$ is the unknown term to be determined. Expressing the mixture internal energy in terms of phase quantities
\begin{equation}
    \frac{D \rho e}{D t} = \sum_{l=1}^2\frac{D (\phi_l\rho_l e_l)}{D t} = \sum_{l=1}^2 \left\{\phi_l\frac{D (\rho_l e_l)}{D t} + \rho_le_l\frac{D \phi_l}{D t}\right\}, 
\end{equation}
and then using Gibbs's relation to express internal energy in terms of entropy, we get
\begin{equation}
    \mathrm{d}(\rho_le_l) = \rho_l\mathrm{d}e_l + e_l\mathrm{d}\rho_l = \rho_lT_l\mathrm{d}s_l + h_l\mathrm{d}\rho_l.
\end{equation}
Using this in Eq. (\ref{eq:noncons_ie}) results in
\begin{equation}
    \sum_{l=1}^2 \left\{\rho_l\phi_lT_l \frac{Ds_l}{Dt} + \phi_lh_l\frac{D\rho_l}{Dt} + \phi_lh_l\rho_l(\vec{\nabla}\cdot\vec{u}) + \rho_le_l\frac{D\phi_l}{Dt}\right\} + X = 0. 
\end{equation}
Now, using Eqs. (\ref{eq:mod_volumefraction})\textendash(\ref{eq:mod_mass}), one obtains
\begin{equation}
    \sum_{l=1}^2 \left\{\rho_l\phi_lT_l \frac{Ds_l}{Dt} + h_l \vec{\nabla}\cdot\vec{R}_l - p_l(\vec{\nabla}\cdot\vec{a}_l)\right\} + X = 0. 
\end{equation}
Hence, if $X=\sum_{l=1}^2 \{p_l(\vec{\nabla}\cdot\vec{a}_l) - h_l \vec{\nabla}\cdot\vec{R}_l \}$, then the condition in Eq. (\ref{eq:entropy}) is satisfied. Now, invoking the isobaric closure law (Section \ref{sec:model}), $\sum_{l=1}^2 \{p_l(\vec{\nabla}\cdot\vec{a}_l)\}=0$, and the proof is complete.
\end{proof}

The internal energy equation of the form in Eq. (\ref{eq:nonconservative_ie}) does not satisfy the IEC, which also alludes to the fact that the entropy is not conserved exactly in a diffuse-interface method with regularization terms. Since we now only seek approximate entropy conservation, we modify the regularization term in Eq. (\ref{eq:nonconservative_ie}) such that it satisfies the IEC; and the conservative form of the regularization term is restored. Thus, we arrive at the final form of the internal energy equation (taking the $h_l$ on the right-hand side of Eq. (\ref{eq:nonconservative_ie}) inside the divergence operator) 
\begin{equation}
    \frac{\partial \rho e}{\partial t} + \vec{\nabla}\cdot(\rho \vec{u} e) + \vec{\nabla}\cdot(p\vec{u}) - \vec{u}\cdot\vec{\nabla} p = \sum_{l=1}^2 \vec{\nabla}\cdot(\rho_l h_l \vec{a}_l).
    \label{eq:conservative_ie}
\end{equation}

In compressible flows, internal energy is not a conserved quantity due to the reversible exchange of compression/expansion work between internal and kinetic energies, but the sum of internal and kinetic energy is conserved. Hence, summing up the internal energy transport equation [Eq. \eqref{eq:conservative_ie}] and the kinetic energy transport equation [Eq. \eqref{eq:kinetic}], we obtain
\begin{equation}
\frac{\partial E}{\partial t} + \vec{\nabla}\cdot(\vec{u} E) + \vec{\nabla}\cdot(p\vec{u}) = \vec{\nabla}\cdot(\vec{f}k) + \sum_{l=1}^2 \vec{\nabla}\cdot{(\rho_l h_l \vec{a}_l )}. 
\end{equation}
Clearly, all the terms in this equation are in conservative form as desired; and since this equation was obtained by summing the forms of internal energy and kinetic energy equations that satisfy the IEC, this form of the total energy equation also satisfies the IEC. With the inclusion of viscous terms, we get the final form of the total energy transport equation in Eq. (\ref{eq:mod_energy}). 

%%%%%%%%%%%%%%%%%%%%%%%%%
%TO DO

%For final paper show the evolution of entropy for a stationary inviscid bubble and show that it is positive for phi profiles that are both initially diffuse and initially sharp.

%Plot the entropy field in the domain.

%To evaluate entropy compute: conservative RHS of internal energy - non-conservative RHS of internal energy.
%%%%%%%%%%%%%%%%%%%%%%%%%

%%%%%%%%%%%%%%%%%%%
%TO DO

%For final paper, rearrange the RHS in internal energy equation such that it resembles the modified flux when brought onto to the LHS.

%%%%%%%%%%%%%%%%%%%

\section{Interface-equilibrium condition \label{sec:iec}}

In incompressible flows, the divergence-free condition constraints the velocity and pressure fields, and hence eliminates the possibility of spurious solutions at the interface (in the absence of surface tension forces). However, such a constraint in compressible flows is absent, and thus care must be taken in the implementation of any numerical scheme in order to avoid spurious solutions at the interface. The IEC provides a consistency condition to check and eliminate the forms of the model and the numerical discretizations that contribute spuriously to the solution. 

\begin{lemma1}
The proposed conservative diffuse-interface model in Eqs. (\ref{eq:mod_volumefraction})-(\ref{eq:mod_energy}) satisfies the IEC defined in Postulate (\ref{post:iec}).
\end{lemma1}

\begin{proof}
\subsection*{Part (a). Mechanical equilibrium: uniform velocity across the interface}

Consider a one-dimensional second-order central discretization of the mass balance equation [Eq. \eqref{eq:mod_mass}] on a uniform grid, assuming $u^k_i=u_0$
\begin{equation}
\left(\rho_l\phi_l\right)^{k+1}_i - \left(\rho_l\phi_l\right)^{k}_i = -\Delta t \left\{ \frac{(\rho_l \phi_l)_{i+1} - (\rho_l \phi_l)_{i-1}}{2\Delta x} \right\}^ku_0 + \Delta t  \left( \frac{R_{l,i+1} - R_{l,i-1}}{2\Delta x} \right)^k,
\label{eq:mass_iec}
\end{equation}
where $k$ is the time-step index, and $i$ is the grid index. Now, consider a one-dimensional second-order central discretization of the momentum equation [Eq. \eqref{eq:mod_mom}] on a uniform grid, assuming $u^k_i=u_0$ and $p^k_i=p_0$
\begin{equation}
(\rho u)^{k+1}_i - \rho^{k}_iu_0 = -\Delta t \left( \frac{\rho_{i+1} - \rho_{i-1}}{2\Delta x} \right)^ku_0^2 + \Delta t  \left( \frac{\sum_{l=1}^2 R_{l,i+1} - \sum_{l=1}^2 R_{l,i-1}}{2\Delta x} \right)^ku_0.
\end{equation}
Subtracting this from the sum of the discrete mass balance equations [Eq. \eqref{eq:mass_iec}] for phases $1$ and $2$ gives $u^{k+1}_i=u_0$.

\subsection*{Part (b). Thermodynamic equilibrium: uniform pressure across the interface}
Consider a one-dimensional second-order central discretization of the internal energy equation [Eq. \eqref{eq:conservative_ie}] in terms of phase quantities, on a uniform grid, assuming $u^k_i=u_0$ and $p^k_i=p_0$, and using Eq. (\ref{eq:enthalpy}) to express $h_l$ in terms of $p$ and $\rho_l$
\begin{equation}
\begin{aligned}
\sum_{l=1}^2 \left(\phi_l\rho_le_l\right)^{k+1}_i - \sum_{l=1}^2 \left(\phi_l\rho_le_l\right)^k_i = -\Delta t \sum_{l=1}^2\left\{ \frac{(\rho_le_l\phi_l)_{i+1}-(\rho_le_l\phi_l)_{i-1}}{2\Delta x} \right\}^k u_0 \\
+\Delta t \left[\sum_{l=1}^2 \left\{\frac{p_0(1+\alpha_l) - \beta_l}{\alpha_l}\right\}\left\{\frac{a_{l, i+1} - a_{l,i-1}}{2\Delta x}\right\} \right]^k,
\end{aligned}
\end{equation}
and expressing $e_l$ in terms of $p_l$ using the EOS results in the discretized equation for pressure
\begin{equation}
\begin{aligned}
%\begin{gather*}
\left(\sum_{l=1}^2\frac{\phi_l}{\alpha_l}\right)^{k+1}p^{k+1}_i - \left(\sum_{l=1}^2\frac{\phi_l\beta_l}{\alpha_l}\right)^{k+1} - \left(\sum_{l=1}^2\frac{\phi_l}{\alpha_l} \right)^kp_0 + \left(\sum_{l=1}^2\frac{\phi_l\beta_l}{\alpha_l} \right)^k\\
= -\Delta t \left\{\frac{\left(\sum_{l=1}^2 \frac{\phi_l}{\alpha_l}\right)_{i+1}p_0 - \left(\sum_{l=1}^2 \frac{\phi_l\beta_l}{\alpha_l}\right)_{i+1}-\left(\sum_{l=1}^2 \frac{\phi_l}{\alpha_l} \right)_{i-1}p_0 + \left(\frac{\sum_{l=1}^2 \phi_l\beta_l}{\alpha_l}\right)_{i-1}}{2\Delta x} \right\}^ku_0\\
+\Delta t \left[\sum_{l=1}^2 \left\{\frac{p_0(1+\alpha_l) - \beta_l}{\alpha_l}\right\}\left\{\frac{a_{l, i+1} - a_{l,i-1}}{2\Delta x}\right\} \right]^k.
\label{eq:pressure_iec}
%\end{gather*}
\end{aligned}
\end{equation}
Now, let $L(\phi_l)$ be a one-dimensional second-order central discretization of the volume fraction advection equation for phase $l$ [Eq. \eqref{eq:mod_volumefraction}] on a uniform grid. Assuming $u^k_i=u_0$, and subtracting Eq. (\ref{eq:pressure_iec}) from the equation $\Big(\sum_{l=1}^2 {L(\phi_l)}/{\alpha_l}\Big)p_0 - \Big(\sum_{l=1}^2 {L(\phi_l)\beta_l}/{\alpha_l}\Big)$, results in $p^{k+1}_i=p_0$, which concludes the proof.
\end{proof}

%%%%%%%%%%%%%%%%%%%%%
%Higher order schemes also satisfy IEC
%%%%%%%%%%%%%%%%%%%%%

\section{Full proposed model including viscous, surface tension, and gravity forces\label{sec:finalmodel}}

We finally present the full proposed model in Eqs. (\ref{eq:volumef})-(\ref{eq:pressuref}) along with the viscous, surface tension, and gravity terms, where the highlighted terms are the newly introduced interface-regularization (diffusion\textendash sharpening) terms. We model the surface-tension force using the continuum surface force (CSF) method \citep{brackbill1992continuum} as a volumetric body force in the momentum equation; and the surface-tension energy term is included in the total energy equation to consistently account for the exchange of surface energy and the kinetic energy in the flow \citep{perigaud2005compressible}.

\begin{equation}
\frac{\partial \phi}{\partial t} + \vec{\nabla}\cdot(\vec{u}\phi) = \phi(\vec{\nabla}\cdot\vec{u})+\highlightgreen{\vec{\nabla}\cdot\vec{a}}
\label{eq:volumef}    
\end{equation}
\begin{equation}
\frac{\partial m_l}{\partial t} + \vec{\nabla}\cdot(\vec{u} m_l) = \highlightgreen{\vec{\nabla}\cdot \vec{R}_l} \hspace{0.5cm} l=1,2
\end{equation}
\begin{equation}
\frac{\partial \rho\vec{u}}{\partial t} + \vec{\nabla}\cdot(\rho \vec{u} \otimes \vec{u} + p \mathds{1}) = \vec{\nabla}\cdot\doubleunderline\tau + \highlightgreen{\vec{\nabla}\cdot(\vec{f}\otimes\vec{u})} + \sigma \kappa \vec{\nabla} \phi + \rho \vec{g}
\end{equation}
\begin{equation}
\frac{\partial E}{\partial t} + \vec{\nabla}\cdot(\vec{u} E) + \vec{\nabla}\cdot(p\vec{u}) = \vec{\nabla}\cdot(\doubleunderline\tau\cdot\vec{u}) + \highlightgreen{\vec{\nabla}\cdot(\vec{f}k)} + \highlightgreen{\sum_{l=1}^2 \vec{\nabla}\cdot{(\rho_l h_l \vec{a}_l )}} + \sigma \kappa \vec{u}\cdot\vec{\nabla} \phi + \rho \vec{g}\cdot \vec{u}
\end{equation}
\begin{equation}
p = \frac{\rho e + \Big(\frac{\phi \beta_1}{\alpha_1} + \frac{(1 -\phi) \beta_2}{\alpha_2}\Big)} {\Big( \frac{\phi}{\alpha_1} + \frac{1- \phi}{\alpha_2}\Big)}
\label{eq:pressuref}
\end{equation}
In Eqs. (\ref{eq:volumef})-(\ref{eq:pressuref}), $\sigma$ is the surface-tension coefficient, $\kappa=-\vec{\nabla}\cdot\vec{n}$ is the curvature of the interface, and $\vec{g}$ is the gravitational acceleration. For the convenience of the readers, the final model [Eqs. \eqref{eq:volumef}-\eqref{eq:pressuref}] has been rewritten in Appendix A, where: the highlighted modeling terms have been further expanded in terms of primitive variables; and the general mixture EOS has been expressed in terms of the parameters of the individual phase stiffened-gas EOSs.

\section{Numerical implementation \label{sec:implement}}

\subsection{Numerical discretization}

In this work, we choose to use the fourth-order Runge-Kutta (RK4) time-stepping scheme and second-order central-differencing scheme for the discretization of all spatial operators.

A finite-volume collocated discretization strategy has been employed, wherein all the variables are stored at cell centers and the fluxes are evaluated on the cell faces. Thus, it can be extended to arbitrary unstructured grids in a relatively straightforward manner. The non-linear sharpening term, that is present on the right-hand side of the volume fraction advection, is also present in all other equations, which is more evident in the fully expanded form of the governing equations in Appendix A; and the discretization of this term should be consistently used across all the equations to obtain accurate oscillation-free solutions.

%The choice of discretization of the non-linear sharpening term on the right-hand side of the volume fraction advection equation [Eq. \eqref{eq:volumef}] is crucial in achieving accurate oscillation-free solutions. The discretization of this term is given in full detail in Appendix B. This non-linear sharpening term is also present in all other equations, which is more evident in the fully expanded form of the governing equations in Appendix A; and the discretization of this non-linear term presented in Appendix B should be consistently used across all the equations.

\subsection{Performance and scalability \label{sec:performance}}

To verify and validate the proposed model and the numerical method, it has been implemented in the in-house code\textemdash CTR-DIs3D\textemdash that has been optimized for better parallel scalability (see Appendix B for more details on the in-house solver). Apart from the performance improvement through solver optimization, diffuse-interface methods are inherently known to be cost effective and easily parallelizable compared to sharp-interface methods. This is due to the absence of expensive and localized geometric reconstruction of the interface that could potentially result in load-balancing issues. The partial-differential-equation-only nature of the diffuse-interface method results in well-balanced loads throughout the domain; and when combined with low-order central-difference schemes, gives rise to a low-cost, robust, and highly-scalable method. 

To evaluate the parallel-scaling efficiency of the in-house CTR-DIs3D solver, and therefore the diffuse-interface method, a strong-scaling test and a weak-scaling test have been performed on the Mira supercomputer at Argonne National Laboratory (ANL) (see, \citet{Jain2018} for the scaling of the two-dimensional solver). The results from the strong-scaling test are shown in Figure \ref{fig:strong-scale}, where the actual speedup and the actual time per time step are compared against the ideal speedup and the ideal time per time step, respectively. The results from the weak-scaling test are shown in Figure \ref{fig:weak-scale}, where the ideal time and the actual time are plotted against the number of cores. 

From the results, it is evident that weak scaling is almost ideal from $1$ to $10^3$ cores, beyond which the efficiency drops to roughly $80\%$ for $25\ \mathrm{K}$ cores. The results from the strong-scaling test show an ideal behavior for large grid sizes per core, and a drop in the scaling efficiency for grid sizes less than $12.5\ \mathrm{K}$ per core, due to a higher communication overhead compared to the time of computation for smaller grid sizes. This could be due to a very small total computational time per time step, which is a result of highly optimized single-core performance of the solver and a low cost numerical method; therefore resulting in a higher communication-to-calculation time ratio and a non-ideal parallel scalability.

\begin{figure}
    \centering
    \includegraphics[width=\textwidth]{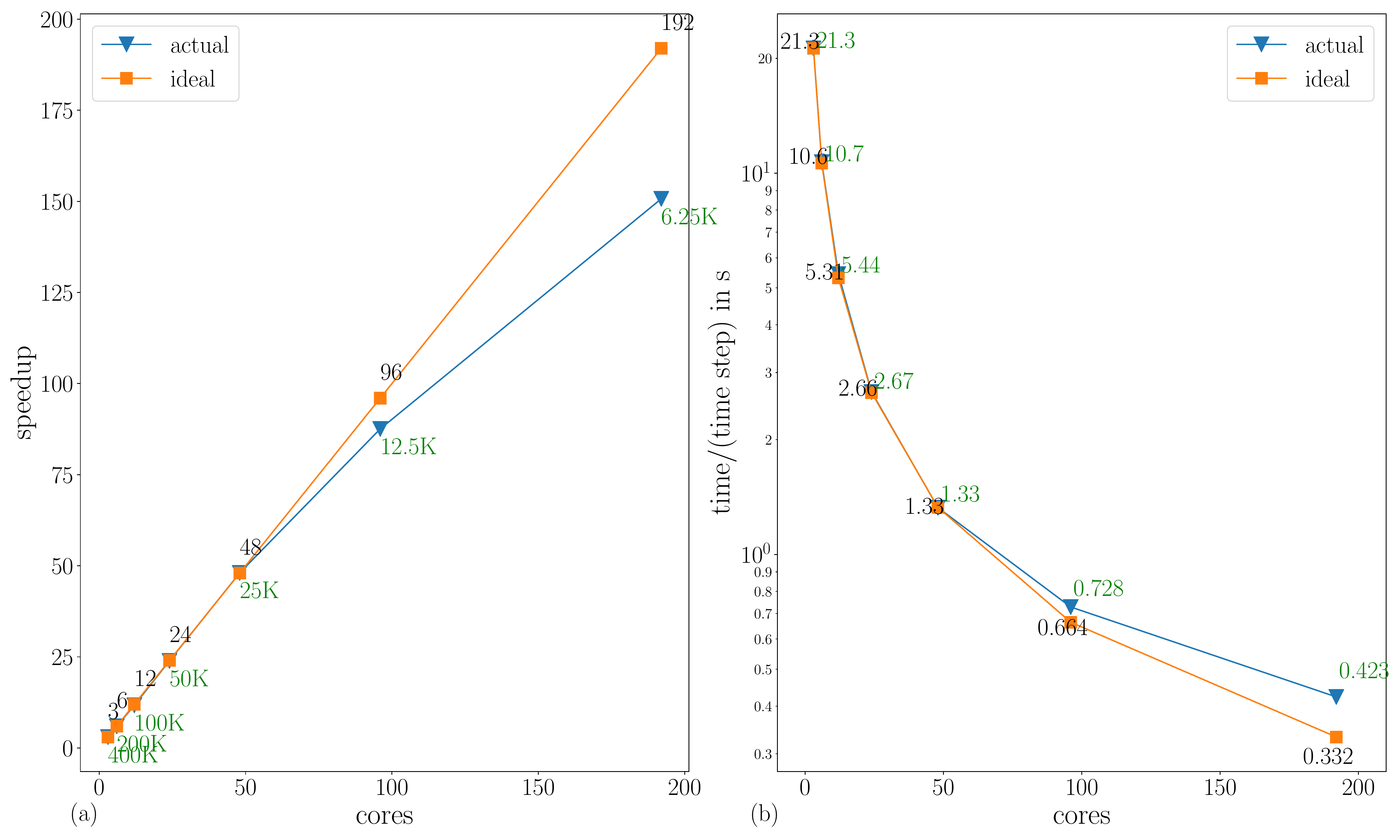}
    \caption{Strong scaling of the CTR-DIs3D solver on the Mira supercomputer at Argonne National Laboratory. (a) The ideal speedup and the actual speedup achieved are plotted against the number of cores. The numerical values in black, above the curve, are the number of cores; and the numerical values in green, below the curve, represent the number of grid points per core. (b) The actual time per time step (along with the numerical values in green, above the curve) and the ideal time per time step (along with the numerical values in black, below the curve) are plotted against the number of cores from the same test.}
    \label{fig:strong-scale}
\end{figure}
\begin{figure}
    \centering
    \includegraphics[width=\textwidth]{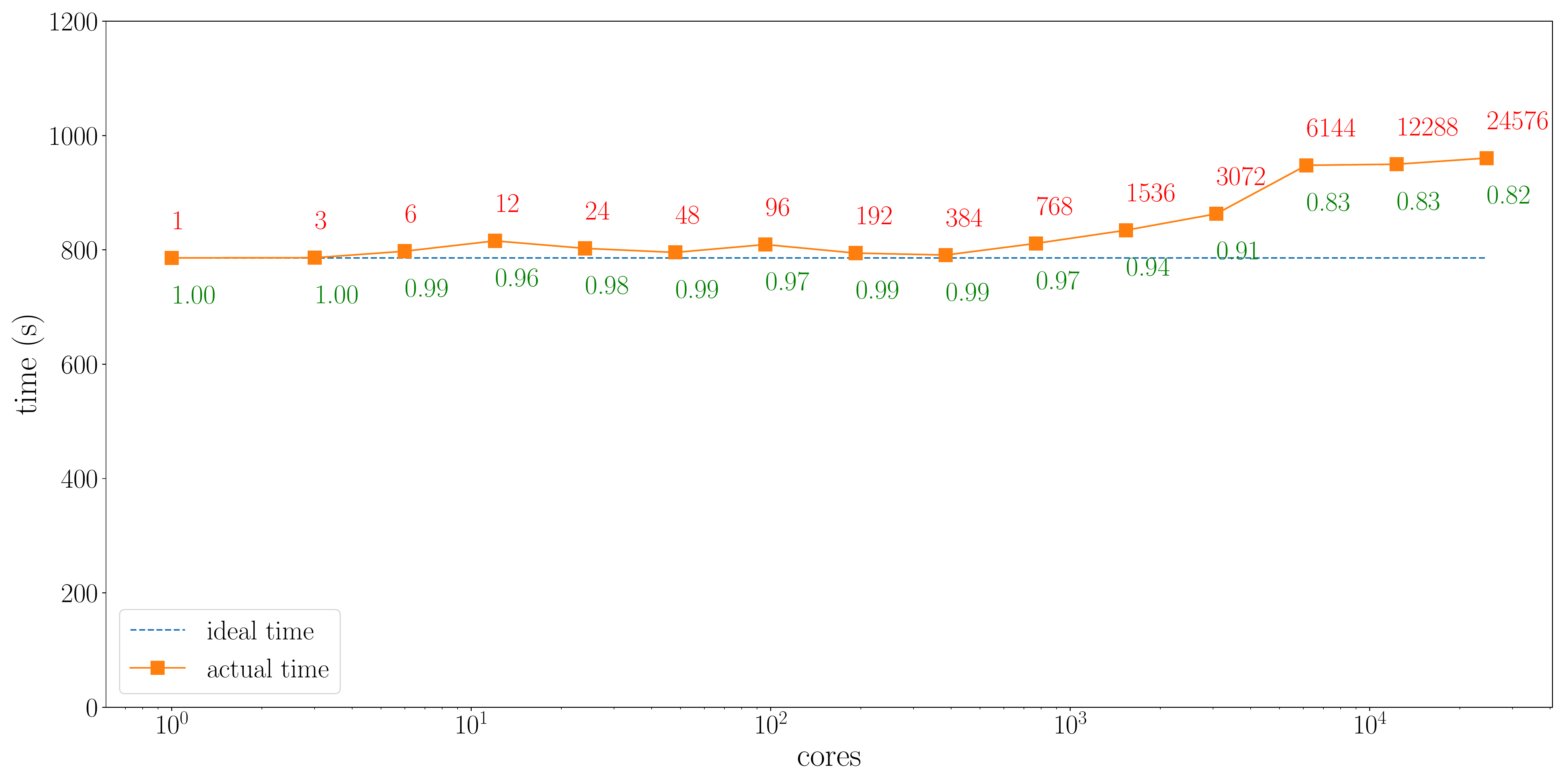}
    \caption{Weak scaling of the CTR-DIs3D solver on the Mira supercomputer at Argonne National Laboratory. (a) Ideal time and the actual time for 50 time steps are plotted against the number of cores. The numerical values in red, above the curve, are the number of cores; and the numerical values in green, below the curve, represent the weak scaling efficiency.}
    \label{fig:weak-scale}
\end{figure}

\section{Results \label{sec:results}}

In this section, several verification tests are presented that are used to assess the newly proposed model, its numerical discretization, and the implementation. The verification tests used in this work can be broadly classified into: (a) interface advection test cases, that test the accuracy of interface-capturing capability of the method, (b) surface tension test case, that tests the accuracy of the implementation of surface tension effects in the model, (c) acoustics test cases, that test the accuracy of propagation of sound and its interaction with material interfaces in the flow, and (d) complex flows, that test the stability and robustness of the numerical scheme, and the accuracy of the method for high-Reynolds-number flows. In all the test cases, properties of the fluids used are that of air, water, and kerosene; unless specified otherwise. The properties of these fluids are listed in Table \ref{tab:prop}. 

\begin{table}[]
\centering
\begin{tabular}{@{}llll@{}}
\toprule
                         & air       & water  & kerosene  \\ \midrule
$\rho\ (\mathrm{kg/m^3})$ & 1.225     & 997   & 820    \\
$\mu\ (\mathrm{N/m^2})$   & 0.0000181 & 0.00089 & 0.00164\\
$\gamma$                 & 1.4       & 4.4     & 4.4 \\
$\pi\ (\mathrm{MPa})$     & 0         & 600  &  326.6  \\ 
$c\ (\mathrm{m/s})$     & 338.1      & 1627.4   & 1324 \\ \bottomrule
\end{tabular}
\caption{Properties of the fluids used in this work.}
\label{tab:prop}
\end{table}

\subsection{Initial conditions}

One approach to set the initial values of volume fraction, $\phi$, is to start with a value of $1$ in one phase, and $0$ in the other; and then reinitialize the $\phi$ field using the equation
\begin{equation}
\frac{\partial \phi}{\partial \tau} = \vec{\nabla}\cdot\left\{\epsilon\vec{\nabla}\phi - \phi(1 - \phi)\vec{n}\right\},
\label{eq:reinitialize}
\end{equation}
where $\tau$ is a pseudo time, such that the $\phi$ field relaxes to the equilibrium solution\textemdash a hyperbolic tangent function\textemdash to the above equation. Alternatively, one could use an initial analytical profile to specify $\phi$ as
\begin{equation}
\phi = \frac{1}{2} \Big\{1 + \tanh{\big(\frac{x-x_0}{2\epsilon}\big)}\Big\}.
\label{equ:phi_sol}
\end{equation}
which is a one-dimensional equilibrium solution to Eq. (\ref{eq:reinitialize}), where $x_0$ is the desired location of the interface. 

In all the test cases in the present work, the initial profile of $\phi$ is analytically specified using Eq. (\ref{equ:phi_sol}) with an initial value of $\epsilon$ as $\epsilon_0=\Delta x$, unless specified otherwise. Once $\phi$ is initialized, the densities are initialized as $\rho_1=\rho_{01}\phi$, and $\rho_2=\rho_{02}(1 - \phi)$; and the viscosity as $\mu=\mu_1\phi+\mu_2(1 - \phi)$. The velocity field, $\vec{u}$, and the pressure, $p$, are initialized as desired for the problem of interest. To initialize the total energy, $E$, the internal energy, $\rho e$, is first computed using the mixture EOS in Eq. (\ref{eq:pressuref}), given the pressure, the volume fraction, and the parameters of the fluids, as
\begin{equation}
    \rho e = {p\Big( \frac{\phi}{\alpha_1} + \frac{1- \phi}{\alpha_2}\Big)} - {\Big(\frac{\phi \beta_1}{\alpha_1} + \frac{(1 -\phi) \beta_2}{\alpha_2}\Big)}, 
\end{equation}
and then the total energy is computed by summing up the kinetic and internal energy contributions ($E = \rho\vec{u}\cdot\vec{u}/2 + \rho e$).

%%%%%%%%%%%%%%%%%%%%%%%%%%%%%%%%%%%%%%%%%%
%Add reinitialization test case for final paper
%%\textbf{Reinitialization case and steady interface thickness in 1 time step}
%%%%%%%%%%%%%%%%%%%%%%%%%%%%%%%%%%%%%%%%%%

\subsection{Interface advection tests\label{sec:interface_advection}}

This section contains some standard test cases and newly proposed test cases, to assess the accuracy of the shape of the interface, computed using the proposed conservative diffuse-interface method. The three test cases presented in this section are: (a) drop in a shear flow, which is a standard test introduced by \citet{BELL1989257} and \citet{rider1998reconstructing}, and has been extensively used in the literature \citep{tryggvason2011direct} to assess the accuracy of the interface in an incompressible shearing flow; (b) drop in a compressible shear flow is a new test case that we propose, to evaluate our model in terms of the accuracy of the shape of the interface in a compressible shearing flow; and (c) star in a spiralling flow is also a new test case that we propose, to evaluate our model in terms of the accuracy in resolving sharp interfacial features in a compressible rotating flow. Additionally, test cases (b) and (c) also help in assessing the volume conservation properties of the method; and this is an important metric since the volumes of individual phases are not inherently conserved in a compressible flow.

\subsubsection{Drop in a shear flow\label{sec:drop_in_shear_flow}}

Consider a two-dimensional computational square domain of dimensions $[0,1]\times[0,1]$. A circular drop of radius, $R=0.15$, is initially centered at $(0.5,0.75)$. Since the quantity being assessed is the accuracy of the temporal evolution of the interface shape, which is computed by solving the volume fraction advection equation, the hydrodynamics can be decoupled from this test case. This is achieved by not solving the momentum and energy balance equations, and by directly imposing the velocity field in the domain at every time step as
\begin{equation}
\begin{aligned}
    u = -\sin^2\left(\pi x\right)\sin(2\pi y)\cos\left(\frac{\pi t}{T}\right),\\ 
    v = -\sin(2\pi x)\sin^2(\pi y)\cos\left(\frac{\pi t}{T}\right), 
\end{aligned}
\end{equation}
where $T=4$ is the time period of the flow; $t$ is the time coordinate; $x$ and $y$ are the spatial coordinates; and, $u$ and $v$ are the velocity components along $x$ and $y$ directions, respectively. This test case is designed in such a way that the drop undergoes a shearing deformation for half a period until $t=T/2=2$, and then the flow field is reversed due to the $\cos(\pi t/T)$ term, such that the initial drop shape should be recovered at the final time, $t=T=4$. Since the velocity field is chosen to be incompressible, the drop undergoes deformation without changing its volume. The magnitude of the velocity field, $\left\Vert\vec{u}\right\Vert_2$, is plotted in Figure \ref{fig:velocity_drop_in_vortex} along with the streamlines at the initial time, $t=0$. 

\begin{figure}
\centering
\includegraphics[width=0.5\textwidth]{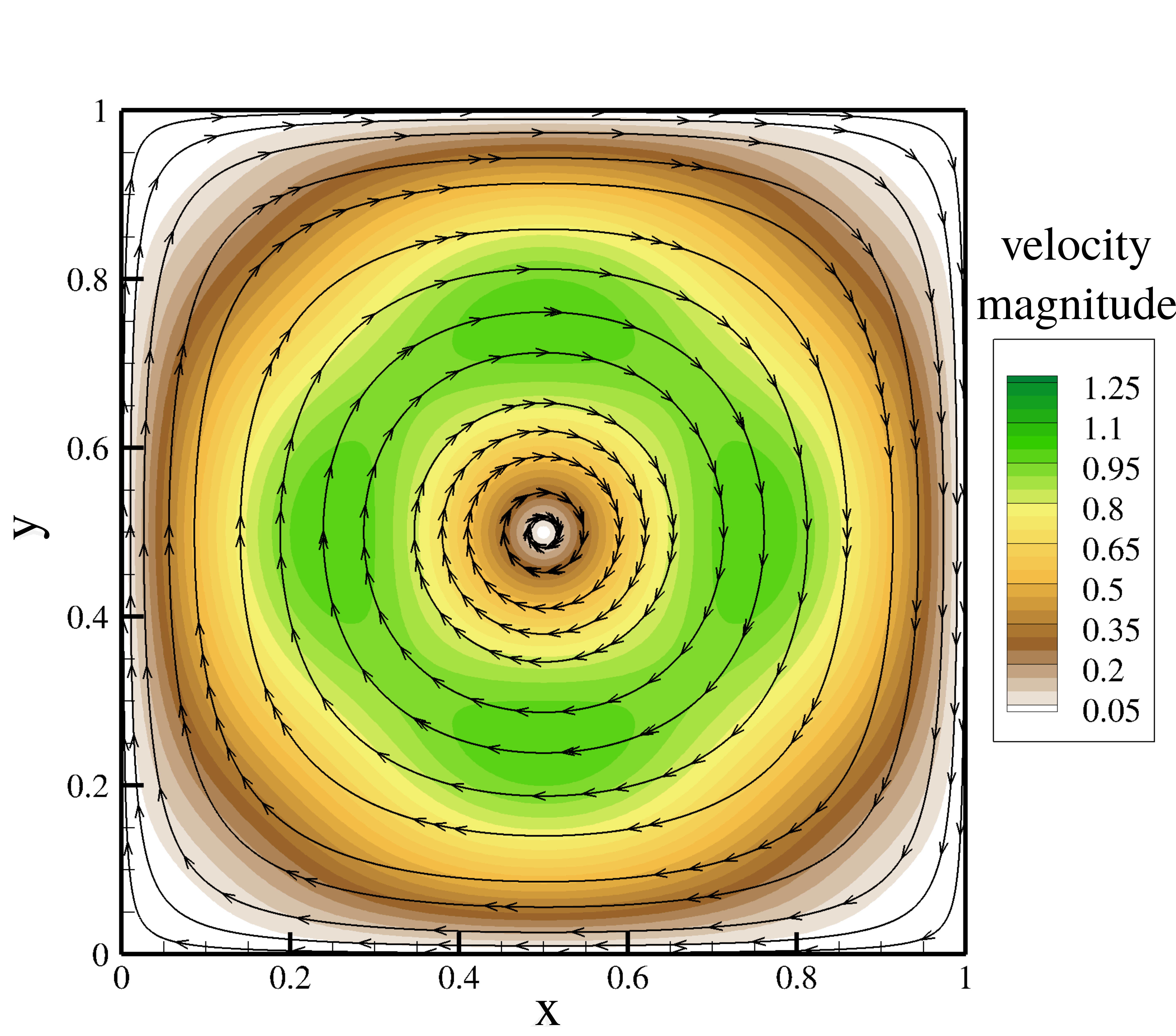}
\caption{The imposed velocity field in the domain for the drop-in-a-shear-flow case at the initial time, $t=0$. The color field represents the magnitude of the velocity field; and the lines represent the streamlines of the velocity field, with the arrows showing the direction of the flow.}
\label{fig:velocity_drop_in_vortex}
\end{figure}

\begin{figure}
    \centering
    \includegraphics[width=\textwidth]{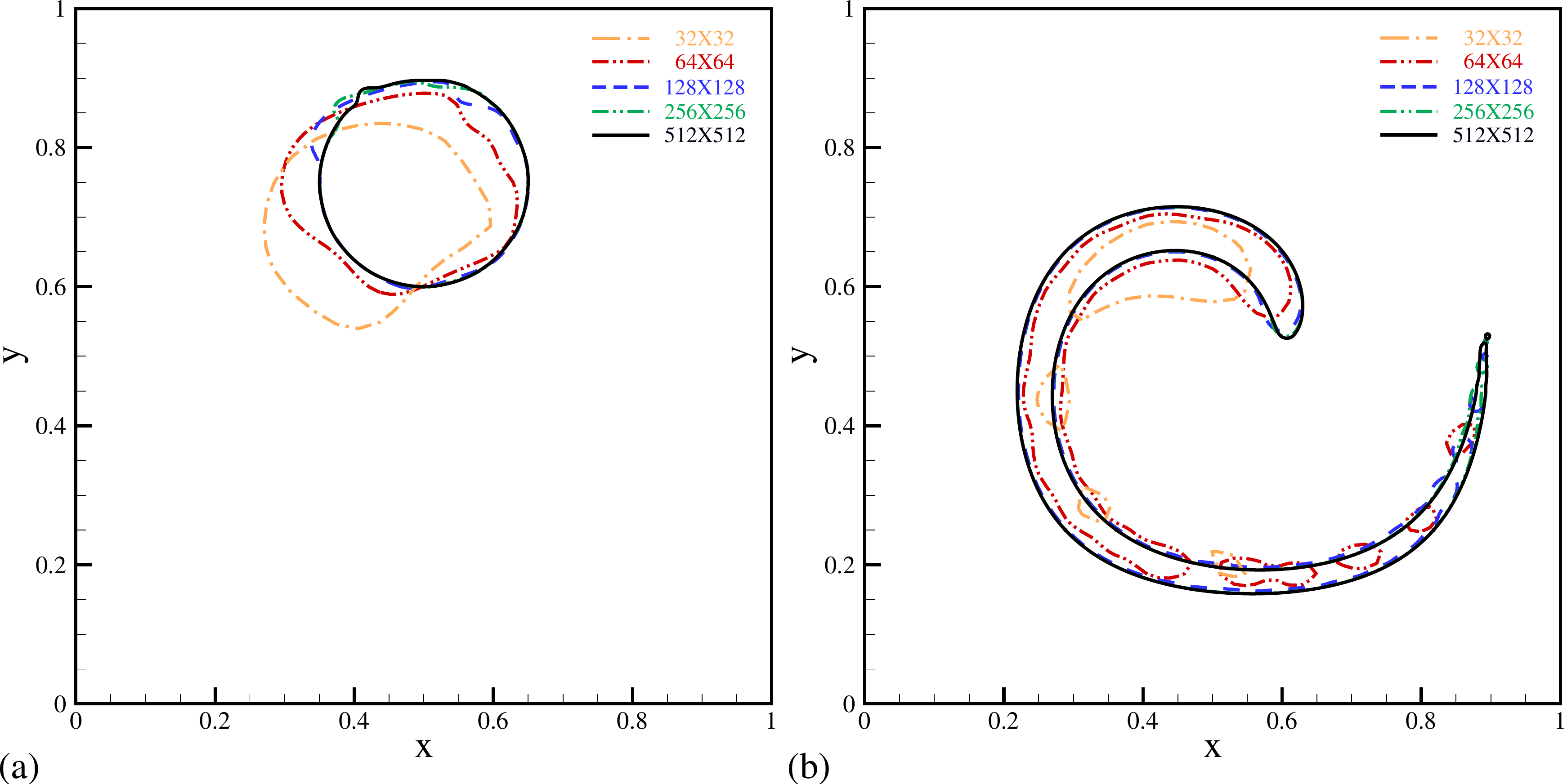}
    \caption{The computed drop shape after (a) half a period at $t=T/2=2$, and (b) a full period at $t=T=4$, on five different grids: $32^2$, $64^2$, $128^2$, $256^2$, and $512^2$. The interface is defined by the $\phi=0.5$ contour.}
    \label{fig:drop_in_vortex_result}
\end{figure}

The domain was discretized using $N_x\times N_y$ grid points; and five different grids, $32^2$, $64^2$, $128^2$, $256^2$, and, $512^2$, were chosen to study the convergence of the error in the shape of the drop. The values of $\epsilon=\Delta x$ and $\Gamma=|u|_{max}$ were used in the simulation. These values were chosen such that they satisfy the boundedness criterion in Eq. \eqref{eq:crossover}, and the TVD criterion in Eq. \eqref{eq:3Dtvd} (see Section \ref{sec:parameters}). Figure \ref{fig:drop_in_vortex_result} shows the resultant shape of the drop obtained at $t=T/2$ and $t=T$ on five different grids. With an increase in the number of grid points, a clear convergence in the drop shape can be seen. Since at the final time, the drop is supposed to return to its original shape at the initial time, the ``exact" final shape of the drop is known to be a circle, and hence the error in the ``actual" shape of the drop obtained can be computed. We compute the error as
\begin{equation}
NS_{error} = \frac{\left\Vert\phi_{f} - \phi_{in}\right\Vert_1}{N_x\times N_y}, 
\label{eq:ns_error}
\end{equation}
where $NS_{error}$ is the cell-normalized shape error, $\phi_f$ is the final volume fraction field, and $\phi_{in}$ is the initial volume fraction field. The total error, $\left\Vert \phi_f - \phi_{in} \right\Vert_1$, is normalized by the number of grid points, $N_x \times N_y$, so that the resultant error, $NS_{error}$, is independent of the number of grid points, and can be compared across simulations performed on different grids. Alternatively, $\phi=0.5$ contour can be chosen to be the location of the interface, and the error associated with the interface shape can be calculated. However, the volume enclosed (area in two dimensions) by the $\phi=0.5$ contour is not a conserved quantity in a diffuse-interface method; and hence the error defined in Eq. (\ref{eq:ns_error}) is often preferred over the latter. The computed shape error, $NS_{error}$, on five different grids are listed in Table \ref{tab:drop_in_shear_flow} along with the order of convergence, $SE_{order}$. The shape error decreases with an increase in the number of grid points, with an order of convergence roughly between $1$ and $2$. The shape error is also plotted against the grid size in Figure \ref{fig:error}; and the slopes of the individual line segments represent the local order of convergence that is listed in Table \ref{tab:drop_in_shear_flow}.

\begin{table}
\centering
\begin{tabular}{@{}lll@{}}
\toprule
Grid & $NS_{error}$ & $SE_{order}$ \\ \midrule
$32\times32$ & 0.05344 & \\
$64\times64$ & 0.02174 & 1.2290 \\
$128\times128$ & 0.004724 & 2.3010 \\
$256\times256$ & 0.001946 & 1.2139 \\ 
$512\times512$ & 0.0006397 & 1.5210 \\ \bottomrule
\end{tabular}
\caption{Grid convergence of the interface shape error for the drop-in-a-shear-flow case.}
\label{tab:drop_in_shear_flow}
\end{table}

\begin{figure}
    \centering
    \includegraphics[width=0.5\textwidth]{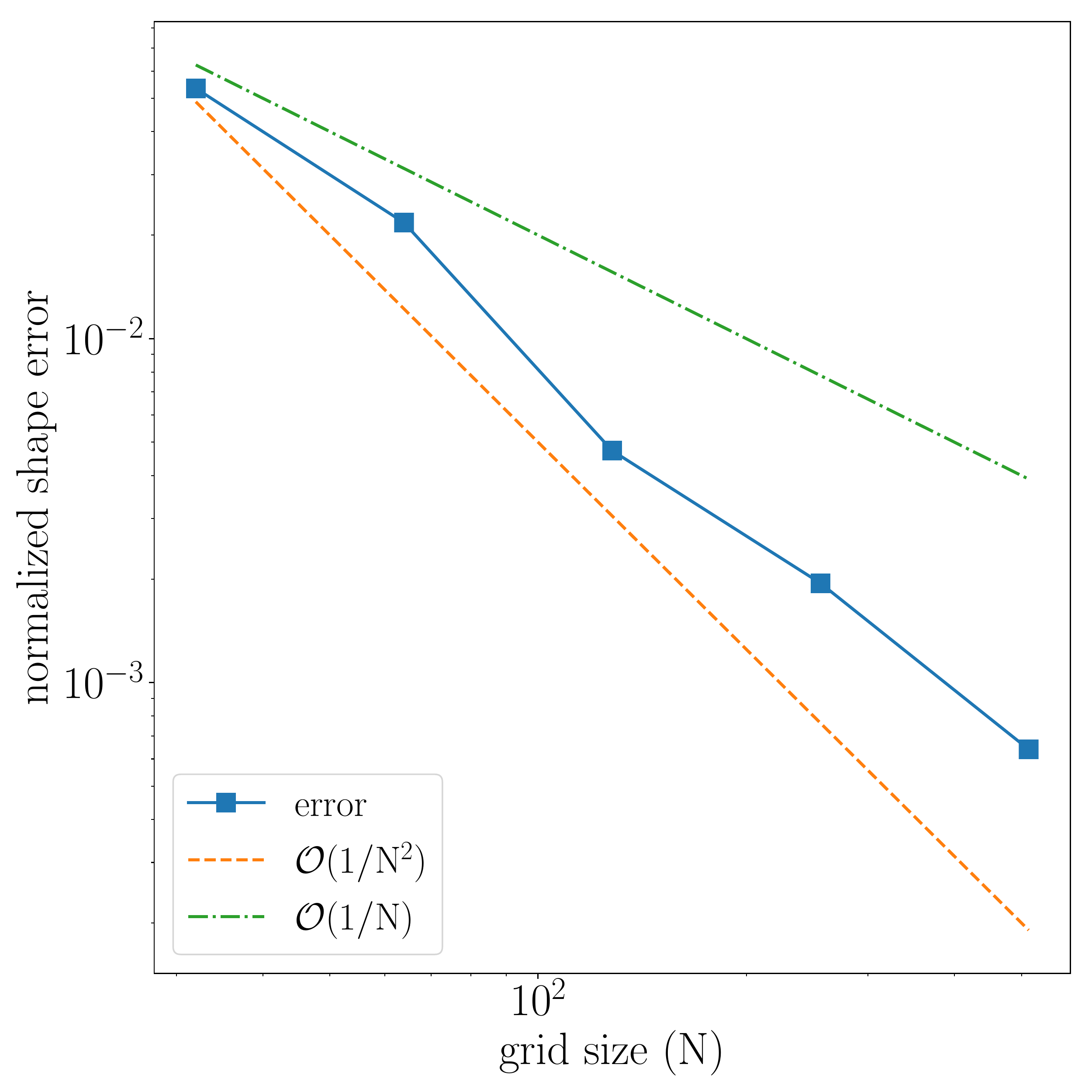}
    \caption{The plot of normalized shape error, $NS_{error}$, against the grid size, $N$, where $N=N_x=N_y$ represents the number of grid points along one of the directions. The dashed and dash-dotted lines represent the reference lines with slopes $1/N^2$ and $1/N$, respectively.}
    \label{fig:error}
\end{figure}

\subsubsection{Drop in a compressible shear flow\label{sec:drop_in_compressible_flow}}

Consider a two-dimensional computational square domain of dimensions $[0,1]\times[0,1]$. A circular drop of radius, $R=0.15$, is initially centered at $(0.5,0.75)$. Since the proposed conservative diffuse-interface method can handle compressibility effects, the imposed velocity field is composed of both solenoidal, $\vec{u}_s$, and dilatational, $\vec{u}_d$, components given by
\begin{equation}
\begin{aligned}
    u_s = -\sin^2(\pi x)\sin(2\pi y)\cos\left(\frac{\pi t}{T}\right),\\ 
    v_s = -\sin(2\pi x)\sin^2(\pi y)\cos\left(\frac{\pi t}{T}\right),\\ 
    u_d = (y - x) \cos\left(\frac{\pi t}{T}\right),\\ 
    v_d = (- x - y + 1) \cos\left(\frac{\pi t}{T}\right),
\end{aligned}
\end{equation}
where $T=2$ is the time period of the flow; $u_s$ and $v_s$ are the solenoidal velocity components along $x$ and $y$ directions, respectively; and $u_d$ and $v_d$ are the dilatational velocity components along $x$ and $y$ directions, respectively. The total imposed velocity, $\vec{u}$, is therefore sum of $\vec{u}_s$ and $\vec{u}_d$ in the domain at every time step. The dilatation is spatially uniform in the domain and is given by
\begin{equation}
    \vec{\nabla}\cdot\vec{u} = -2 \cos\left(\frac{\pi t}{T}\right).
    \label{eq:dil}
\end{equation}
This test case is designed in such a way that the drop undergoes a shearing deformation along with a uniform compression for half a period until $t=T/2=1$, and then the flow field is reversed due to the $\cos(\pi t/T)$ term, such that the initial drop shape and the volume should be recovered at the final time, $t=T=2$. The compression ratio of the drop can be defined as
\begin{equation}
        CR=\frac{V_{in}}{V_h}=3.57,
\end{equation}
where $V_{in}$ is the initial volume of the drop; and $V_h$ is the volume of the drop at $t=T/2$. The magnitude of the velocity field, $\left\Vert\vec{u}\right\Vert_2$, is plotted in Figure \ref{fig:drop_in_compressible_flow} (a) along with the streamlines, at the initial time; and the temporal evolution of the volume of the drop is plotted in Figure \ref{fig:drop_in_compressible_flow} (b).

\begin{figure}
\centering
\includegraphics[width=\textwidth]{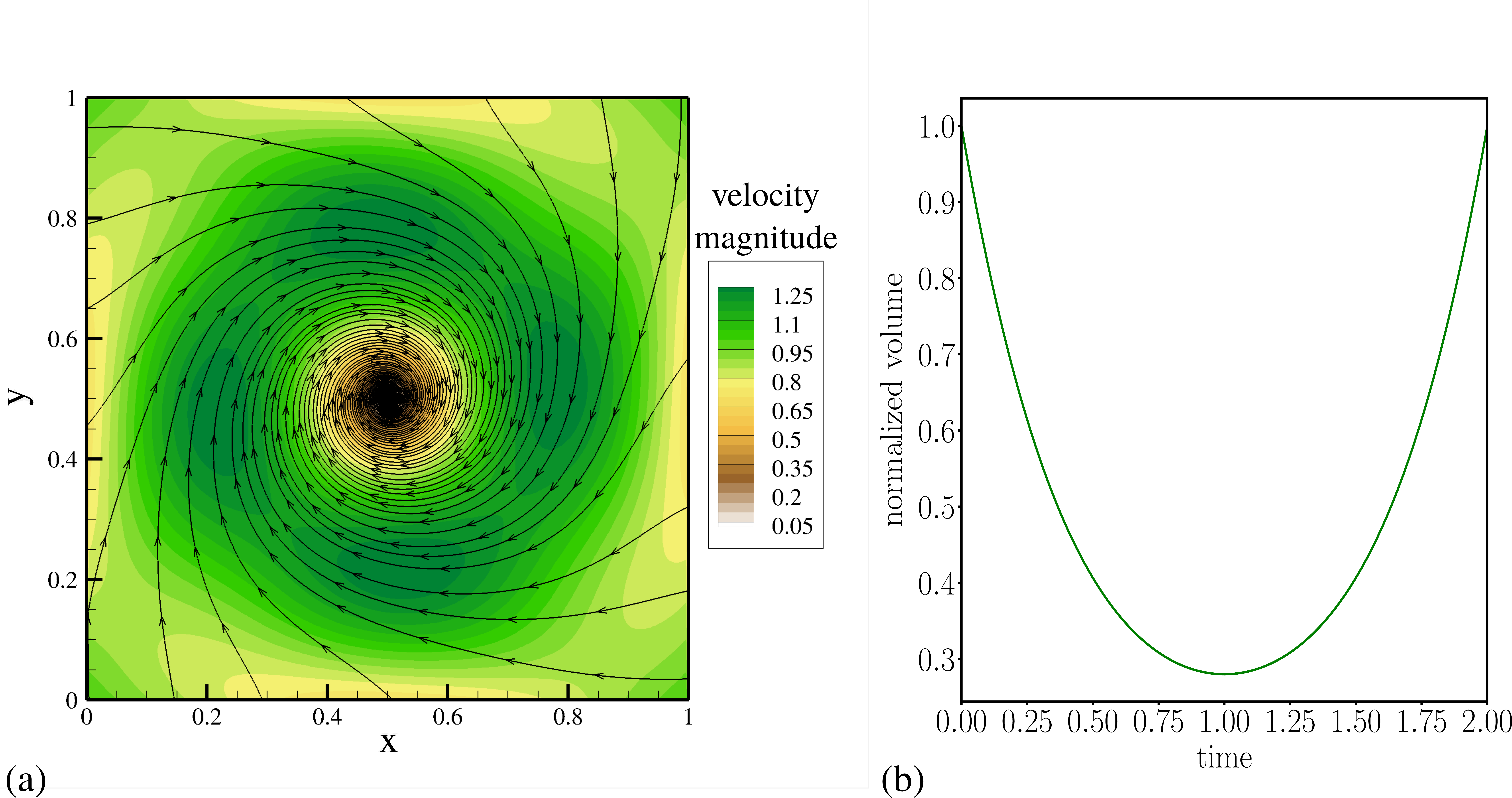}
\caption{(a) The imposed velocity in the domain for the drop-in-a-compressible-shear-flow case at the initial time, $t=0$. The color field represents the magnitude of the velocity field; and the lines represent the streamlines of the velocity, with the arrows showing the direction of the flow. (b) The volume of the drop as a function of time, showing that the initial volume is recovered at the final time, $t=2$.}
\label{fig:drop_in_compressible_flow}
\end{figure}

\begin{figure}
    \centering
    \includegraphics[width=\textwidth]{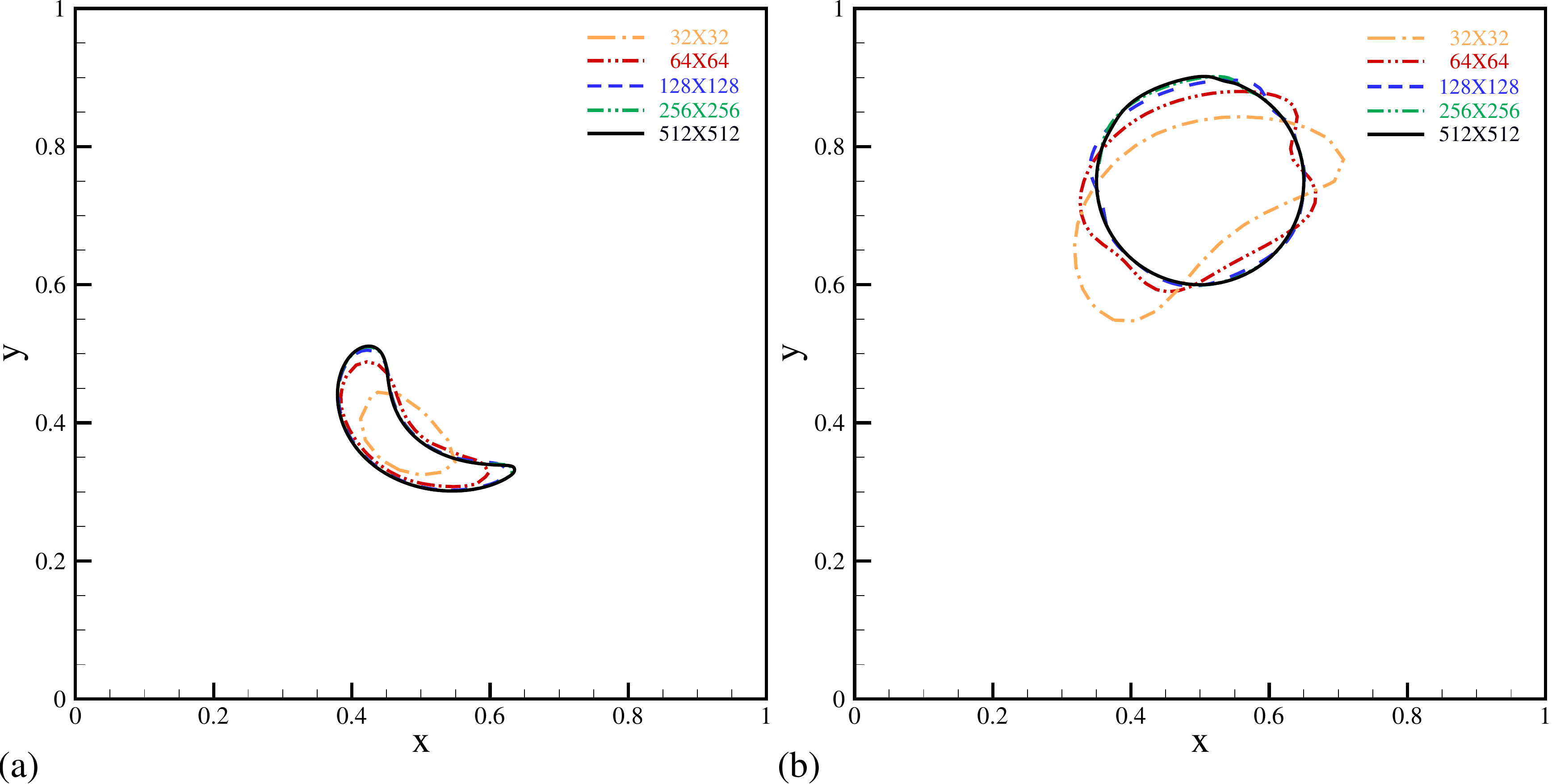}
    \caption{The computed drop shape  after (a) half a period at $t=T/2=1$, and (b) a full period at $t=T=2$, on five different grids: $32^2$, $64^2$, $128^2$, $256^2$, and $512^2$. The interface is defined by the $\phi=0.5$ contour.}
    \label{fig:drop_in_compressible_flow_result}
\end{figure}

The domain was discretized using $N_x\times N_y$ grid points; and five different grids, $32^2$, $64^2$, $128^2$, $256^2$, and $512^2$, were chosen to study the convergence of the error in the shape and in the volume of the drop. The values of $\epsilon=\Delta x$ and $\Gamma=|u|_{max}$ were used in the simulation. Figure \ref{fig:drop_in_compressible_flow_result} shows the resultant shape of the drop obtained at $t=T/2$ and $t=T$ on five different grids. With an increase in the number of grid points, a clear convergence in the drop shape can be seen. Since at the final time the drop is supposed to return to its original shape and volume, the ``exact" final shape and the volume of the drop is known; and hence the error in the ``actual" shape and the volume of the drop can be computed. We compute the shape error as already defined in Eq. (\ref{eq:ns_error}), and the volume error as 
\begin{equation}
V_{error} = V_{f} - V_{in}=\int_\Omega (\phi_f - \phi_{in}) dV=\sum_{i=1}^{i=N_x\times N_y} \left(\frac{\phi_{f,i} - \phi_{in,i}}{N_x\times N_y}\right),
\label{eq:nv_error}
\end{equation}
where $V_f$ is the final volume of the drop; and $\Omega$ is the domain. Similar to the normalized shape error, the total volume error, $V_{error}$, is a quantity that is independent of the number of grid points, and can be compared across simulations performed on different grids. Additionally, we also define the percentage change in volume of the drop as
\begin{equation}
\%V_{error} = \frac{V_{f} - V_{in}}{V_{in}}\times100 = \frac{\sum_{i=1}^{i=N_x\times N_y}\left(\phi_{f,i}-\phi_{in,i}\right)}{\sum_{i=1}^{i=N_x\times N_y}\phi_{in,i}}\times100,
\label{eq:pv_error}
\end{equation}
which is also independent of the number of grid points and can be compared across simulations performed on different grids. The computed shape error, $NS_{error}$, the volume error, $V_{error}$, and the percent volume change, $\%V_{error}$, on five different grids are listed in Table \ref{tab:drop_in_compressible_flow} along with the order of convergence for the shape error, $SE_{order}$. The shape error decreases with an increase in the number of grid points, with an order of convergence approximately equal to $2$, which is better than the incompressible case (Table \ref{tab:drop_in_shear_flow}). The absolute values of $NS_{error}$ are also smaller compared to the incompressible case, which could be due to the reduced volume of the drop that results in reduced shearing deformation. $V_{error}$ and $\%V_{error}$ also decrease with an increase in the number of grid points, but at a much higher rate; and $V_{error}$ can be seen to have reached machine precision for the $256^2$ and $512^2$ grids, even for a compression ratio as high as $CR=3.57$. This shows that the proposed method has good volume conservation properties.

\begin{table}
\centering
\begin{tabular}{@{}lllll@{}}
\toprule
Grid & $V_{error}$ & $\%V_{error}$ & $NS_{error}$ & $SE_{order}$\\ \midrule
$32\times32$ & $2.0010 \times10^{-4}$ & $0.2488$ & 0.04529 & \\
$64\times64$ & $5.2806 \times10^{-6}$ & $0.007213$ & 0.01581 & 1.4317 \\
$128\times128$ & $1.2158 \times10^{-9}$ & $1.7048 \times10^{-6}$ & 0.003924 & 2.0153 \\
$256\times256$ & $-2.1663 \times10^{-14}$ & $-3.0579 \times10^{-11}$ & 0.0009728 & 2.01663 \\ 
$512\times512$ & $1.5266 \times10^{-15}$ & $2.1584 \times10^{-12}$ & 0.0002654 & 1.8329 \\ \bottomrule
\end{tabular}
\caption{Grid convergence of the shape and volume error for the drop-in-a-compressible-shear-flow case.}
\label{tab:drop_in_compressible_flow}
\end{table}

%\textit{\textbf{Volume of the drop at:}}
%\begin{itemize}
%    \item $T=0$ is $1160.6978353420$.
%    \item $T=1$ is $324.9063781043$.
%    \item $T=2$ is $1160.6977839505$.
%\end{itemize}

%{\color{ForestGreen} \textbf{Net loss in the volume of the drop for a compression ratio of $3.57$ is $0.00000442763\ \%$}}. {\color{red}UPDATE}

\subsubsection{Star in a spiralling flow\label{sec:star_in_spiralling_flow}}

In this test case, a two-dimensional computational square domain, $[-0.5,0.5]\times[-0.5,0.5]$, is used. A star shaped drop of radius, $R=0.2(1 + {\cos(4\theta)}/{4})$, is initially centered at $(0,0)$. Unlike the test case in Section \ref{sec:drop_in_compressible_flow}, the imposed velocity field is composed of only the dilatational component, $\vec{u}=\vec{u}_d$, given by
\begin{equation}
\begin{aligned}
u_d = (y - x) \cos\left(\frac{\pi t}{T}\right),\\ 
v_d = (- x - y) \cos\left(\frac{\pi t}{T}\right),
\end{aligned}
\end{equation}
where $T=2$ is the time period of the flow. The dilatation is spatially uniform in the domain, and is the same as in the test case in Section \ref{sec:drop_in_compressible_flow} and is given in Eq. (\ref{eq:dil}). 

This test case is designed in such a way that the star undergoes a rotational motion along with a uniform compression for half a period until $t=T/2=1$, and then the flow field is reversed due to the $\cos(\pi t/T)$ term, such that the initial star shape and the volume should be recovered at the final time, $t=T=2$. The compression ratio in this test case is also $3.57$. The magnitude of the velocity field, $\left\Vert\vec{u}\right\Vert_2$, is plotted in Figure \ref{fig:star_in_spiralling_flow} along with the streamlines, at the initial time.

%{\color{ForestGreen} New setup that tests the capability of interface-capturing method for sharp features under dilatation!} 

\begin{figure}
\centering
\includegraphics[width=0.5\textwidth]{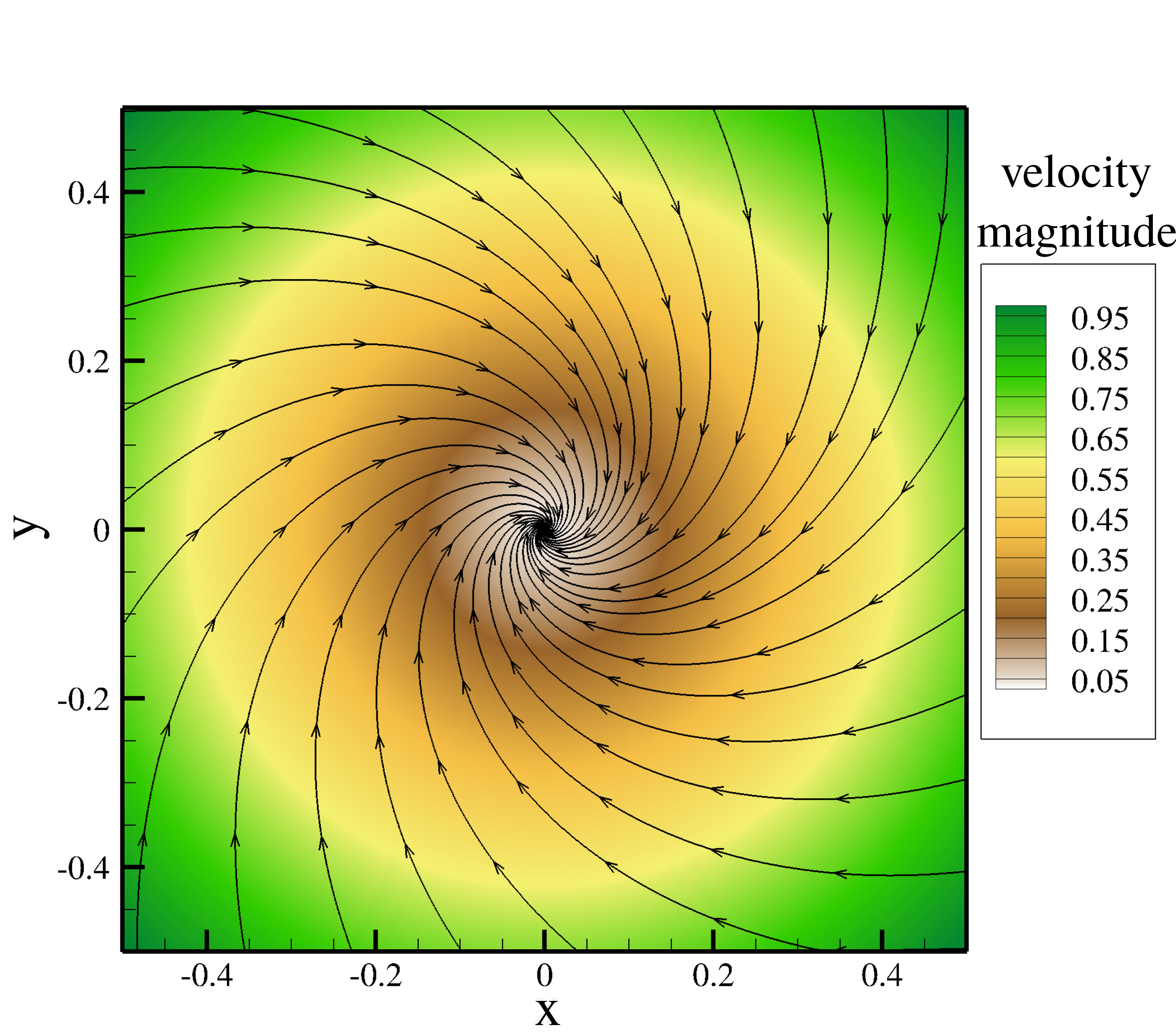}
\caption{The imposed velocity field in the domain for the star-in-a-spiralling-flow case at $t=0$. The color field represents the magnitude of the velocity field and the lines represent the streamlines along with the arrows showing the direction of the flow.}
\label{fig:star_in_spiralling_flow}
\end{figure}

\begin{figure}
    \centering
    \includegraphics[width=\textwidth]{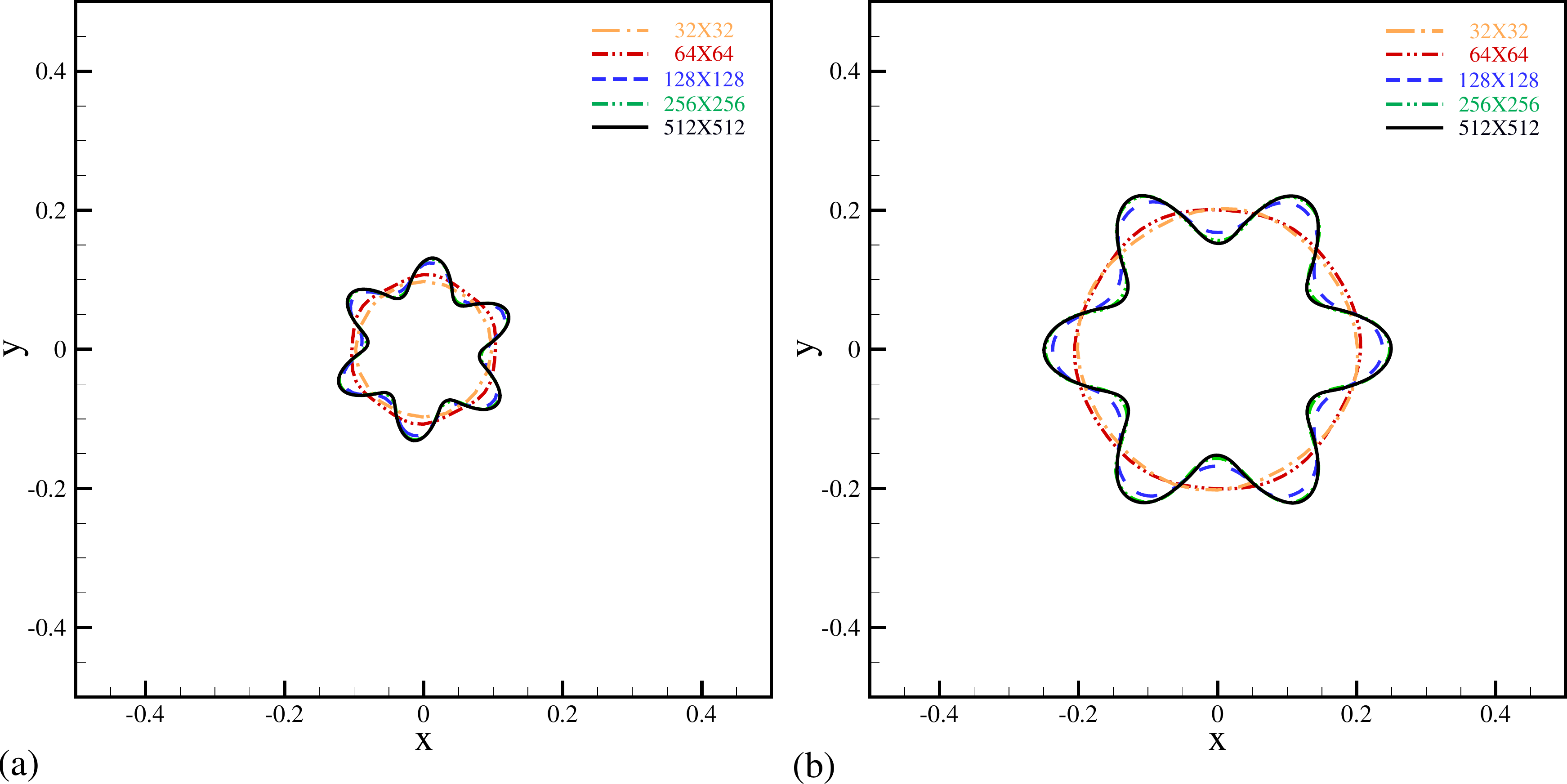}
    \caption{The computed star shaped drop after (a) half a period at $t=T/2=1$, and (b) a full period at $t=T=2$, on five different grids: $32^2$, $64^2$, $128^2$, $256^2$, and $512^2$. The interface is defined by the $\phi=0.5$ contour.}
    \label{fig:star_in_spiralling_flow_result}
\end{figure}

The domain was discretized using $N_x\times N_y$ grid points; and five different grids, $32^2$, $64^2$, $128^2$, $256^2$, and $512^2$, were chosen to study the convergence of the error in the shape and in the volume of the star. The values of $\epsilon=\Delta x$ and $\Gamma=|u|_{max}$ were used in the simulation. Figure \ref{fig:star_in_spiralling_flow_result} shows the resultant shape of the star obtained at $t=T/2$ and $t=T$ on five different grids. With an increase in the number of grid points, a clear convergence in the star shape can be seen. Since at the final time the star is supposed to return to its original shape and volume, the ``exact" final shape and the volume of the star is known; and hence the error in the ``actual" shape and the volume of the star can be computed. We compute the shape error as already defined in Eq. (\ref{eq:ns_error}), the volume error as defined in Eq. (\ref{eq:nv_error}), and the percentage change in volume as defined in Eq. (\ref{eq:pv_error}). The computed shape error, $NS_{error}$, the volume error, $V_{error}$, and the percent volume change, $\%V_{error}$, on five different grids are listed in Table \ref{tab:star_in_spiralling_flow} along with the order of convergence for the shape error, $SE_{order}$. 

The shape error decreases with an increase in the number of grid points, with an order of convergence between $1$ and $2$ for grid sizes $128^2$ and higher. The absolute values of $NS_{error}$ are also higher compared to the drop-in-a-compressible-shear-flow case, which could be due to the presence of sharp features in the star shaped drop. Interestingly, the sub-first-order convergence for the shape error could be due to the total loss of the sharp interface features on the star for the grids $32^2$ and $64^2$, as can be seen in Figure \ref{fig:star_in_spiralling_flow_result}. This shows that there exists a minimum grid size, $\Delta x_c$, that is required to resolve the sharp interface features in the flow; and the $\Delta x_c$ is clearly not met for the grids $32^2$ and $64^2$ when the star is in the fully compressed state at $t=T/2$. A zoomed-in image of the star at $t=T/2$, computed on the $64^2$ grid is shown in Figure \ref{fig:zoom_curv}(a), and on the $128^2$ grid is shown in Figure \ref{fig:zoom_curv}(b) along with the meshes; and the star computed on the $512^2$ grid is also shown as a reference. The sharp curved features of the star have around $2$ to $3$ grid points across them, on the $64^2$ grid, when the star is in the most compressed state, and hence are not resolved. However, doubling the resolution, results in a much better representation of the sharp features on the $128^2$ grid. The volume error, $V_{error}$, and the percent volume change, $\%V_{error}$, also decrease with an increase in the number of grid points, but at a much higher rate. The $V_{error}$ can be seen to have reached machine precision for the grid sizes $128^2$ and higher, and hence faster compared to the drop-in-a-compressible-shear-flow case. Hence, in brief, the proposed method requires at least 3 grid points to resolve sharp interfacial features in the flow. However, the volume conservation property is fairly insensitive to the resolution of the sharp features on the grid, which is a favorable quality of the method.

\begin{figure}
    \centering
    \includegraphics[width=\textwidth]{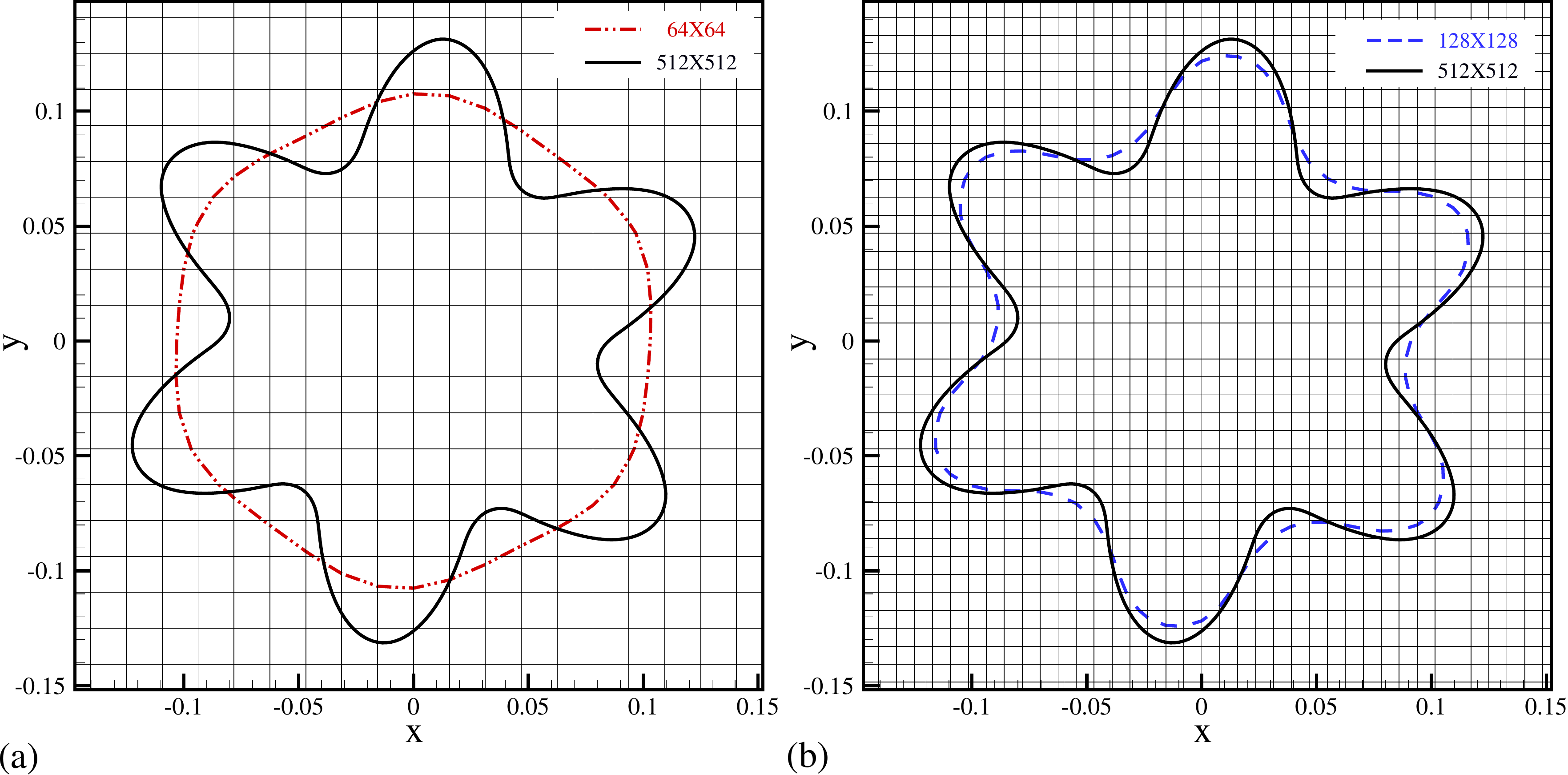}
    \caption{Zoomed-in image of the star in the most compressed state at $t=T/2=1$, computed on (a) the $64^2$ and $512^2$ grids, and on (b) the $128^2$ and $512^2$ grids. The mesh in (a) is shown for the $64^2$ grid case, and the mesh in (b) is shown for the $128^2$ grid case.}
    \label{fig:zoom_curv}
\end{figure}

\begin{table}
\centering
\begin{tabular}{@{}lllll@{}}
\toprule
Grid & $V_{error}$ & $\%V_{error}$ & $NS_{error}$ & $SE_{order}$ \\ \midrule
$32\times32$ & $-5.8102 \times10^{-5}$ & $-0.04158$ & $0.04023$ & \\
$64\times64$ & $-1.0315 \times10^{-8}$ & $-7.8079 \times10^{-6}$ & $0.03803$ & 0.5290 \\
$128\times128$ & $-2.0761 \times10^{-14}$ & $-1.5943 \times10^{-11}$ & $0.01255$ & 1.5155 \\
$256\times256$ & $-3.067 \times10^{-14}$ & $-2.3638 \times10^{-11}$ & $0.004503$ & 1.3929 \\
$512\times512$ & $7.6328 \times10^{-15}$ & $5.8881 \times10^{-12}$ & $0.001836$ & 1.2261 \\ \bottomrule
\end{tabular}
\caption{Grid convergence of the shape and volume error for the star-in-a-spiralling-flow case.}
\label{tab:star_in_spiralling_flow}
\end{table}

\subsubsection{Effect of interface thickness parameter $\epsilon$ \label{sec:parameters}}

As already discussed in Section \ref{sec:tvd}, for all the simulations presented in this work, the time-step size, $\Delta t$, was chosen to satisfy the acoustic CFL condition, and the criterion in Eq. (\ref{eq:crossover}) was used to define $\Gamma$ for a given $\epsilon$. These choices of $\Delta t$ and $\Gamma$ were sufficient to maintain the boundedness and TVD properties, since the additional criteria on $\Delta t$ in Eq. (\ref{eq:3Dboundtime}) and on $\Gamma$ in Eq. (\ref{eq:3Dtvd}) were already satisfied, and did not pose additional constraints. However, one needs to be aware that these criteria could potentially add additional restrictions on $\Delta t$ and $\Gamma$ in more severe flow conditions, such as flows that involve shocks (shock-interface and shock-turbulence interactions). 

The $\Gamma$ parameter represents an artificial regularization velocity scale; and the value of $\Gamma$ obtained from the criterion in Eq. (\ref{eq:crossover}) is such that the interface-regularization terms are the stiffest terms in the volume fraction advection equation. As a result, the interface is maintained as close as possible to the equilibrium shape at all times. 

The $\epsilon$ parameter represents an interface thickness scale; and the thickness of the interface is $\approx2\epsilon$. Therefore, as $\epsilon/\Delta x\rightarrow 0.5$, the numerical solution of the diffuse-interface method reaches the limit of a sharp-interface method, where the interface thickness is $\approx \Delta x$. However, from Eq. (\ref{eq:crossover}), to maintain the boundedness of $\phi$, this requires that $\Gamma/|u|_{max}\rightarrow\infty$, which is not practical, since $\Delta t \rightarrow 0$ as $\Gamma$ approaches $\infty$ according to Eq. (\ref{eq:3Dboundtime}). Therefore, practically, one could reduce $\epsilon$ to an extent that the increase in $\Gamma$ does not add any additional constraints on $\Delta t$ already imposed by the physical CFL limits (acoustic, convective, viscous, and thermal) in the problem. The increase in $\Gamma$ not only adds additional constraints on the time step, but could potentially lead to artificial alignment of the interface along the grid \citep{chiodi2017reformulation}, since an increase in the value of $\Gamma$ is equivalent to performing more reinitialization. Hence the choice of $\epsilon$ and $\Gamma$ is a tradeoff between accuracy and cost.

Since $\epsilon$ is an important parameter that governs the accuracy of the method, we studied the effect of $\epsilon$ on the drop-in-a-shear-flow, drop-in-a-compressible-shear-flow, and star-in-a-spiralling-flow cases. Decreasing the value of $\epsilon$ does not necessarily imply better accuracy, since it requires an increased value of $\Gamma$, which could reduce the accuracy because of stronger reinitialization. The values of $\epsilon=\Delta x$ and $\Gamma=|u|_{max}$ are optimal, and are used for all the test cases presented in this work, unless specified otherwise. However, in this Section we use the values of $\epsilon=0.75\Delta x$ and $\Gamma=2|u|_{max}$, and compare the results against the results from Sections \ref{sec:drop_in_shear_flow}-\ref{sec:star_in_spiralling_flow}. Note here that the value of $\Gamma$ is fixed by the criterion in Eq. \eqref{eq:crossover}, and hence $\epsilon$ is the only free parameter.

Table \ref{tab:sharp} lists the computed shape error, $NS_{error}$, the volume error, $V_{error}$, and the percent volume change, $\%V_{error}$, on five different grids with $\epsilon=0.75\Delta x$ and $\Gamma=2|u|_{max}$. Compared to the simulations with $\epsilon=\Delta x$ and $\Gamma=|u|_{max}$ presented in Tables \ref{tab:drop_in_shear_flow}-\ref{tab:star_in_spiralling_flow}, the sharper interface\textemdash small $\epsilon$\textemdash simulations have similar shape error, but a significantly lower volume error and percent volume change, i.e., the volume error reaches machine precision values for much coarser grids and hence, the accuracy is higher. Hence a decrease in the $\epsilon$, i.e., a sharper interface, results in better volume conservation. However in this test case, the increase in $\Gamma$ was not large enough to adversely impact the shape of the interface, or to add additional constraints on $\Delta t$.

%\begin{figure}
%    \centering
%    \includegraphics[width=0.8\textwidth]{pressure_driven_bubble_parametric_study.pdf}
%    %\caption{Caption}
%    %\label{fig:my_label}
%\end{figure}

%\subsection{Advection of a water drop in air}

\begin{table}
\centering
\begin{tabular}{@{}lllll@{}}
\toprule
Grid & $V_{error}$ & $\%V_{error}$ & $NS_{error}$ & $SE_{order}$\\ 
\midrule
\multicolumn{5}{c}{Drop in a shear flow}\\
\midrule
$32\times32$ & & & 0.06664 & \\
$64\times64$ & & & 0.02848 & 1.1700 \\
$128\times128$ & & & 0.007160 & 1.9886 \\
$256\times256$ & & & 0.002564 & 1.3961 \\ 
$512\times512$ & & & 0.0005550 & 2.3101 \\ 
\midrule
\multicolumn{5}{c}{Drop in a compressible shear flow}\\
\midrule
$32\times32$ & $-1.9812 \times10^{-4}$ & $-0.2597$ & 0.04561 & \\
$64\times64$ & $3.2340 \times10^{-7}$ & $0.0004485$ & 0.01531 & 1.4897 \\
$128\times128$ & $1.8007 \times10^{-11}$ & $2.5348 \times10^{-8}$ & 0.003738 & 2.0477 \\
$256\times256$ & $-4.5797 \times10^{-16}$ & $-6.4708 \times10^{-13}$ & 0.0009116 & 2.05021 \\ 
$512\times512$ & $-1.1061 \times10^{-14}$ & $1.5643 \times10^{-11}$ & 0.0002558 & 1.7817 \\ 
\midrule
\multicolumn{5}{c}{Star in a spiralling flow}\\
\midrule
$32\times32$ & $-4.4604 \times10^{-7}$ & $-0.0003296$ & $0.03939$ & \\
$64\times64$ & $9.0111 \times10^{-13}$ & $6.8781 \times10^{-10}$ & $0.03967$ & 0.4964 \\
$128\times128$ & $-3.0503 \times10^{-14}$ & $-2.3474 \times10^{-11}$ & $0.01317$ & 1.5064 \\
$256\times256$ & $-3.3584 \times10^{-15}$ & $-2.5898 \times10^{-12}$ & $0.004347$ & 1.5144 \\
$512\times512$ & $-1.8707 \times10^{-14}$ & $-1.4433 \times10^{-11}$ & $0.001964$ & 1.1068 \\ \bottomrule
\end{tabular}
\caption{Comparison of the shape error, the volume error, and the percent change in volume for the drop-in-a-shear-flow, drop-in-a-compressible-shear-flow, and star-in-a-spiralling-flow cases with $\epsilon=0.75\Delta$ and $\Gamma = 2|u|_{max}$.}
\label{tab:sharp}
\end{table}

%\subsubsection{stagnant drop}

%\subsection{Secondary atomization in a crossflow {\color{red}(UPDATE)}}

\subsection{Surface tension test: Oscillating drop \label{sec:surface_test}}

This section contains a standard test case that is used to assess the accuracy of the model in simulating flows dominated by surface tension effects. It was previously used by \citet{perigaud2005compressible,olsson2007conservative,ii2012interface,shukla2014nonlinear}, and \citet{garrick2017finite}. Consider a two-dimensional computational square domain of dimensions, $[-2,2]\times[-2,2]$. An initially ellipse-shaped drop is placed at $(0,0)$, and is at rest. The shape of the drop is given by the equation
\begin{equation}
    \frac{x^2}{1.25^2} + \frac{y^2}{0.8^2} = 1.
\end{equation}
Since the equilibrium shape of the drop is a circle, surface tension forces deform the drop towards its equilibrium shape. The balance of inertia and the surface tension forces results in an oscillating drop that eventually goes to rest when all the energy\textemdash kinetic and surface tension\textemdash is lost due to viscous dissipation.

The properties of the fluid in the drop are $\rho_l=1000$, $\mu_l = 8.9\times10^{-4}$, $\pi_l = 6000$, and $\gamma_l = 4.4$; and for the surrounding fluid are $\rho_g = 1$, $\mu_g=1.81\times10^{-5}$, $\pi_g = 0$, and $\gamma_g = 1.4$. The surface tension coefficient for the interface between the fluids is $\sigma = 1$. The domain was discretized using $N_x\times N_y$ grid points and three different grids, $100^2$, $200^2$, and $400^2$, were chosen to study the convergence following \citet{garrick2017finite}. The values of $\epsilon=\Delta x$ and $\Gamma=\left\vert u\right \vert_{max}$ were used in the simulation. The total time of integration was $T_{tot} = 120$.

Figure \ref{fig:oscillating_drop} shows the computed global kinetic energy ($\int\rho \left\Vert\vec{u}\right\Vert_2 dV$) on three different grids along with the results from \citet{garrick2017finite}. The total energy, $E_o$\textemdash sum of kinetic energy and surface tension energy\textemdash is also shown in Figure \ref{fig:oscillating_drop} as a reference. Here, $E_o$ is a conserved quantity in the absence of viscous dissipation. However, the viscous dissipation is not zero but negligible due to the small $\mu_l$ and $\mu_g$ values. The periods of oscillation on all three grids are identical and are in good agreement with the results from \citet{garrick2017finite}. 

Figure \ref{fig:oscillating_drop} also shows that with the present numerical scheme, the global kinetic energy does not spuriously decay, and is fairly constant throughout the simulation on all grids, indicating the non-dissipative nature of the scheme. Small differences in the global kinetic energy at later times in the simulation ($t>100$) could be due to: (a) the physical viscous dissipation; (b) the non-conservative surface tension model; and (c) the spurious currents. However, the combined effects of these are still quite small, and the results can be considered grid independent. In contrast, the calculations of \citet{garrick2017finite} clearly show significant energy decay, presumably owing to the numerical dissipation in their method. 
%Note that the momentum conserving surface tension models exist in the literature \citep{saurel2018diffuse,perigaud2005compressible}, however a model that discretely conserves kinetic energy plus the surface tension energy is yet to be developed.

\begin{figure}
    \centering
    \includegraphics[width=0.8\textwidth]{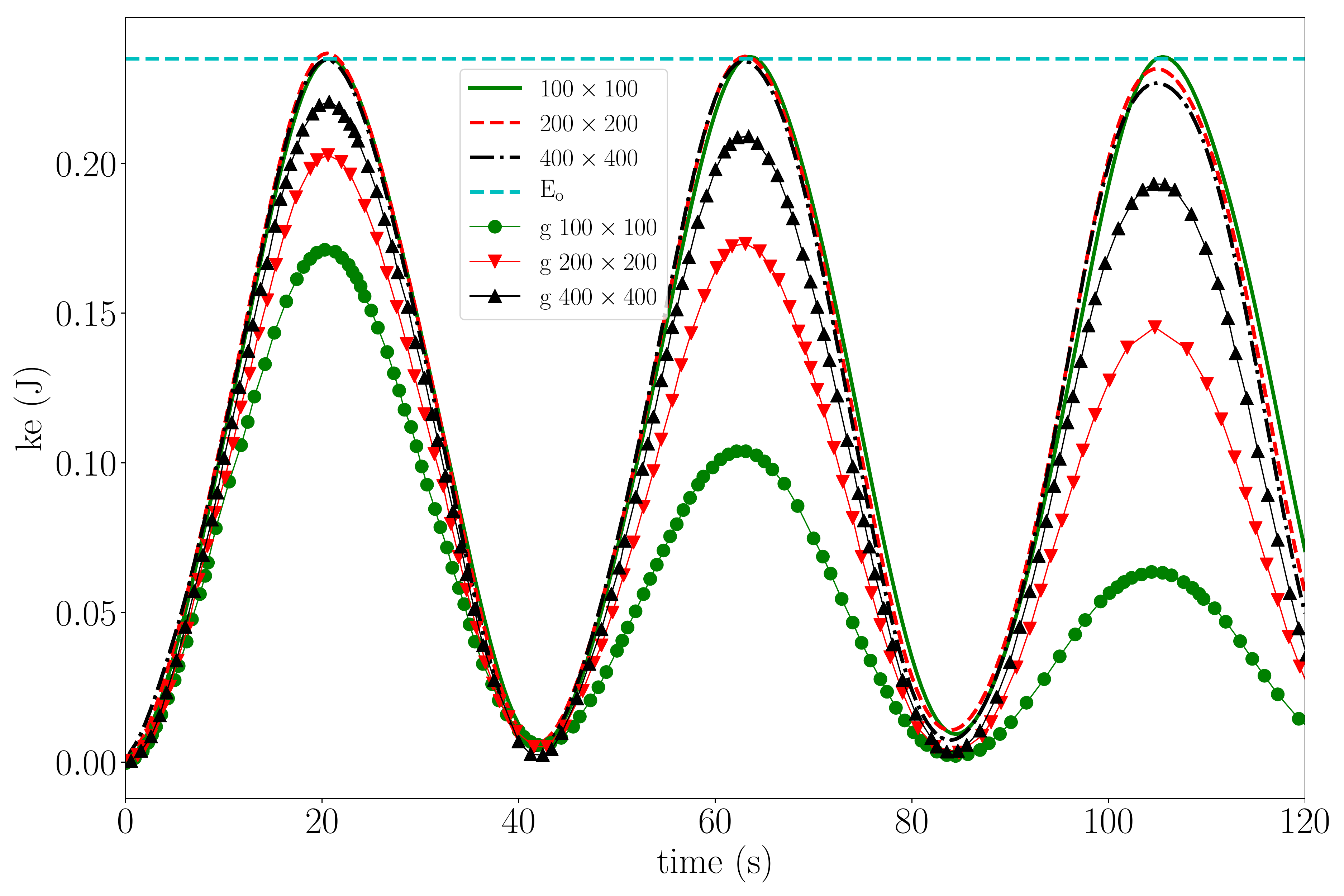}
    \caption{The global kinetic energy on three different grids $100^2$, $200^2$ and $400^2$ along with the results from \citet{garrick2017finite} denoted as ``$g$" in the legend. The dashed cyan line represents the total (kinetic + surface tension) energy $E_o$ (for an inviscid problem).}
    \label{fig:oscillating_drop}
\end{figure}

%\subsection{Compressible Rayleigh-Taylor simulation {\color{red}(UPDATE)}}
%
%
%\begin{figure}
%    \centering
%    \includegraphics[width=0.49\textwidth]{initial_RT.png}
%    \includegraphics[width=0.49\textwidth]{final_RT.png}
%    \caption{Caption}
%    \label{fig:my_label}
%\end{figure}
%
%\begin{figure}
%    \centering
%    \includegraphics[width=0.8\textwidth]{final_vel_RT.png}
%    \caption{Caption}
%    \label{fig:my_label}
%\end{figure}

\subsection{Acoustic test cases \label{sec:acoustic_test}}

In this section, numerical tests are presented to assess the accuracy of the proposed diffuse-interface method for the simulation of propagation of acoustics and its interaction with material interfaces. The two test cases presented in this section are: (a) pressure-driven bubble oscillation, that is used to evaluate the accuracy of the method in handling acoustic-bubble interactions (a similar test case was presented in \citet{huber2015time} for an axisymmetric setup); and (b) the interaction of a plane acoustic wave with a flat interface, that is used to evaluate the accuracy of the method in capturing the reflected and transmitted acoustic wave amplitudes across material interfaces, and their direction of propagation. 

\subsubsection{Pressure-driven bubble oscillation}
%\subsection{Rayleigh-Plesset equations}

%Rayleigh-Plesset-Prosperetti (3D - Bounded domain)
%\begin{equation*}
%\frac{P_B(t) - P_S(t)}{\rho_L} = \Big(\frac{S-R}{S}\Big)\Big\{R\ddot{R} + (\dot{R})^2\Big\} %+ \Big(\frac{R^4-S^4}{2S^4}\Big)+ \frac{4\nu_L\dot{R}}{R} + \frac{2\sigma}{\rho_L R}
%\end{equation*}

For the test case of pressure-driven bubble oscillation, we compare the results against the analytical solution of the Rayleigh-Plesset equation. In three dimensions, the Rayleigh-Plesset equation can be written as \citep{brennen2013cavitation}

\begin{equation}
\frac{P_B(t) - P_{\infty}(t)}{\rho} = R\ddot{R} + \frac{3}{2}\left(\dot{R}\right)^2 + \frac{4\nu\dot{R}}{R} + \frac{2\sigma}{\rho R},
\end{equation}
where $P_B(t)$ is the uniform pressure inside the bubble, $P_\infty(t)$ is the liquid pressure at infinity, $R(t)$ is the radius of the bubble, $\rho$ is the liquid density, $\nu$ is the liquid kinematic viscosity, $\sigma$ is the surface tension, which is taken to be zero in this work, and each dot represents the operation $\mathrm{d}/\mathrm{d}t$. A two-dimensional Rayleigh-Plesset equation does not exist and cannot be derived due to the presence of a logarithmic singularity at infinity. However, a finite-domain analytical solution can still be derived and can be used to verify the numerical solution. Hence, we derive a two-dimensional equivalent of the Rayleigh-Plesset equation for finite-size domains (see Appendix C) as 
\begin{equation}
\frac{P_B(t) - P_S(t)}{\rho} = \ln\left(\frac{S}{R}\right) \left\{\left(\dot{R}\right)^2 + R\ddot{R} \right\} + \left(\frac{R^2 - S^2}{2S^2} \right)\left(\dot{R}\right)^2 + \frac{2\nu \dot{R}}{R} + \frac{\sigma}{\rho R}.
\label{equ:rayleigh_plesset_2D}
\end{equation}
where $P_R$, and $P_S$, are the liquid pressures, on the surface of the bubble at $r=R(t)$, and at a finite distance from the center of the bubble at $r=S$, respectively.

In this test case, an air bubble of diameter $4\ \upmu \mathrm{m}$ is placed at the center of a square domain of size $10\ \upmu\mathrm{m}\times 10\ \upmu\mathrm{m}$ (with coordinates $[-5,5]\ \upmu\mathrm{m}\times[-5,5]\ \upmu\mathrm{m}$), as shown in Figure \ref{fig:schematic_bubble_oscillation}. On all four sides of the domain, a Dirichlet boundary condition of the form $10^5\{1 + 0.1\sin(10\omega_ct)\}$ for the pressure and a Neumann boundary condition for the velocity are imposed, where $\omega_c=10208967.75\ \mathrm{s^{-1}}$ is the characteristic resonance frequency of the bubble \citep{minnaert1933xvi}. The $\phi$ field is initialized with an analytical hyperbolic-tangent function given by $1 - 0.5\Big[1+\tanh\big\{{(\sqrt{x^2 + y^2}-r)}/{(2\epsilon_0)}\big\}\Big]$, where $r$ is the radius of the bubble.

\begin{figure}
    \centering
    \includegraphics[width=0.4\textwidth]{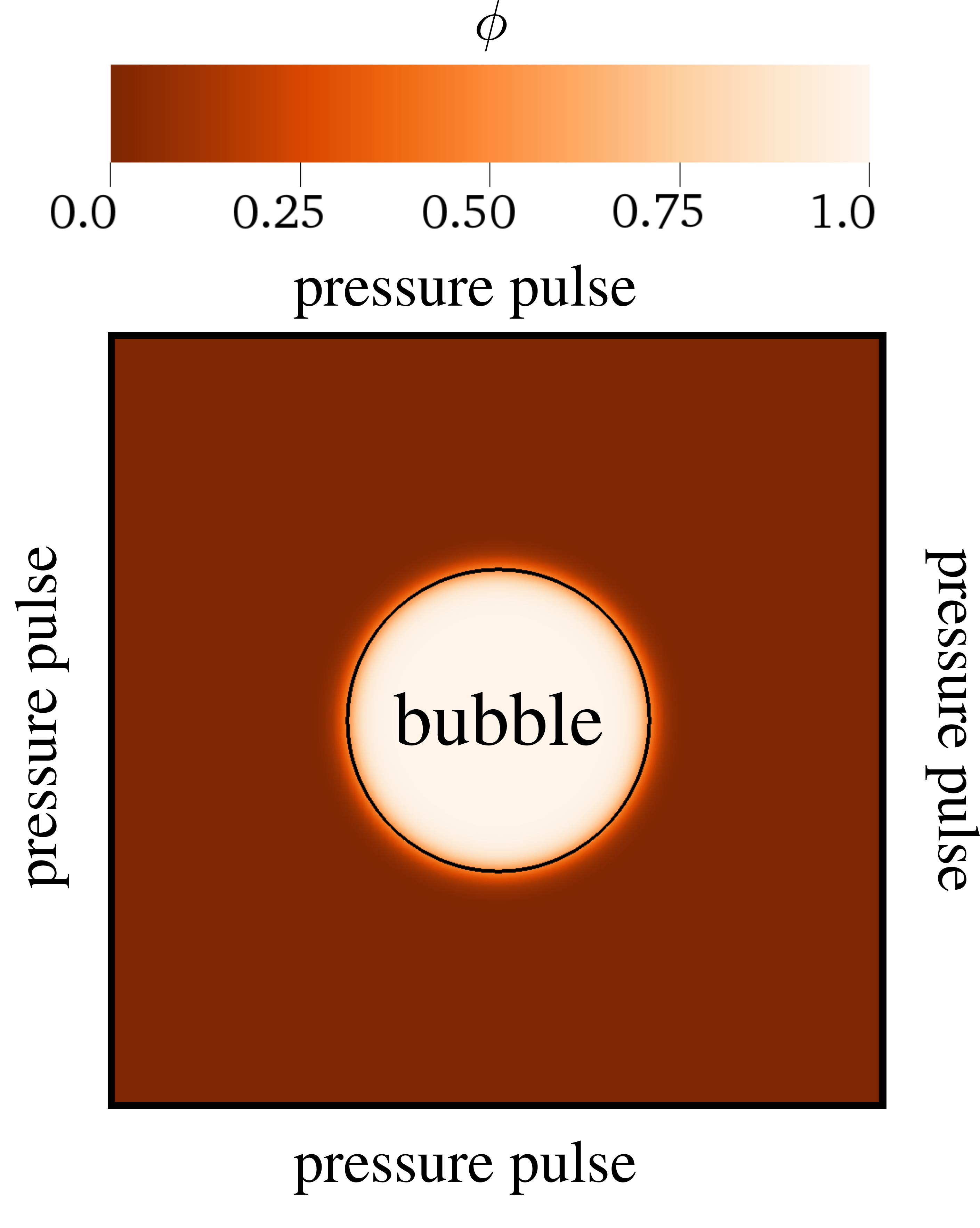}
    \caption{Schematic of the domain used in the case of pressure-driven bubble oscillation.}
    \label{fig:schematic_bubble_oscillation}
\end{figure}

The solution was numerically integrated for a total of $25\ \upmu \mathrm{s}$. Three different grids were chosen, $100^2$, $200^2$, and $400^2$, to study the convergence of the solution. The values of $\epsilon=0.55\Delta x$ and $\Gamma=10|u|_{max}$ were used in the simulation; and the time steps were chosen based on the acoustic CFL condition for the particular grid size. Here, a smaller $\epsilon$ was used to achieve better volumetric conservation as described in Section \ref{sec:parameters}, without having to decrease the time step size due to an increase in the value of $\Gamma$. Results from the various grid sizes are shown in Figure \ref{fig:bubble_oscillation}, and are compared with the semi-analytical solution obtained from numerically integrating the ordinary differential equation in Eq. (\ref{equ:rayleigh_plesset_2D}) along with the ideal-gas law, where the bubble area is computed as $\int \phi dV$ in the simulations. Figure \ref{fig:bubble_oscillation} shows the bubble response at initial times ($0\ \upmu \mathrm{s}$ to $0.6\ \upmu \mathrm{s}$) and at later times ($24\ \upmu \mathrm{s}$ to $24.5\ \upmu \mathrm{s}$). The initial transient response of the bubble in Figure \ref{fig:bubble_oscillation}(a) shows a clear convergence of the numerical solution to the analytical solution with an increase in the number of grid points. Moreover, the solution is very accurate even on the coarsest grid for the bubble response at later times [Figure \ref{fig:bubble_oscillation}(b)]. This test case also shows that the numerical solution is stable for long-time integrations.

\begin{figure}
    \centering
    \includegraphics[width=\textwidth]{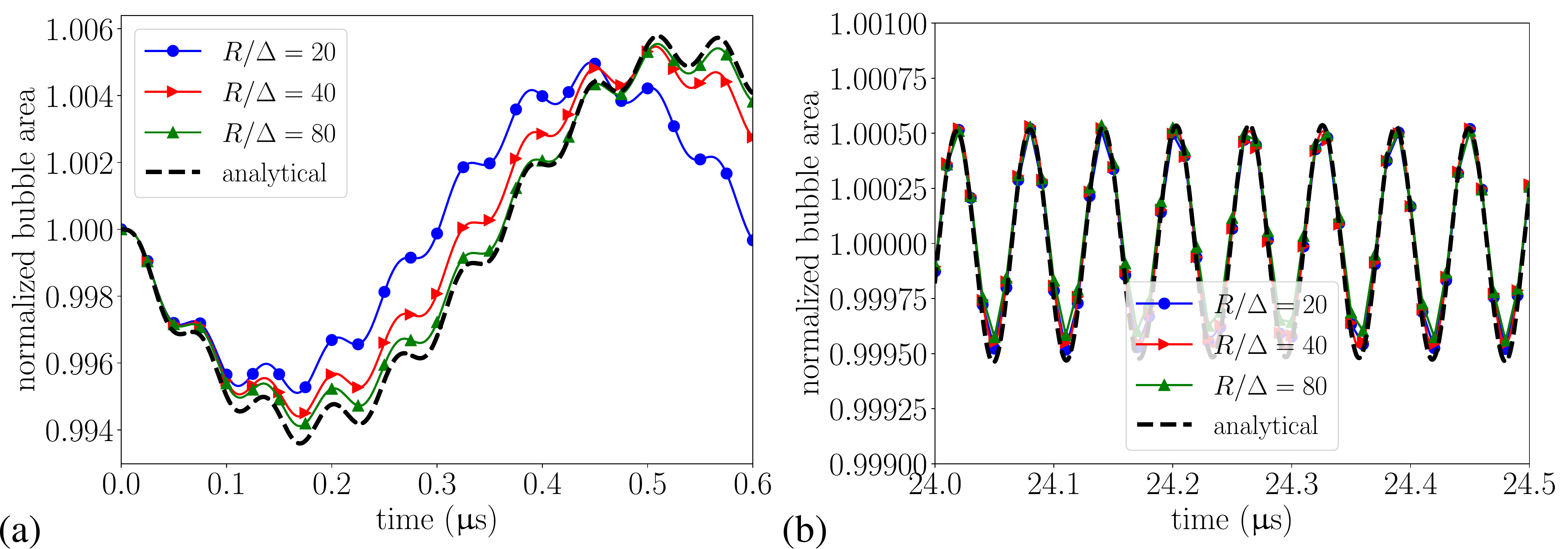}
    \caption{Bubble response at (a) initial and (b) later times.}
    \label{fig:bubble_oscillation}
\end{figure}

%%%%%%%%%%%%%%%%%%%%%%%%%%%%%%%%%%%%%%%%%%%%%%%%%%%%
%{\color{red}TALK ABOUT KAPILA'S CORRECTION.}
%%%%%%%%%%%%%%%%%%%%%%%%%%%%%%%%%%%%%%%%%%%%%%%%%%%%

%%%%%%%%%%%%%%%%%%%%%%%%%%
%Check the order of convergence for the final paper
%%%%%%%%%%%%%%%%%%%%%%%%%%

%\subsection{Propagation of a plane acoustic wave in water}
%
%
%
%\textbf{3D - wave in water - show line plots of pressure at various times.}
%
%\begin{figure}
%\begin{subfigure}{.25\textwidth}
%    \centering
%  \includegraphics[width=\linewidth]{wave_in_water_schematic.pdf}
%  \caption{Schematic}
%  \label{fig:schematic_wave_in_water}
%\end{subfigure}%
%\begin{subfigure}{.75\textwidth}
%  \centering
%  \includegraphics[width=\linewidth]{wave_in_water_pressure_line.pdf}
%  \caption{Line plot of pressure along $x$ at various times}
%  \label{fig:pressure_line_wave_in_water}
%\end{subfigure}
%    \caption{Propagation of a plane acoustic wave in water}
%    \label{fig:schematic_wave_in_water}
%\end{figure}

%%%%%%%%%%%%%%%%%%%%%%%%%%%%%%%%%%%%
%Compare attenuation rate with experiments
%%%%%%%%%%%%%%%%%%%%%%%%%%%%%%%%%%%%%%%%%

%%%%%%%%%%%%%%%%%%%%%%%%%%%%%%%%%%%%%%%%
%Test the effect of bulk viscosity of water
%%%%%%%%%%%%%%%%%%%%%%%%%%%%%%%%%%%%%%%%

\subsubsection{Interaction of a plane acoustic wave with a flat air\textendash water interface: normal incidence}

In this test case, a long three-dimensional domain of size  $10\ \upmu \mathrm{m}\times 0.1\ \upmu \mathrm{m}\times0.1\ \upmu \mathrm{m}$ (with coordinates $[0,10]\ \upmu\mathrm{m}\times[0,0.1]\ \upmu\mathrm{m}\times[0,0.1]\ \upmu\mathrm{m}$) is used, with a flat air\textendash water interface located at $x=5\ \upmu \mathrm{m}$, as shown in Figure \ref{fig:schematic_wave_in_water_interface}. The $\phi$ field is initialized with the analytical function $1 - 0.5\Big[1+\tanh\big\{{(x-x_0)}/{(2\epsilon_0)}\big\}\Big]$, where $x_0$ is the location of the interface. The domain is filled with air for $x<5\ \upmu \mathrm{m}$ and water for $x>5\ \upmu \mathrm{m}$. Perflectly reflecting wall boundary conditions are imposed on the domain face at $x=0\ \upmu \mathrm{m}$, and periodic boundary conditions are imposed for the faces at $y=0\ \upmu \mathrm{m}$, $y=0.1\ \upmu \mathrm{m}$, $z=0\ \upmu \mathrm{m}$, and $z=0.1\ \upmu \mathrm{m}$. For the wall at $x=10\ \upmu \mathrm{m}$, a Dirichlet boundary condition of the form $10^5\{1 - 0.5\sin(\omega t)\}$ for the pressure and a Neumann boundary condition for the velocity are imposed for $t<614.5\ \mathrm{ps}$, where $\omega={2\pi c}/{\lambda}$, $\lambda=1\ \upmu \mathrm{m}$ and $c$ is the speed of sound in water. Later, it is switched to a perfectly reflecting wall boundary conditions for $t>614.5\ \mathrm{ps}$ such that a pulse (half wave) is generated at the boundary and its propagation in the domain can be monitored. 

\begin{figure}
    \centering
    \includegraphics[width=0.4\textwidth]{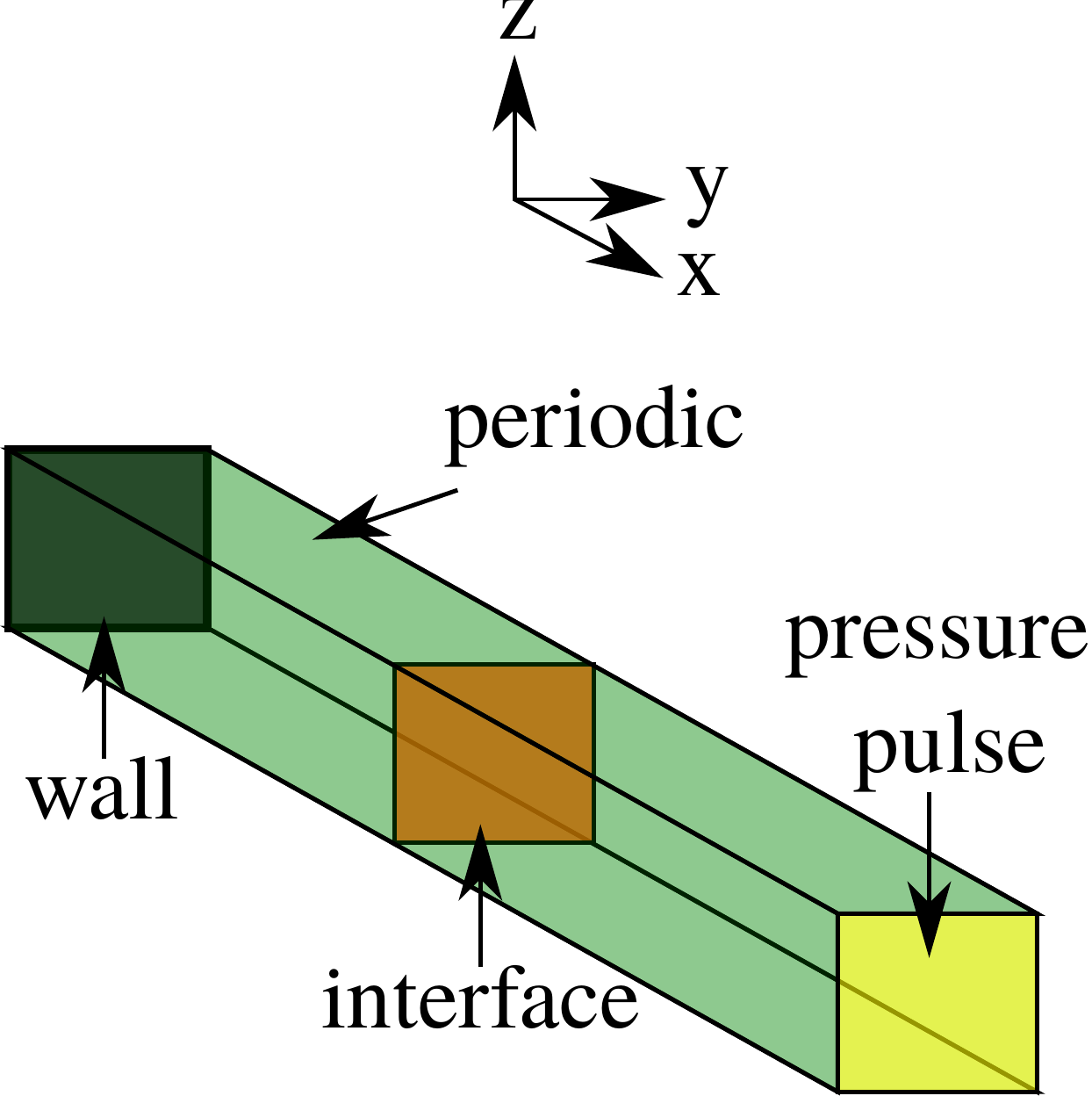}
    \caption{Schematic of the domain used in the case of interaction of a plane acoustic wave with a flat air\textendash water interface (normal incidence).}
    \label{fig:schematic_wave_in_water_interface}
\end{figure}

The solution was numerically integrated for a total of $1\ \upmu \mathrm{s}$. A grid with $1000\times10\times10$ points was used in this simulation along with the time step of $\Delta t = 1\ \mathrm{ps}$. The values of $\epsilon=\Delta x$ and $\Gamma=|u|_{max}$ were used in the simulation. Results from the simulation are shown in Figure \ref{fig:pressure_line_wave_in_water_interface}. The pressure along $x$ is plotted at various times for $y=z=0.05\ \upmu \mathrm{m}$. The acoustic wave interacts with the air\textendash water interface and reflects back, as can be seen from the results at $3\ \mathrm{ns}$, $3.5\ \mathrm{ns}$ and $4\ \mathrm{ns}$. Clearly, nothing gets transmitted across the interface, and the reflected wave amplitude is approximately equal to the incident wave amplitude but the wave is flipped. This behavior of reflection and transmission of an acoustic wave across a flat air\textendash water interface can be predicted using linear acoustic theory. The reflection coefficient is given by $\mathds{R}={(Z_a-Z_w)}/{(Z_a + Z_w)}=-0.999516$, and the transmission coefficent is given by $\mathds{T}={2Z_a}/{(Z_a + Z_w)}=4.8\times10^{-3}$, where $Z_a$ and $Z_w$ are the acoustic impedances of air and water, respectively. $\mathds{R}$ being roughly equal to $-1$ indicates that the reflected wave amplitude is the same as the incident wave amplitude and the wave is flipped. $\mathds{T}$ being roughly equal to $0$ indicates that nothing gets transmitted across the interface. Hence, the numerical solution is in good agreement with the theoretical prediction. Solutions at $6.5\ \mathrm{ns}$ and $12.5\ \mathrm{ns}$ also show the pressure-doubling behavior at the wall, which is again predicted by the theory \citep{blackstock2000fundamentals}. 

\begin{figure}
    \centering
    \includegraphics[width=\textwidth]{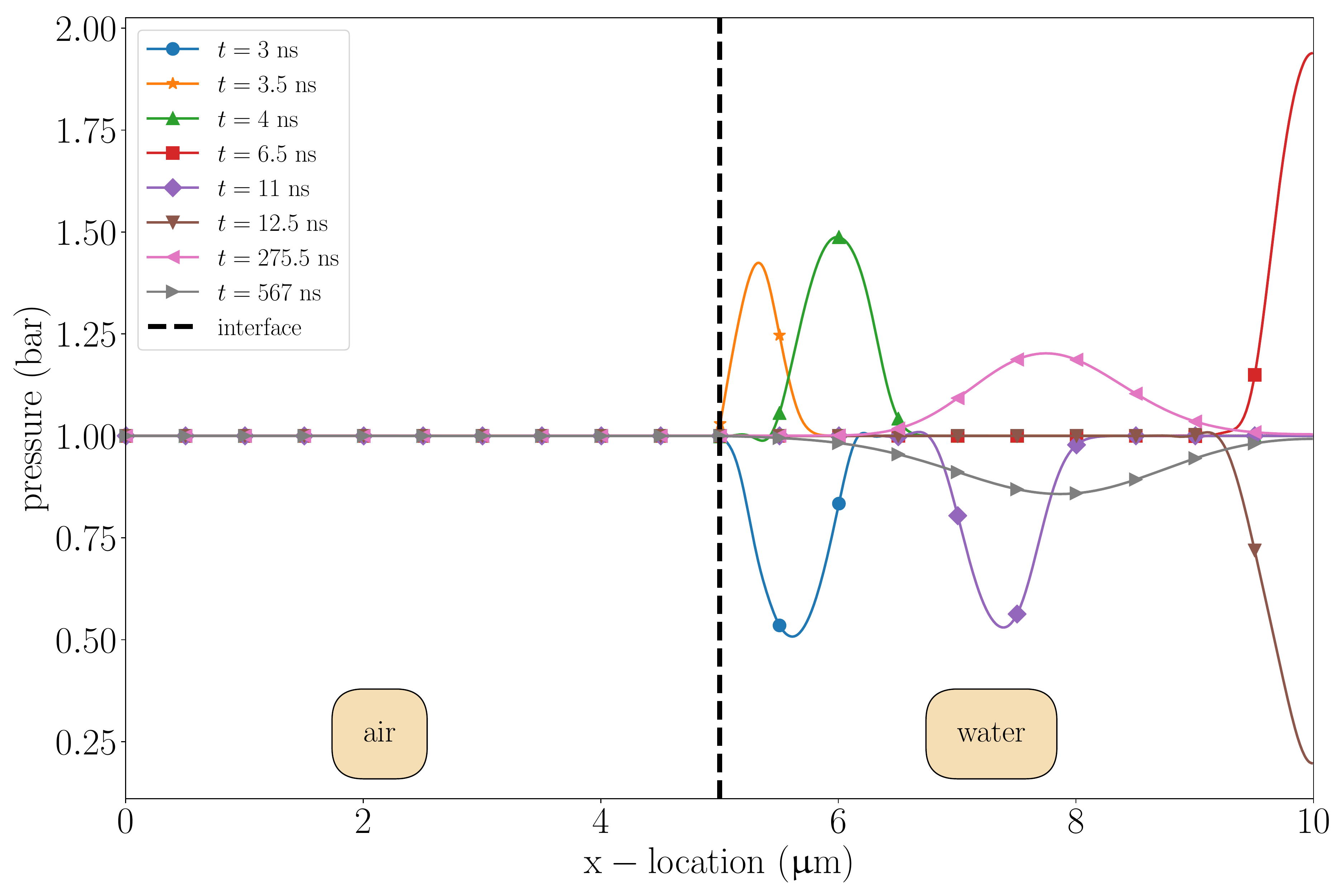}
    \caption{Line plot of pressure along $\mathrm{x}$ at various times. The dashed line represents the interface location.}
    \label{fig:pressure_line_wave_in_water_interface}
\end{figure}

%%%%%%%%%%%%%%%%%%%%%%%%%%%%%%%%%%
%Add analytical solution in the final paper.
%%%%%%%%%%%%%%%%%%%%%%%%%%%%%%%%%%

\subsubsection{Interaction of a plane acoustic wave with a flat kerosene\textendash water interface: oblique incidence}

In this test case, a two-dimensional domain of size  $10\ \upmu \mathrm{m}\times 10\ \upmu \mathrm{m}$ (with coordinates $[0,10]\ \upmu \mathrm{m}\times[0,10]\ \upmu \mathrm{m}$) is used, with a flat interface that is aligned along the principal diagonal of the domain, as shown in Figure \ref{fig:schematic_oblique_incidence}. The $\phi$ field is initialized with the analytical function $1 - 0.5\Big[1 + \tanh\big\{{(10^{-5}-x-y)}/{(2\epsilon_0)}\big\}\Big]$. The domain is filled with water below the interface (medium 2) and kerosene above the interface (medium 1). A Dirichlet boundary condition of the form $10^5\{1 - 0.5\sin(\omega t)\}$ for pressure, where $\omega = {2\pi c}/{\lambda}$, $\lambda=2\ \upmu \mathrm{m}$ and $c$ is the speed of sound in kerosene, is imposed on the right domain boundary for $t<755.297\ \mathrm{ps}$ such that a pulse is generated, and propagates into the domain along the boundary normal, so that the incident acoustic wave on the interface makes an angle $\theta_i=\ang{45}$.

\begin{figure}
    \centering
    \includegraphics[width=0.4\textwidth]{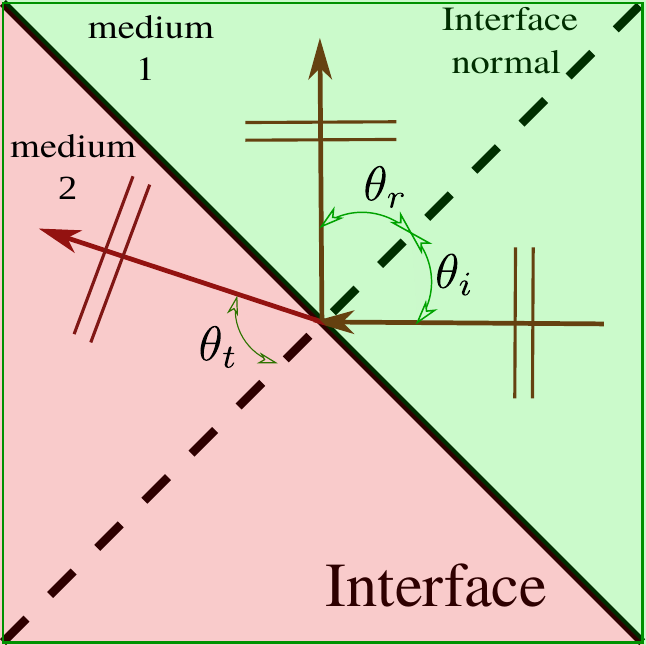}
    \caption{Schematic of the domain used in the case of interaction of a plane acoustic wave with a flat kerosene\textendash water interface (oblique incidence). The solid line is the interface; the dashed line is the interface normal; and the three arrows represent the direction of propagation of the incident, reflected, and transmitted waves that are at an angle, $\theta_i$, $\theta_r$, and $\theta_t$, respectively, with the interface normal.}
    \label{fig:schematic_oblique_incidence}
\end{figure}

A grid with $1000\times1000$ points was used in this simulation along with the time-step size of $\Delta t = 5\ \mathrm{ps}$. The values of $\epsilon=\Delta x$ and $\Gamma=|u|_{max}$ were used; and the results are shown in Figure \ref{fig:result_oblique_indicence}. The pressure field is plotted at time $t=6\ \mathrm{ns}$. The acoustic wave interaction with the water-kerosene interface results in a reflected wave and a transmitted wave. The behavior of reflection and transmission of an oblique acoustic wave across a flat interface can be predicted using linear acoustic theory. The angle of the transmitted wave with the interface, $\theta_t$, is given by the Snell's law of refraction
\begin{equation}
    \frac{\sin{(\theta_i)}}{c_1}  = \frac{\sin{(\theta_t)}}{c_2}
    \label{equ:snell}
\end{equation}
where $c_1$ and $c_2$ are the speeds of sound in medium 1 and 2, respectively, and the angle of reflection, $\theta_r$, is same as the angle of incidence, $\theta_i$ \citep{pierce1990acoustics}. In this problem $c_1=1324\ \mathrm{m/s}$ and $c_2=1627.4\ \mathrm{m/s}$ and since $c_2/c_1 > 1$, $\theta_t$ only exists if $|(c_2/c_1)\mathrm{sin}(\theta_i)|<1$. Therefore, $\exists$ a critical angle $\theta_c$
\begin{equation}
    \theta_c = \arcsin{\left(\frac{c_1}{c_2}\right)} = \ang{54.5}
    \label{equ:critical_angle}
\end{equation}
such that, $\forall$ $\theta_i>\theta_c\, \nexists\ \theta_t$, and the incident wave results in total internal reflection. Hence $\theta_i=\ang{45}$ is chosen in this problem such that there is no total internal reflection. From Eq. (\ref{equ:snell}), $\theta_t=\ang{60.358}$ and $\theta_r=\ang{45}$. In Figure \ref{fig:result_oblique_indicence}, three arrows along the incident, reflected, and transmitted waves are plotted based on the theoretical prediction of $\theta_r$ and $\theta_t$ for the given $\theta_i$. The numerical solution is thus in good agreement with the theoretical prediction.

%%%%%%%%%%%%%%%%%%%%%%%%%%%%
%Calculate the actual angles
%%%%%%%%%%%%%%%%%%%%%%%%%%%%

%\begin{figure}
%    \centering
%    \includegraphics[width=0.47\textwidth]{oblique-wave.pdf}
%    \caption{Caption}
%    \label{fig:oblique_incidence_result}
%\end{figure}

\begin{figure}
    \centering
    \includegraphics[width=0.65\textwidth]{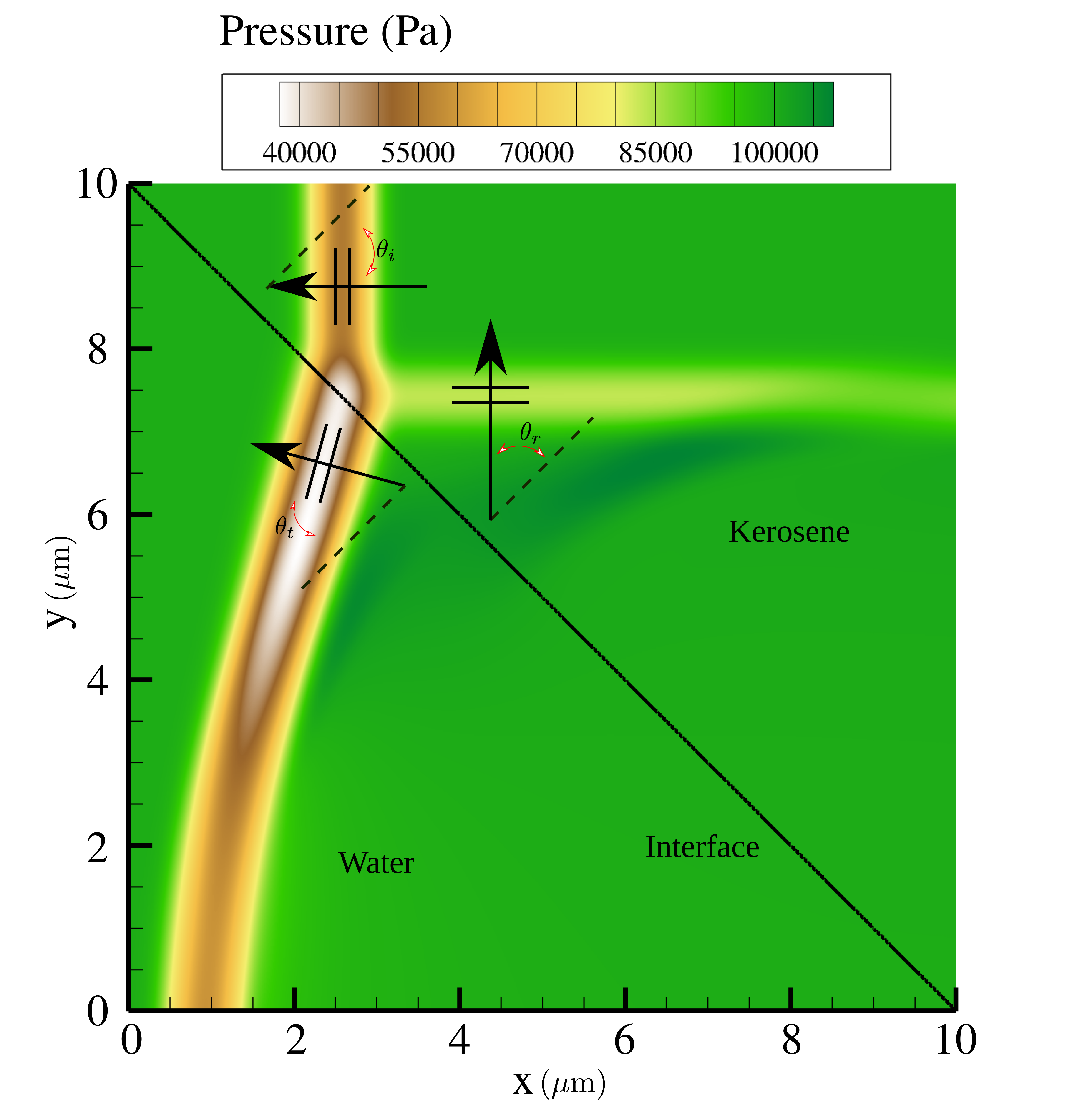}
    \caption{The pressure field (in $Pa$) at time $t=6\ \mathrm{ns}$. The solid line is the interface; the dashed lines represent the interface normal; and the three arrows represent the direction of propagation of the incident, reflected, and transmitted waves, that are at an angle, $\theta_i$, $\theta_r$, and $\theta_t$, respectively, with the interface normal.}
    \label{fig:result_oblique_indicence}
\end{figure}

%\begin{figure}
%    \centering
%    \includegraphics[width=\textwidth]{wave_in_water_interface_pressure_color.pdf}
%    \caption{Color plot of pressure along $x-z$ plane (a)$2\mathrm{ns}$ (b)$3\mathrm{ns}$, (c)$3.5\mathrm{ns}$, (d)$4\mathrm{ns}$}
%    \label{fig:pressure_color_wave_in_water_interface}
%\end{figure}

%1D, 2D, 3D - turbulence, acoustics.

%Spurious velocity - test of interface equilibrium condition - Compare with Denner et al.

%%%%%%%%%%%%%%%%%%%%%%%%%%%%%%%%%%%
%Advection of a drop 
%%%%%%%%%%%%%%%%%%%%%%%%%%%%%%%%%%%

%More coupled test cases

%%%%%%%%%%%%%%%%%%%%%%%%%%%%%%%%%%%
% Bubble rise
%%%%%%%%%%%%%%%%%%%%%%%%%%%%%%%%%%%

%%%%%%%%%%%%%%%%%%%%%%%%%%%%%%%%%%%
% Secondary atomization of the drop
%%%%%%%%%%%%%%%%%%%%%%%%%%%%%%%%%%%

\subsection{Complex flow: Three-dimensional Rayleigh-Taylor instability\label{sec:complex_test}}

In this section, we present the simulation of late-time growth of a three-dimensional single-mode Rayleigh-Taylor instability and validate the accuracy of the proposed diffuse-interface method. This simulation also helps in evaluating the robustness of the numerical scheme to simulate complex high-Reynolds-number flows. In this test case, we compare the results against a previous numerical study by \citet{liang2016lattice}, where a lattice Boltzmann multiphase model with the multiple-relaxation-time collision operator was used. We also validate our method by comparing against an experimental study by \citet{wilkinson2007experimental}.

It is known that the three-dimensional Rayleigh-Taylor instability at sufficiently high Reynolds number, undergoes four stages of development \citep{sharp1983overview}: (a) the linear growth stage, where the amplitude of the perturbation grows exponentially with time until it reaches $\approx \mathcal{O}(\lambda)$, where $\lambda$ is the wavelength of the initial perturbation; (b) the terminal velocity growth stage, where the perturbation grows non-linearly with the heavy fluid (referred to as spike) and the light fluid (referred to as bubble) penetrating into each other at a constant velocity; (c) the reacceleration stage, where the spike rolls up along the sides to form a mushroom structure due to Kelvin-Helmholtz instability \citep{Glimm2002,ramaprabhu2006limits,wilkinson2007experimental}; and (d) the chaotic development stage, where the spike breaks up resulting into a turbulent and chaotic mixing of the fluids \citep{ramaprabhu2012late}. For the two-dimensional Rayleigh-Taylor instability and the growth stages, see \citet{wei2012late}. See \citet{zhou2017rayleigh} for an extensive review and recent developments on the single and multi-mode Rayleigh-Taylor instability induced flow and turbulence. 

Following \citet{liang2016lattice}, we use a three-dimensional computational domain of size $12\lambda\times\lambda\times\lambda$, where $\lambda=1$ (with dimensions $[-6,6]\times[-0.5,0.5]\times[-0.5,0.5]$). Initially, a square-mode perturbation
\begin{equation}
    h(y,z) = 0.05\lambda\left\{\cos{(ky)+\cos(kz)}\right\}
\end{equation}
is imposed at the midplane ($x=0$), where $k=(2\pi)/\lambda$ is the wavenumber. The $\phi$ field is initialized with the analytical function $1 - 0.5\Big[1 + \tanh\big\{{(h(y,z)-x)}/{(2\epsilon_0)}\big\}\Big]$. The domain is filled with heavy fluid for $x<h(y,z)$ and lighter fluid for $x>h(y,z)$. A no-slip Dirichlet boundary condition is imposed on the walls at $x=0$ and $x=12$, and periodic boundary conditions are imposed for the faces at $y=-0.5$, $y=0.5$, $z=-0.5$, and $z=0.5$.

The properties of the heavy fluid are $\rho_l=1$, $\mu_l=10^{-3}$, $\pi_l = 30$, and $\gamma_l=1.4$; and of the light fluid are $\rho_g=0.74$, $\mu_g=0.74\times10^{-3}$, $\pi_g = 40$, and $\gamma_g=1.4$. The dynamic viscosities are chosen in such a way that the kinematic viscosity is the same for heavy and light fluids $\nu_l=\nu_g=10^{-3}$. The surface tension coefficient for the interface between the fluids is $\sigma=0$ and the imposed gravitational force is $\vec{g}=1\hat{x}$. The relevant non-dimensional numbers in this problem are Atwood number,
\begin{equation}
    A_t=\frac{(\rho_l - \rho_g)}{(\rho_l + \rho_g)}\approx0.15,
\end{equation}
and Reynolds number,
\begin{equation}
    Re=\frac{\lambda U}{\nu}=1000,
\end{equation}
where $U=\sqrt{g \lambda}=1$ is the velocity scale in the problem; and the characteristic time scale can be defined as
\begin{equation}
    \tau = \frac{1}{\sqrt{A_tgk}}\approx1.
\end{equation}

The solution was numerically integrated for a total of $15\ \tau$. A grid of size $1200\times100\times100$ ($\Delta=\lambda/100$) was used in this simulation along with the time-step size of $\Delta t=0.001$. Two other grid sizes $\Delta=\lambda/50$ and $\Delta=\lambda/25$ were also used to study the convergence of the solution. The values of $\epsilon=\Delta x$ and $\Gamma=|u|_{max}$ were used in the simulation. Results from the simulation are shown in Figure \ref{fig:RT_volume} at various times ($t/\tau$). The spike and the bubble penetrate into each other as time evolves. The roll\textendash up of the spike due to the Kelvin-Helmholtz instability and the formation of mushroom-like structure can be seen at $t=5.8\ \tau$. The roll\textendash ups further shrink until $t=9.7\ \tau$ as the spike penetrates, eventually leading to a more chaotic behavior at $t=13.6\ \tau$ and the formation of small drops. Symmetry is maintained in the simulation at all times, which is a sign of a good numerical method. A similar observation was made by \citet{liang2016lattice} and \citet{wei2012late}; however, \citet{ramaprabhu2012late} reported the break of symmetry at late times in their simulation.

\begin{figure}
    \centering
    \includegraphics[height=0.95\textheight]{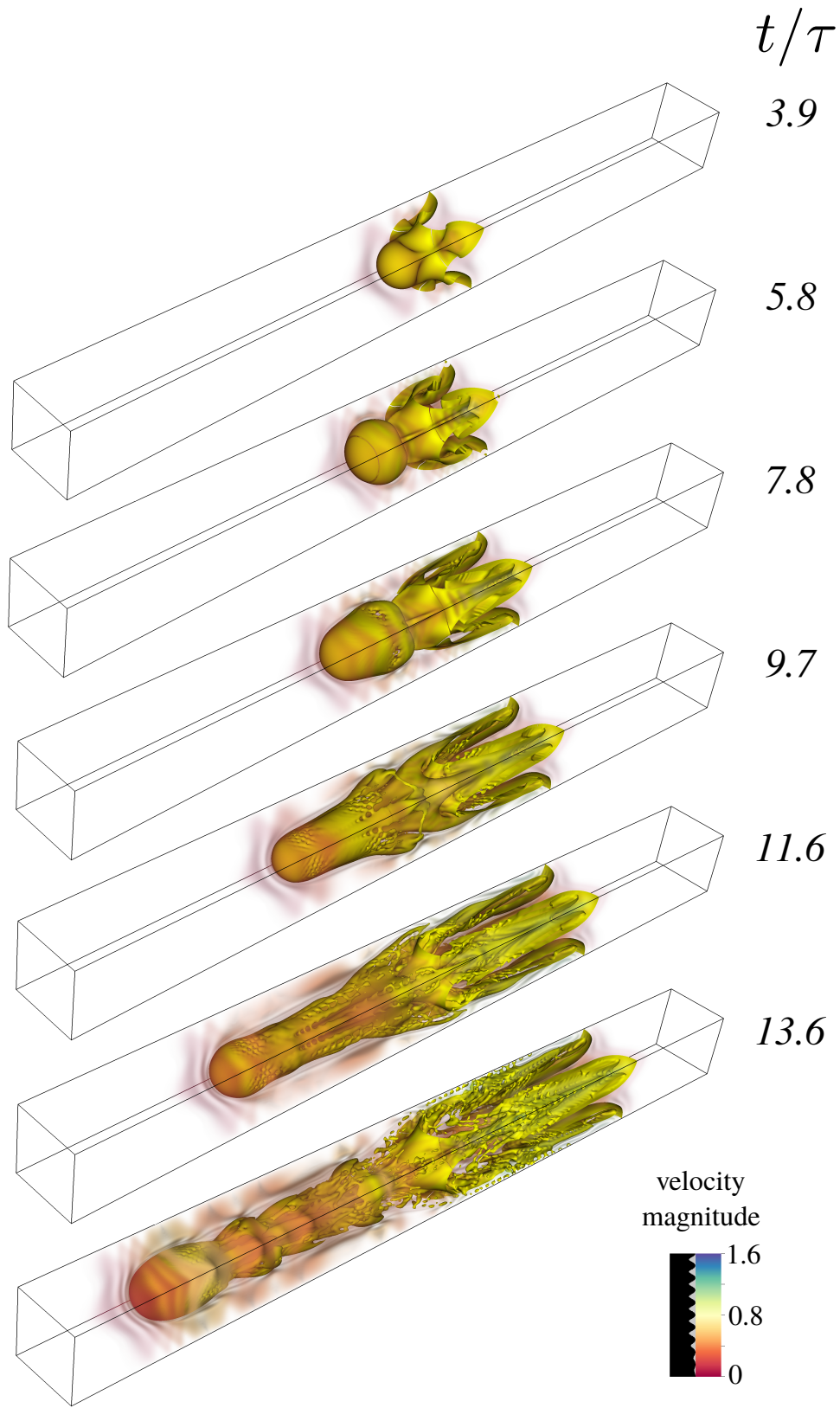}
    \caption{Interface evolution and the volumetric rendering of the velocity field in a Rayleigh-Taylor instability induced flow for $A_t=0.15$ and $Re=1000$ at various times $t/\tau$.}
    \label{fig:RT_volume}
\end{figure}

To further quantify the Rayleigh-Taylor growth at late times, we define the non-dimensional bubble and spike velocities as
\begin{equation}
    Fr_b = \frac{u_b}{\sqrt{\frac{A_t g \lambda}{1 + A_t}}},\ Fr_s = \frac{u_s}{\sqrt{\frac{A_t g \lambda}{1 + A_t}}},
    \label{equ:froude}
\end{equation}
where $Fr_b$ and $Fr_s$ are the bubble and spike Froude numbers, respectively; and $u_b$ and $u_s$ are the velocity of the bubble and spike fronts, respectively. Figure \ref{fig:RT_converge} shows the plot of bubble and spike Froude numbers as a function of time for three grid sizes, $\Delta=\lambda/25$, $\Delta=\lambda/50$, and $\Delta=\lambda/100$, along with the results by \citet{liang2016lattice}. The four distinct growth stages exhibited by the Rayleigh-Taylor instability induced flow can be clearly seen in Figure \ref{fig:RT_converge} as (a) the linear growth stage for $t\le2$, (b) the terminal velocity growth stage for $2\ge t\le6$, (c) the reacceleration stage for $6\ge t \le 10$, and (d) the chaotic mixing stage for $t\ge10$. The results from grid sizes $\Delta=\lambda/50$ and $\Delta=\lambda/100$ are in a good agreement, showing the grid convergence of the results. Our results are also in fair agreement with the results of \citet{liang2016lattice} for the bubble Frounde number, however there is a small disagreement for the spike Froude number at the late-time chaotic mixing stage. Moreover, the results from \citet{liang2016lattice} exhibit an oscillatory behavior throughout all four growth stages in the simulaton, which could be a numerical artefact. 

\begin{figure}
    \centering
    \includegraphics[width=\textwidth]{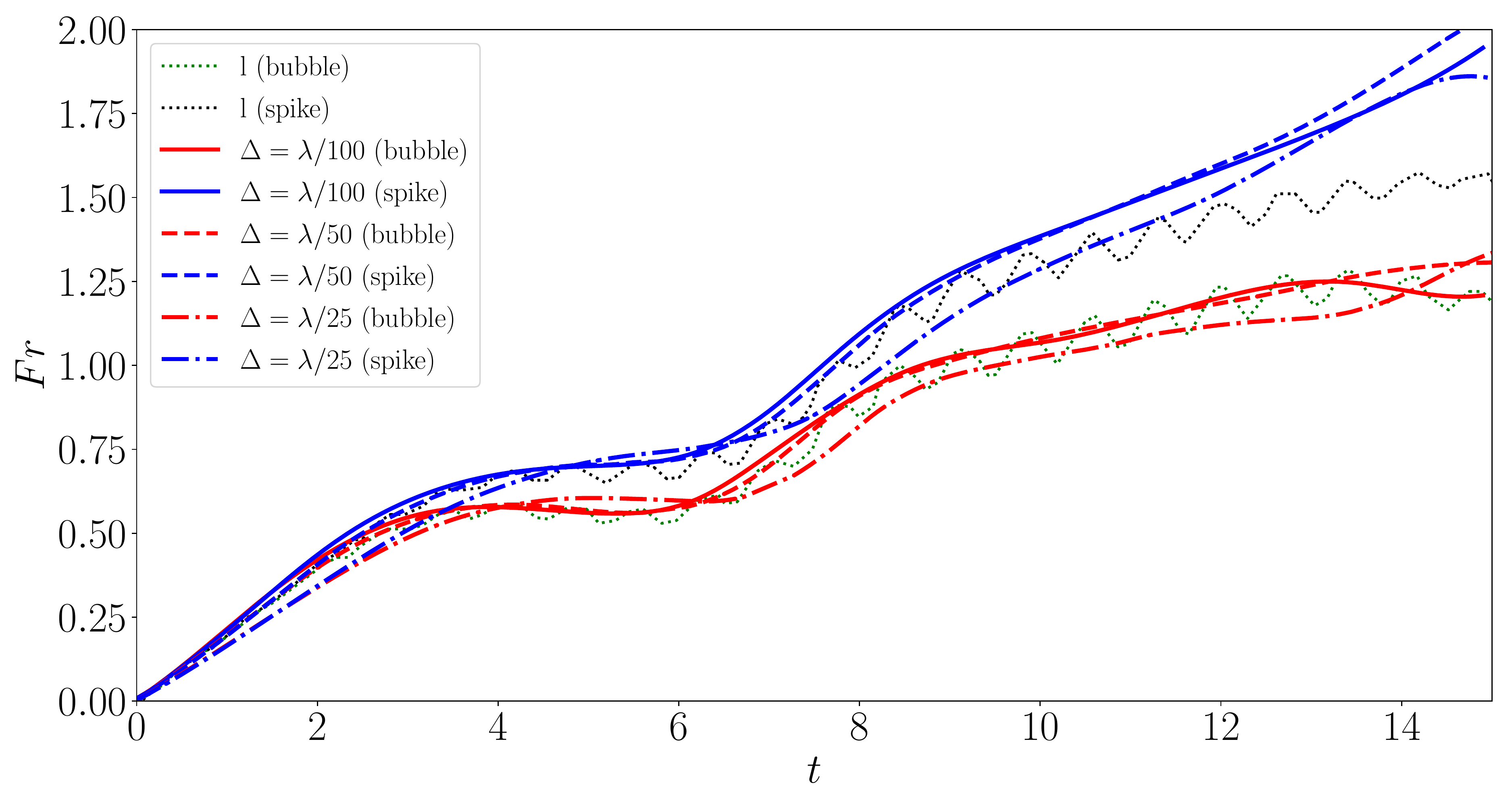}
    \caption{The bubble and spike Froude numbers as a function of time on three different grid sizes, $\Delta=\lambda/25$, $\lambda/50$, and $\lambda/100$. The dotted lines are the results from \citet{liang2016lattice}, and are denoted as $``\mathrm{l}"$ in the legend.}
    \label{fig:RT_converge}
\end{figure}

Figure \ref{fig:RT_velocity} shows the plot of $Fr_b$ and $Fr_s$ as a function of non-dimensional bubble and spike heights ($h_b/\lambda$ and $h_s/\lambda$), along with the experimental results of \citet{wilkinson2007experimental}. Note that the experimental results are limited to the first three stages of the flow and the numerical solution is in agreement with the experiments for all stages, thus validating the diffuse-interface method. The dashed lines in Figure \ref{fig:RT_velocity} show the second stage terminal velocity for bubble and spike predicted by the potential flow model of \citet{goncharov2002analytical}
\begin{equation}
    u_b = \sqrt{\frac{2A_t g}{k(1 + A_t)}},\ u_s= \sqrt{\frac{2 A_t t}{k(1 - A_t)}},
\end{equation}
and expressed in terms of Froude numbers ($Fr_b=0.564$ and $Fr_s=0.656$) using Eq. (\ref{equ:froude}).

\begin{figure}
    \centering
    \includegraphics[width=\textwidth]{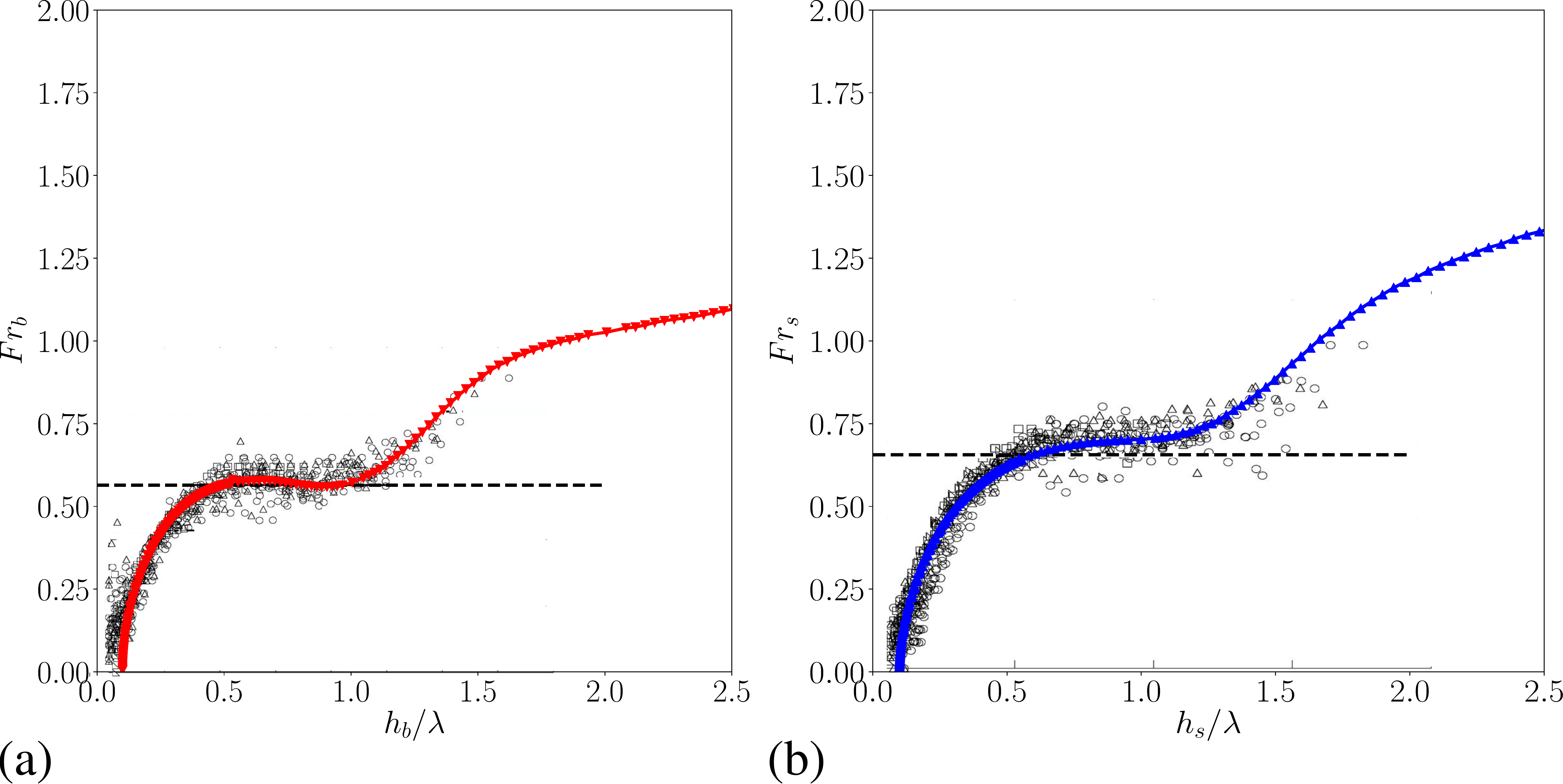}
    \caption{The Froude number as a function of non-dimensional height for (a) the bubble, and (b) the spike. The symbols are the experimental results of \citet{wilkinson2007experimental}. The dashed line represents the analytical solution of the potential flow model by \citet{goncharov2002analytical}.}
    \label{fig:RT_velocity}
\end{figure}

\section{Summary of the results and findings \label{sec:summary}}

The performance and scalability tests in Section \ref{sec:performance} show an ideal weak scaling from $1$ to $10^3$ cores, and a good strong scaling for up to $6.25\ \mathrm{K}$ grid points per core; and imply that the proposed diffuse-interface method, due to its partial-differential-equation-only nature, results in a low cost, highly scalable method. The interface advection tests in Section \ref{sec:interface_advection} show an overall order of convergence between $1$ and $2$ for the interface shape. The order of convergence was found to be dependent on the presence of sharp interfacial features, and their resolution on the grid. The volume error was shown to decrease to machine precision with an increase in the number of grid points, showing excellent volume conservation of the method. We also showed that the choice of interface parameters, $\epsilon$ and $\Gamma$, results in a trade\textendash off between accuracy and cost; and the most optimum choice would be to use $\epsilon=\Delta x$ and $\Gamma=|u|_{max}$ for most situations. The surface tension tests in Section \ref{sec:surface_test} show that the surface tension dynamics are captured accurately in the present method; and that the non-dissipative nature of the numerical scheme results in highly accurate solutions even for very coarse grids. The acoustic test cases in Section \ref{sec:acoustic_test} show that the method is accurate in capturing bubble-acoustic interactions and is stable for long-time numerical integration. The results also show that the numerical simulations are in good agreement with linear acoustic theory, thus verifying the method. The complex flow simulations in Section \ref{sec:complex_test} show that the numerical scheme is robust in simulating high-Reynolds-number flows. Further, the results from the numerical simulations show good agreement with the experimental results, thus validating the method.

\section{Conclusion\label{sec:conclude}}

In the present work, we proposed a novel conservative diffuse-interface method for the simulation of compressible two-phase flows. The proposed method discretely conserves the mass of each phase, momentum and total energy of the system. The advantages of the newly proposed interface-regularization terms compared to the state-of-the-art methods are that they maintain the conservative property of the underlying baseline model; and uses central-difference schemes for all the operators in the model, which leads to a non-dissipative \textit{discrete} implementation that is crucial for the simulation of turbulent flows and acoustics.  

Furthermore, we proved that our model maintains the boundedness property of the volume fraction field, which is a physical realizability requirement for the simulation of two-phase flows. We also proved that our model inherently satisfies the total-variation-diminishing property for the transport of the volume fraction field, without having to add any flux limiters that destroy the non-dissipative nature of the scheme. We showed that the proposed interface-regularization terms in the model do not spuriously contribute to the kinetic energy of the system, and therefore do not affect the non-linear stability of the numerical simulation; and also showed that the modeling terms in the energy equation are consistent with the second law of thermodynamics. 

Finally, we presented a wide variety of numerical simulations and tests using the model to assess, evaluate, verify, and validate the newly developed diffuse-interface method for: (a) the accuracy of evolution of the interface shape; (b) the implementation of surface tension effects; (c) the propagation of acoustic waves and their interaction with material interfaces; (d) the accuracy and robustness of the numerical scheme for the simulation of complex high-Reynolds-number flows; and (e) the performance and parallel scalability. In this work, simple geometries with Cartesian grids were considered. Extensions to include curved solid boundaries with the use of unstructured grids will be explored in the future.

\section*{Appendix A: Expanded form of the model}

Rewriting the highlighted newly-proposed regularization terms in the model presented in Section \ref{sec:finalmodel}, in terms of primitive variables: the volume fraction equation can be written as
\begin{equation}
\frac{\partial \phi_1}{\partial t} + \vec{\nabla}\cdot(\vec{u}\phi_1) = \phi_1(\vec{\nabla}\cdot\vec{u})+\vec{\nabla}\cdot\left[\Gamma\left\{\epsilon\vec{\nabla}\phi_1 - \phi_1(1 - \phi_1)\vec{n}_1\right\}\right];
\end{equation}
the mass balance equations can be written as
\begin{equation}
\frac{\partial \rho_l\phi_l}{\partial t} + \vec{\nabla}\cdot(\rho_l\vec{u}\phi_l) = \vec{\nabla}\cdot\left[\rho_{0l}\Gamma\left\{\epsilon\vec{\nabla}\phi_l - \phi_l(1 - \phi_l)\vec{n}_l\right\}\right], \hspace{0.5cm} l=1,2;
\end{equation}
the momentum balance equation can be written as
\begin{equation}
\begin{aligned}
\frac{\partial \rho\vec{u}}{\partial t} + \vec{\nabla}\cdot(\rho \vec{u} \otimes \vec{u} + p \mathds{1}) = \vec{\nabla}\cdot\doubleunderline\tau + \sigma \kappa \vec{\nabla} \phi_1 + \rho \vec{g}\\
+ \vec{\nabla}\cdot\left[\sum_{l=1}^2\rho_{0l}\Gamma\left\{\epsilon\vec{\nabla}\phi_l - \phi_l(1 - \phi_l)\vec{n}_l\right\}\otimes\vec{u}\right];
\end{aligned}
\end{equation}
and the total energy equation can be written as
\begin{equation}
\begin{aligned}
\frac{\partial E}{\partial t} + \vec{\nabla}\cdot(\vec{u} E) + \vec{\nabla}\cdot(p\vec{u}) = \sigma \kappa \vec{u}\cdot\vec{\nabla} \phi_1 + \vec{\nabla}\cdot(\doubleunderline\tau\cdot\vec{u}) + \rho \vec{g}\cdot\vec{u}\\
+ \vec{\nabla}\cdot\left[\sum_{l=1}^2\rho_{0l}\Gamma\left\{\epsilon\vec{\nabla}\phi_l - \phi_l(1 - \phi_l)\vec{n}_l\right\}k\right] \\
+ \vec{\nabla}\cdot{\left[\sum_{l=1}^2 \rho_l h_l \Gamma\left\{\epsilon\vec{\nabla}\phi_l - \phi_l(1 - \phi_l)\vec{n}_l\right\}\right]}.
\end{aligned}
\end{equation}
Rewriting the general mixture EOS in terms of the parameters of the individual phase stiffened-gas EOS, we get
\begin{equation}
p = \frac{\rho e - \sum_{l=1}^2\frac{\phi_l \gamma_l \pi_l}{\gamma_l - 1}} {\Big( \sum_{l=1}^2\frac{\phi_l}{\gamma_l-1}\Big)}.
\end{equation}

\section*{Appendix B: In-house code, CTR-DIs3D, and its performance optimization}

CTR-DIs3D is an acronym for three-dimensional Center-for-Turbulence-Research-Diffuse-Interface-method solver. It is an in-house parallel code written in C++, that employs non-dissipative numerical methods for simulating compressible and incompressible two-phase flows. The parallelization has been achieved using the Message Passing Interface (MPI) library, with the capability of arbitrary Cartesian-based domain decomposition. 

The solver has been optimized by using contiguous memory allocations for the arrays to minimize cache misses. The number of communication calls has been minimized by the use of custom-defined MPI datatypes, and by aggregating multiple message data into a single contiguous data. This increases the message size, and decreases the number of communication calls, thereby achieving higher bandwidth and better parallel scalability. Additionally, the communication calls are in synchronous non-blocking mode to hide latency and communication overhead, which further increases the parallel scalability.

\section*{Appendix C: Two-dimensional Rayleigh-Plesset equation for a cylindrical bubble}

The Rayleigh-Plesset equation is derived by integrating the mass and momentum conservation equations in the liquid region around the bubble. The liquid is assumed to be incompressible, and the bubble is assumed to oscillate in only the first volumetric mode, which is axisymmetric in nature. Now, balancing the mass in the liquid region between the radius of the bubble, $R(t)$, and a distance $r$ from the center of the bubble, we can write the radial velocity at a radius, $r$, as
\begin{equation}
	u(r,t) = \frac{R(t)}{r}\frac{dR(t)}{dt}.
	\label{eq:ray_vel}
\end{equation}

Starting with the radial component of the incompressible Navier-Stokes equation in polar coordinates
\begin{equation}
	\frac{\partial u}{\partial t} + u\frac{\partial u}{\partial r} = -\frac{1}{\rho}\frac{\partial p}{\partial r} + \nu\left[\frac{1}{r}\left\{\frac{\partial }{\partial r}\left(r \frac{\partial u}{\partial r}\right)  \right\}  - \frac{u}{r^2} \right],
\end{equation}
and substituting for the velocity from Eq. (\ref{eq:ray_vel}), we obtain
\begin{equation}
	\frac{1}{r} \left\{\left(\frac{\mathrm{d} R(t)}{\mathrm{d}t}\right)^2 + R(t)\frac{\mathrm{d}^2 R(t)}{\mathrm{d}t^2} \right\} - \frac{R^2(t)}{r^3} \left(\frac{\mathrm{d} R(t)}{\mathrm{d}t}\right)^2 = -\frac{1}{\rho}\frac{\partial p}{\partial r}.
	\label{equ:rayleigh_plesset_mom1}
\end{equation}
This equation is valid in the liquid region, and hence can be integrated from the surface of the bubble, $R(t)$. If we integrate this to infinity, we encounter a logarithmic singularity unlike in the three-dimensional Rayleigh-Plesset equation. To avoid this, we integrate Eq. (\ref{equ:rayleigh_plesset_mom1}) to a finite distance, $S$, from the center of the bubble and obtain
\begin{equation}
	\frac{P_R(t) - P_S(t)}{\rho} = \ln\left\{\frac{S}{R(t)}\right\} \left\{\left(\frac{\mathrm{d} R(t)}{\mathrm{d}t}\right)^2 + R(t)\frac{\mathrm{d}^2 R(t)}{\mathrm{d}t^2} \right\} + \left(\frac{R^2(t) - S^2}{2S^2} \right)\left(\frac{\mathrm{d} R(t)}{\mathrm{d}t}\right)^2,
	\label{equ:rayleigh_plesset_mom2}
\end{equation}
where $P_R$, and $P_S$, are the liquid pressures, at the surface of the bubble $r=R$, and at $r=S$, respectively. Now, balancing the pressure, viscous, and surface tension forces at the surface of the bubble 
\begin{equation}
	0 = -P_R(t) + 2\mu\left. \frac{\partial u}{\partial r}\right\vert_{(R(t),t)} + P_B(t) - \frac{\sigma}{R(t)},
\end{equation}
where $P_B$ is the uniform pressure inside the bubble. Substituting this in Eq. (\ref{equ:rayleigh_plesset_mom2}), we obtain the two-dimensional equivalent of the Rayleigh-Plesset equation for the finite-size circular domain

\begin{equation}
\frac{P_B(t) - P_S(t)}{\rho} = \ln\left(\frac{S}{R}\right) \left\{\left(\dot{R}\right)^2 + R\ddot{R} \right\} + \left(\frac{R^2 - S^2}{2S^2} \right)\left(\dot{R}\right)^2 + \frac{2\nu \dot{R}}{R} + \frac{\sigma}{\rho R}.
%\label{equ:rayleigh_plesset_2D}
\end{equation}

\section*{Acknowledgments} 

This investigation was supported by the Office of Naval Research, Grants \#N00014-15-1-2726 and \#N00014-15-1-2523. 
S. S. Jain is also funded by a Franklin P. and Caroline M. Johnson Fellowship. The authors acknowledge the use of computational resources from the Certainty cluster awarded by the National Science Foundation to CTR, as well as the Mira supercomputer from the INCITE program. The authors would like to thank Dr. Shahab Mirjalili and Dr. Javier Urzay for the fruitful discussions, and Ronald Chan for helpful comments on the manuscript.

\bibliographystyle{model1-num-names}
%\bibliography{sample.bib}
\bibliography{diffuse_interface_full}

%% Authors are advised to submit their bibtex database files. They are
%% requested to list a bibtex style file in the manuscript if they do
%% not want to use model1-num-names.bst.

%% References without bibTeX database:

% \begin{thebibliography}{00}

%% \bibitem must have the following form:
%%   \bibitem{key}...
%%

% \bibitem{}

% \end{thebibliography}

\end{document}